\documentclass[12pt,english]{article}
\usepackage[T1]{fontenc}
\usepackage[latin9]{inputenc}
\usepackage{geometry}
\usepackage{amsthm}
\usepackage{amsmath}
\usepackage{amssymb}
\usepackage{graphicx}
\usepackage{setspace}
\usepackage{esint}
\theoremstyle{plain}
\newtheorem{thm}{\protect\theoremname}[section]
  \theoremstyle{plain}
  \newtheorem{lem}[thm]{\protect\lemmaname}
  \theoremstyle{remark}
  \newtheorem*{rem*}{\protect\remarkname}
  \theoremstyle{definition}
  \newtheorem{defn}[thm]{\protect\definitionname}
  \theoremstyle{plain}
  \newtheorem{prop}[thm]{\protect\propositionname}
  \theoremstyle{plain}
  
  \theoremstyle{remark}
  
  \theoremstyle{definition}
  \newtheorem*{example*}{\protect\examplename}

\newtheorem{mainthm}{\protect\theoremname}

\usepackage{changepage}
\usepackage{caption}
\makeatother

\usepackage{babel}
  \providecommand{\claimname}{Claim}
  \providecommand{\corollaryname}{Corollary}
  \providecommand{\definitionname}{Definition}
  \providecommand{\examplename}{Example}
  \providecommand{\lemmaname}{Lemma}
  \providecommand{\propositionname}{Proposition}
  \providecommand{\remarkname}{Remark}
\providecommand{\theoremname}{Theorem}
\usepackage{color}

\title{On the structure of Hamiltonian impact systems }

\author{M. Pnueli$^1$ and V. Rom-Kedar$^{1,2}$ \\
\normalsize $^1$ Department of Computer Science and Applied Mathematics, \\ \normalsize The Weizmann Institute of Science, Rehovot, Israel \\
\normalsize $^2$ The Estrin Family Chair of Computer Science and Applied Mathematics.}

\begin{document}

\maketitle

\begin{abstract}
Near-integrability is usually associated with smooth small perturbations of smooth integrable systems.
Tools for  analysing  dynamics in a class of  2 degrees-of-freedom Hamiltonian impact systems with underlying separable integrable structure are derived. Integrable, near-integrable and far-from integrable cases are considered. In particular, a generalization of the energy momentum bifurcation diagram, Fomenko graphs and the hierarchy of bifurcations framework to this class is constructed. The projection of Liouville leaves of the smooth integrable dynamics to the configuration space allows to extend these tools to impact surfaces  which produce  far from integrable dynamics. It is suggested that such  representations classify dynamically different regions in phase space.  For the integrable and near integrable cases these provide global information on the dynamics whereas for the far from integrable  regimes (caused by finite deformations of the impact surface), these provide information on the singular set and on the non-impact orbits.   The results are presented and demonstrated for the Duffing-center system with  impacts from a slanted wall.
\end{abstract}

\section{Introduction}

A global qualitative analysis of a dynamical system includes, as a first step,  a classification to phase space regimes where similar dynamical behavior is expected.  Here, we derive such tools for classifying the rich dynamics in a class of Hamiltonian impact systems (HIS). An HIS  can be viewed as a billiard with an additional background potential or as a Hamiltonian system confined by some billiard boundaries. The analysis of HIS is complex as it combines two non-trivial dynamical systems, both having typically mixed phase space structure, and one of them non-smooth. While both Hamiltonian systems and billiards have a rich arsenal of research tools (see e.g. \cite{Arnold2007CelestialMechanics,meyer2008introduction,kozlov1991billiards,chernov2006chaotic}), combining them generally produces systems whose global phase space structure is far too complex for straightforward use of these tools. For example, classical billiard dynamics concepts such as the billiard map or caustics are generally not well defined (aside of  special cases, such as integrable potentials in an ellipse \cite{Radnovic2015,Fedorov2001} or the behavior near special periodic motions \cite{dullin1998linear}).  Similarly, classical results regarding  the structure of smooth Hamiltonian flows are generally not applicable since the introduction of a billiard boundary makes the HIS a non-smooth dynamical system. Finally, while the notion of near-integrability is well established for billiards and for smooth Hamiltonian systems, for impact systems such notions are new and non-trivial, see \cite{pnueli2018near}.

Here we begin to extend tools used for analyzing integrable smooth systems to 2 degrees-of-freedom (d.o.f) HIS with  integrable Hamiltonians. In particular, we extend the notion of energy-momentum bifurcation diagrams (EMBD) \cite{lerman1998integrable,Arnold2007CelestialMechanics},  Fomenko graphs \cite{fomenko2004integrable} and the hierarchy of bifurcation framework \cite{ShlRK05chaos,ShlRK10} to a class of such HIS.  The EMBD and the Fomenko graphs of a 2 d.o.f  smooth integrable system  encode the changes in the energy surface foliation by the second integral level sets.
The first level in the hierarchy of bifurcation corresponds to constructing the Fomenko graphs, namely, identifying the singular level sets belonging to a given energy surface and realizing that these divide the regular level sets on the energy surface to families (e.g. of tori in the compact case). The second level corresponds to constructing the EMBD, namely, to  identifying the singular energies at which the Fomenko graphs change their structure. The third level of the hierarchy is to classify the parameter values at which the singular energy values bifurcate (e.g. change their order), see \cite{ShlRK05chaos,ShlRK10}. It turns out that with the proper choice of the momentum variable there is a close connection between the singular energy values and low order resonances \cite{litvak2004energy}. This classification allows to identify the most dynamically interesting regimes under small perturbations (such as neighborhoods of separatrices and  hyperbolic, elliptic and parabolic resonances).

In \cite{dragovic2009bifurcations}, the notion of Fomenko graphs was introduced to the case of the integrable motion in billiards with boundaries defined by confocal quadratics, and was later expanded upon in \cite{dragovic2015topological} (see also
\cite{fokicheva2012description,fokicheva2014classification}). Other recent works have presented Fomenko graphs and classification theory to integrable billiards in non-convex domains \cite{Moskvin2018} and the more general topological billiards (see e.g. \cite{fomenko2019singularities,fomenko2019topological} and references within). For the billiards, the dynamics are independent of the energy and thus for any given table there is a single  graph which represents how the dynamics depend on the constant of motion which is preserved by the integrable or quasi-integrable billiard motion.

Here, we begin to extend these notions to the impact case, thus obtaining global information on the dynamics for a class of integrable and near-integrable HIS.  Specifically, we begin by  constructing the IEMBD (Impact-EMBD) and the Impact Fomenko graphs (IFG) of a class of separable Hamiltonians\footnote{ Separable systems means hereafter  decoupled systems -  product systems of two 1 d.o.f mechanical Hamiltonians. The  more general class of separable systems defined in  \cite{marshall1988hamiltonian}  is not analyzed here.} (see  Eq. (\ref{eq:Hint})) defined in the half plane and impacting from the half plane boundary - a straight line which is perpendicular to one of the axes.
Utilizing recent analysis of HIS  that are close to such systems  \cite{pnueli2018near} shows that the IEMBD and IFG of  systems with perpendicular walls describe the behavior of nearby systems.  To extend such results to more general walls (not close to being perpendicular to one of the axes and not necessarily straight lines), we study the projection of trajectories to the configuration space,  revealing when impacts occur. We thus discuss the relation between the  IEMBD description and these projections.

Our approach is related to recent works regarding integrable HIS with  symmetry shared by the billiard domain and the polynomial potential \cite{Fedorov2001,Radnovic2015,Dragovic2014a}.   Effectively, in    \cite{Radnovic2015},  the IEMBD and the corresponding Fomenko graphs were found for the integrable HIS of a Hooke potential in an ellipse, where the smooth motion is super-integrable, hence always periodic, leading to a beautiful and non-trivial structure of the level sets   \cite{Radnovic2015}. Here we consider other symmetries, and, more importantly, cases in which the 2 d.o.f. system is integrable yet not super-integrable, so the typical motion on regular level sets is quasi periodic - periodic motion occurs only on the measure zero resonant level sets.

The paper is organized as follows. In Section \ref{sec:setup} the model setup of the investigated systems is presented and the main theorems are introduced. Section \ref{sec:integrable} contains the analysis of the Liouville-integrable systems in which the billiard wall preserves the integrability of the underlying system. The  IEMBD, the Impact Fomenko graphs and the singular energy values  diagrams are found for such systems, thus completing the hierarchy of bifurcation classifications for the Liouville-integrable case. We end this section with a discussion of the extension of these constructions to impacts with multiple straight lines  and the relation of these to quasi-integrable billiards \cite{Athreya2012,Dragovic2014,Dragovic2015,Dragovic2015a,Moskvin2018,Issi2019}. Section \ref{sec:Hill-region} explores the breaking of the integrability when the wall no longer preserves the separable structure. Introducing the relation between the level sets properties and their projections to the configuration space,   it is shown that classification of impacting and tangent initial conditions may be derived and projected into the IEMBD even for  non-perturbative cases, and these are utilized to prove the main results described in   Section \ref{sec:setup}. Section \ref{sec:discussion} is devoted to discussion.

\section{Setup and  main results}\label{sec:setup}
Consider a 2 degrees-of-freedom (d.o.f) HIS with a Hamiltonian $H$ of the form:
\begin{equation}
H=H_{int}(q_{1},p_{1},q_{2},p_{2})+\epsilon_r V_r(q_1,q_2)+b\cdot{}V_{b}(q_{}^{w}(q;\epsilon_w))\label{eq:Hgeneral}
\end{equation}
The Hamiltonian is comprised of three components: an underlying integrable structure $H_{int}$, a small, regular perturbation term $\epsilon_r V_r$, and the impact term, formally written as a billiard potential $b\cdot V_b$ which defines the allowed region of motion to be at \(q^{w}(q;\epsilon_w)\geq0\).
Below we describe our assumptions on each  of these components. \\

\noindent\textit{I. The integrable Hamiltonian:}
 $H_{int}$ is a mechanical Hamiltonian satisfying the S3BN conditions:

\begin{defn}\label{def:s3b} \emph{Hamiltonians of the S3BN (Separable, Smooth, Simple, Bounded Nondegenerate level sets) class} are integrable, mechanical Hamiltonian functions of the form
\begin{equation}
H_{int}=\frac{||p||^{2}}{2}+V_{int}(q_{1},q_{2})=\frac{p_{1}^{2}}{2}+V_{1}(q_{1})+\frac{p_{2}^{2}}{2}+V_{2}(q_{2})
=H_{1}(q_{1},p_{1})+H_{2}(q_{2},p_{2})\label{eq:Hint}
\end{equation}
satisfying the following conditions:
\begin{description}
\item{[S1]} $H_{int}$ can be written as the sum of two $1$ d.o.f mechanical Hamiltonians \(H_{1},H_2\).
\item{[S2]} $V_1,V_2$ are $C^r$ smooth,  $r>4$.
\item{[S3]} $V_1,V_2$ have only a finite number of simple  extremum points: \(q_{i,k_{i}} ^{ext},V'_{i}(q_{i,k_{i}}^{ext})=0, V'_{i}(q_{i,k_{i}}^{ext})\neq 0,\,k_{i}=1,..,n_i,i=1,2.\)
\item{[B]} $V_1,V_2$ are bounded from below and go to infinity as $|q_1|,|q_2|\rightarrow{}\infty$ respectively.
With no loss of generality  assume that $\min_{q_2}V_2=0$.
\item{[N]}  $H_{int}$ satisfies the iso-energy KAM non-degeneracy condition \cite{Arnold2013mathematicalmethods}.
\end{description}
\end{defn}

The condition [S1] is called hereafter the separability assumption (a more general notions of separability is introduced in  \cite{marshall1988hamiltonian}).  The condition [B] implies that $H_{int}$ has only bounded level sets. The  [S2]  smoothness requirement and the non-degeneracy condition [N]  are added so that, together with the construction of \cite{pnueli2018near},  under small perturbations KAM theory can be employed.

The Hamiltonian system defined by \(H_{int}\) is clearly  Liouville integrable \cite{Arnold2007CelestialMechanics, jovanovic2011completely}:
\begin{defn}\label{def:Liouvilleintegrable}
A 2 d.o.f Hamiltonian system is called \textit{Liouville integrable} (also \textit{completely integrable}) if there are 2 Poisson-commuting smooth integrals $F_1,F_2$, $\{F_1,F_2\}=0$ whose differentials are independent in an open, dense subset of the phase space.
\end{defn}

Each level set of these integrals may be composed of several connected components.  Each connected component is called a Liouville leaf. From the Liouville-Arnold theorem \cite{Arnold2007CelestialMechanics, jovanovic2011completely}, a Liouville integrable system with compact level sets has the property that the motion on regular leaves (leaves on which the differentials of the constants of motion are independent) of the system is diffeomorphic to motion on a torus (for 2 d.o.f, a 2-torus), so, on each such torus and its neighborhood transformation to local action-angle variables, \((q,p)\rightarrow(I,\theta)\) may be defined.
The topological structure of the level sets of integrable system is fully described by the Fomenko graphs  and the  EMBD  \cite{lerman1998integrable,Arnold2007CelestialMechanics,fomenko2004integrable,ShlRK10,ShlRK05chaos} (adding additional numerical marks to the graphs allows to classify the systems according to their Liouville equivalence class \cite{fomenko2004integrable}). In section \ref{sec:integrable} we describe their construction for the S3BN class.
\\
\\
\noindent\textit{II. The smooth perturbing potential:}
 The term $\epsilon_rV_r(q_1,q_2)$ is  a regular coupling term for the potential.  We examine the behavior of the  Hamiltonian system under impacts with and without this small smooth perturbation:  $0\leq\epsilon_r\ll1$ and \(V_r(q_1,q_2)\) satisfies the smoothness condition [S2]: it is    $C^r$ smooth, $r>4$. Without the impacts, by the S3BN assumption and the smoothness of \(V_r(q_1,q_2)\),   the conditions for KAM theory are fulfilled and   thus for sufficiently small \(\epsilon_r\) the majority of the tori on non-degenerate energy surfaces survive  \cite{Arnold2007CelestialMechanics}. The Fomenko graphs and the EMBD of the integrable system supply qualitative classification of the perturbed motion   \cite{litvak2004energy,ShlRK05chaos,ShlRK10,pnueli2018near}.
\\
\\
\noindent\textit{III. The impact term and the wall properties:}
The impact term \(b\cdot{}V_{b}(q_{}^{w}(q;\epsilon_w))\) corresponds to elastic collisions with the billiard boundary, the wall  $q^{w}(q;\epsilon_w)=0$ (the properties of   $q^{w}(q;\epsilon_w)$  are specified below).  Motion is allowed only when  $q\in D=\{q|q^w(q;\epsilon_w)\geqslant0\}$, where, hereafter, \(D\) is the billiard domain. At the billiard boundary, whenever the normal is defined, the particle reflects with the angle of reflection equals to the angle of incidence. When the normal is undefined at the collision point (hereafter called corner points) the motion stops. The reflection rule is formally represented by the singular potential $V_b$ :
\begin{equation}
V_b(q^{w}(q;\epsilon_w))=\begin{cases}
0, & \text{when }q^{w}(q;\epsilon_w)>0\\
1, & \text{when }q^{w}(q;\epsilon_w)<0
\end{cases}\label{eq:billiardpot}
\end{equation}
where  $b$ is a large positive constant so that the billiard boundary is impassable. The formal definition can be made rigorous by introducing soft steep potentials, see \cite{kozlov1991billiards,RK2014smooth,rom2012billiards}. The resulting dynamics correspond to the composition of the smooth flow with the impact map (see e.g. \cite{Woj98,dullin1998linear,RK2014smooth,LRK12}.). For most of this paper (with the exception of section \ref{sec:multiwall})  we consider a billiard domain which is diffeomorphic to the half plane, so that the billiard boundary is a single smooth wall with no singularities:
\begin{defn}
\label{def:gws}A \emph{general wall system} (GWS) is the impact system defined by equations (\ref{eq:Hgeneral}),  (\ref{eq:Hint}), (\ref{eq:billiardpot}) where the potential \(V\) satisfy the S3BN conditions, the potential \(V_{r}\)  is smooth  \((C^{r},r>4)\), and the wall can be represented as a smooth \((C^{r},r>4)\) graph with bounded derivatives over one of the \(q_{i}\) axes.
\end{defn}

For example,  when the motion occurs on the right half plane we set:\begin{equation}
 q^w(q_{1},q_2)={}q_1-\epsilon_w Q(q_{2}), \:q_{ 2}\in\mathbb{R}, \text{ where } \,Q(0)=0, \ \left\Vert Q'(q_{ 2})\right\Vert_{C^{r}}<M, r>4, \label{eq:wavywal}
\end{equation}
where
 $\epsilon_w$ is  arbitrary and \(M\) is a fixed finite positive bound. The zero and small \(\epsilon_w\) regimes, where the wall is vertical/nearly vertical  are discussed in section \ref{sec:integrable} whereas
 $\epsilon_w$ is of order one in section \ref{sec:Hill-region}. Another example which we will discuss are slanted walls  with a tilt angle $\alpha\in[0,\frac{\pi }{2}]$ :
\begin{equation}
q_{\alpha}^w(q_{1},q_2)=\sin{\alpha}\cdot{}q_1-\cos{\alpha}\cdot{}q_2.\label{eq:straightwall}
\end{equation}
For \(\alpha\neq0,\) the slanted wall may be represented as a graph of the form (\ref{eq:wavywal}) with \(\epsilon _{w}(\alpha) Q(q_2)=\cot \alpha \cdot q_2\), whereas for \(\alpha\approx0 \) the roles of \(q_1,q_2\) are interchanged (notice that by the above convention  \(q_{0}^w(q_{1},q_2)=-q_{2}\), so the billiard domain is the lower half plane).
A central class of GWS  which are completely analyzable are the perpendicular wall systems:

\begin{defn}
\label{def:pws}A \emph{perpendicular wall system}, PWS, is a GWS where the wall is a straight line which is perpendicular to either the   $q_1$ or   the $q_2$ axis.\end{defn}
 The PWS corresponds to the case $\epsilon_w=0$ in (\ref{eq:wavywal}) or $\alpha=0$ or $\frac{\pi}{2}$ in (\ref{eq:straightwall}).

\subsection{The main example: the Duffing-center Hamiltonian}

For concreteness, hereafter, all the figures are presented for the Duffing-Center Hamiltonian \(H^{dc}_{int}\) (see Figure \ref{fig:phasespace} for its phase space structure at \(\omega=1)\):

\begin{equation}
H^{dc}_{int}=\frac{p_{1}^{2}}{2}+\frac{p_{2}^{2}}{2}-\frac{1}{2}\cdot(q_{1}-q_{1s})^{2}+
\frac{1}{4}\cdot(q_{1}-q_{1s})^{4}+\frac{\omega^{2}}{2}\cdot(q_{2}-q_{2c})^{2}=\frac{p_{1}^{2}}{2}+\frac{p_{2}^{2}}{2}+V^{dc}_{1}(q_{1};q_{1s})+V^{dc}_{2}(q_{2};q_{2c})\label{eq:doublewell}
\end{equation}
 with either perpendicular or slanted walls. The phase space structure in $(q_1,p_1)$ corresponds to a symmetric double-well potential with a saddle point at $(q_{1s},0)$ and two symmetric center points at $(q_{1s}\pm 1,0)$. The energy $H_1=0$ corresponds to the energy at the saddle point and its respective separatrix. The phase space structure in the $(q_2,p_2)$ plane is that of a single linear center at $(q_{2c},0)$, and the Hamiltonian $H_2$ may be written in global action-angle coordinates as $H_2=\omega{}I$, where $I$ is the action \cite{Arnold2007CelestialMechanics}. The Hamiltonian (\ref{eq:doublewell}) satisfies the S3BN assumption. It has three special feature:  \(I\) is globally defined,  \(H_{2}\) is linear in \(I\) and \(H_{1,2}\) are symmetric in  \((q_{1}-q_{1s})\)  and \((q_{2}-q_{2c})\) respectively. These special features of (\ref{eq:doublewell}) simplify the presentation yet they are not assumed nor used in the general formulation and analysis.

\begin{figure}[h]
\begin{centering}
\includegraphics[scale=0.35]{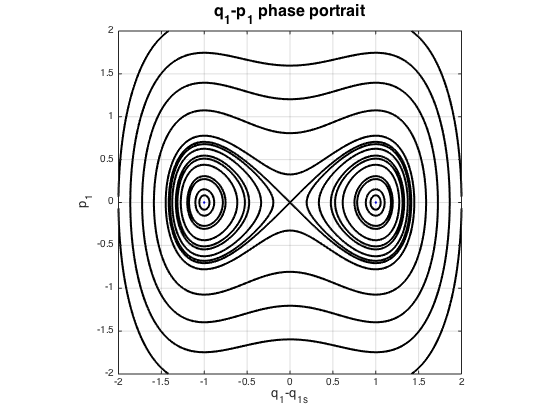}\includegraphics[scale=0.35]{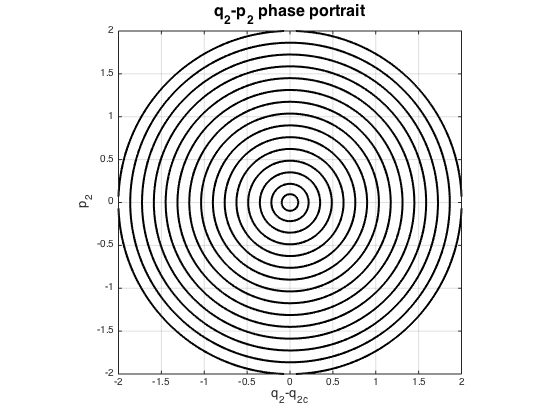}
\par\end{centering}
\protect\caption{\label{fig:phasespace}Phase portraits of the  $(q_{1},p_{1})$ (left), $(q_{2},p_{2})$ (right) dynamics for the Duffing-Center Hamiltonian \(H^{dc}_{int}\) (Eq. (\ref{eq:doublewell})).}
\end{figure}

 \subsection{Main results}

The aim of this paper is to provide tools for the global analysis of GWS.
To this aim, we distinguish between three types of \textit{invariant impact sets} of  an HIS:

\begin{defn}The \textit{impact division } is the division of phase space to  \emph {non-impact, singular and transverse impact sets. The non-impact set } is   the union of orbits of the HIS which do not impact the billiard boundary. The \emph{singular impact set} is the union of orbits which have at least one tangent impact or segments of orbits that end at a billiard corner. The \emph{transverse impact set} is the union of orbits of the HIS for which impacts occur and all the impacts are regular.   \end{defn}

 For billiards, the non-impact set is empty and the properties of the transverse and singular impact sets  depend on the billiard table. For example, for dispersing billiards,  the singular impact set is dense in phase space, and thus to prove hyperbolicity one needs to restrict   the dynamics to the transverse impact set and study carefully its properties (see \cite{chernov2006chaotic} and references therein). In \cite{Woj98} similar results are established for HIS with linear potentials and impacts from wedges.  On the other hand, for smooth convex billiards, the singular impact set is trivial - it consists of points on the boundary and their tangent direction which is out of the billiard. So here there are no interior singular  segments, and away from the whispering gallery orbits (which are confined by KAM curves), all orbits belong to the transverse impact set which is bounded away from the singular impact set.

Here  we propose some tools for studying the\textit{ impact-division} of GWS and  their dynamics.
 First, we observe that such analysis  is possible when the system with impacts  preserves the Liouville integrability and/or the typical  near-integrable dynamics;
To this aim we define:

\begin{defn}\label{def:liouvilimpint}
An HIS with compact level sets is a Liouville-Integrable HIS (LIHIS) if:
\begin{description}
\item[RespF]  The integrals of motion of the smooth Liouville-integrable Hamiltonian flow are preserved under impacts.

\item[Resp\(\theta\)]  The motion  on any component of a regular level set, namely, on components in the allowed region of motion on which the differentials of the constants of motion are independent, is conjugated to a motion on a torus.  \end{description}
\end{defn}
While for smooth systems the Resp\(\theta\) condition follows from the first one by the Arnold-Liouville theorem, for non-smooth systems this is not true. Indeed, in \cite{Issi2019} it is show that there are HIS systems satisfying the first condition and not the second one so integrability and Liouville integrability are not equivalent for HIS  (integrability for HIS  is defined in definition 1.2 of \cite{Woj98}  by the first criterion only; there,  integrability is associated  with the property of zero Lyapunov exponents for almost all initial conditions).

 After constructing the IFG and the IEMBD for  PWS in section \ref{sec:integrable}, we prove in section \ref{sec:proofofmainintegrb}  that:
\begin{mainthm}
\label{thm:integrability} \textit{At  \(\epsilon_{r}=0\)  the perpendicular wall system is a Liouville-Integrable HIS.   The IFG and the IEMBD provide  the impact-division of such systems whereas, for some cases, the FG and the EMBD of the Hamiltonian  (\ref{eq:Hint}) do not provide even the topological classification of the level set foliation of such systems.}
 \end{mainthm}

As integrability is associated with continuous symmetries
\cite{Arnold2013mathematicalmethods,meiss2007differential}, the above theorem suggests  that for  PWS the wall \emph{respects the  integrable Hamiltonian symmetry}. In section \ref{sec:multiwall} we discuss this perspective  and  possible extensions of Theorem \ref{thm:integrability} to billiard domains defined by any number of vertical and horizontal  infinite straight  lines. There, we also briefly discuss how  integrable HIS systems which  are  not Liouville-integrable arise when some of the lines are  semi-infinite. While the IEMBD for such cases can be similarly constructed, the construction of IFG  for such cases is beyond the scope of this paper (as quasi-integrable dynamics arises \cite{Issi2019}, similar to the dynamics  in billiards with \(270^{o}\) corners \cite{Athreya2012,Dragovic2014,Dragovic2015,Dragovic2015a,Moskvin2018}).

In analogy with the smooth case, we define near-integrability of an HIS as the property of preservation of most of the Liuoville tori on energy surfaces:

\begin{defn}\label{def:nihis}
A family of HIS depending on a parameter \(\epsilon\) is a Near-Integrable HIS  if for \(\epsilon=0\) it is a LIHIS and, on all energy surfaces with, possibly, the exception of a finite number of intervals in the neighborhood of  singular energies, the set of initial conditions for which the motion is conjugated to a motion on a torus approaches the full measure as \(\epsilon\rightarrow0\).
\end{defn}

Applying previous results on KAM theory in impact systems  \cite{pnueli2018near} together with Theorem \ref{thm:integrability} we establish in section \ref{sec:proofofmainintegrb} that the integrable results predict the global  behavior of nearby GWS:  \begin{mainthm} \label{thm:nearintegrability} \textit{For sufficiently small \(\epsilon=|\epsilon_{r}|+|\epsilon_{w}|\) the near perpendicular wall system is a near integrable HIS. The IFG and the IEMBD of the integrable limit (\(\epsilon=0\)) provide, up to sets of small measure,     the impact division. In particular, the division includes a small phase space region where the non-impact, singular and transverse-impact sets are mixed whereas the rest of the phase space is divided to a finite number of open regions in which only one impact type of motion exists.}     \end{mainthm}

The size and the dynamics in the near-tangent region in this near-integrable setup is studied in \cite{Pnueli2020}.
In section \ref{sec:Hill-region} the impact division in the non-perturbative regime is examined (the GWS with arbitrary \(\epsilon_{w}\)). The properties  of the \(\epsilon_{w}-\) dependent IFG and IEMBD for a GWS  are summarized in Proposition \ref{thm:global} and proved in sections \ref{sec:hillfoliation}-\ref{subsec:impactzones}. Briefly, the IFG divides the leaves of an energy surface to \textit{impact zones }according to the nature of the impacts of segments on the leaves  (see section \ref{sec:hierarchyleaves}). In particular, the tangent zone is the only zone with leaves that include tangencies at the wall. For the PWS  the iso-energy tangent zone is of measure zero. Defining: \begin{defn}\label{def:wallgeneralposition}{ A GWS of the form (\ref{eq:wavywal}) is in \textit{general position }if the extremal points of \(V\) are not on the wall and if, at \( q_{2}^{w-min}\), a global minimizer of \(V_{1}(\epsilon _{w}Q(q_2))+V_2(q_2)\),   the following non-degeneracy condition holds:
\begin{equation}\label{eq:generalpositionwall}
 \left(|Q''(q_{2})|+|Q'(q_{2})|\left[\,|V_{1}'(\epsilon_{w}Q(q_{2}))|  +|V_{2}''(q_{2})-V_{1}''(\epsilon _{w}Q(q_{2}))|+|Q'''(q_{2})|\,\right] \right)_{q_{2=}q_{2}^{w-min}}\neq0.
\end{equation} }
 \end{defn}
If  \(Q'(q_{2}^{w-min})\neq0\), e.g. if the wall is a slanted wall with  \(\alpha\in(0,\frac{\pi }{2})\), then (\ref{eq:generalpositionwall}) is automatically satisfied since the extremal points of \(V\) are not on the wall. Proposition \ref{thm:global} states that  provided the GWS is in general position, above some critical known energy, the measure of the iso-energy tangent zone and of the transverse zone are both positive. The division to impact zones defines regions of initial conditions which are, in general, not invariant, so their relation to the impact division is not immediate. Nonetheless,  some insights may be derived with regards to the non-impact zone and the relevance of the tangency zone; We define:

\begin{defn}\label{def:interiorGWS}
An \emph{interior general wall system}  is a GWS with  at least one local minimizer of the potential  \(V\)  inside the billiard domain. Otherwise, the GWS is called an  \emph{exterior GWS}.
\end{defn}

\begin{mainthm}     \label{thm:nonimpactgen}\textit{For sufficiently small \(\epsilon_{r}\) and arbitrary \(\epsilon_{w}\), the phase space measure of the iso-energy non-impact set is \(O(\sqrt{\epsilon_{r}})\) close to the measure of the non-impact set  defined by the IFG and  IEMBD of the GWS  at  \(\epsilon_{r}=0\).  For interior GWS this set has a positive \(O(1)\) measure for a range of energies, whereas for exterior GWS the measure of this set is at most  of \(O(\sqrt{\epsilon_{r}})\).}
\end{mainthm}

In section \ref{sec:interiortangentsegments}
we study when a tangency on the wall corresponds to an external tangent point (like in convex billiards) or, conversely, corresponds to a tangent segment which is included in the billiard domain (like in dispersing billiards). We prove in Theorem \ref{lem:h2tanofq2}  that dispersing wall regions always give rise to tangent segments above certain energy whereas at focusing wall regions the tangent segments are in the billiard for at most  a finite range of energies. At flat wall segments the sign of the normal force  determines whether the tangent segment is in the billiard or is out of the billiard.

When at a given energy surface there are only external tangencies the impact division is simple (the iso-energy singular impact set consists of a finite number of finite length segments in phase space). On the other hand, in all other cases the iso-energy singular impact set consists of orbits which may be chaotic and possibly mix non-trivially with the transverse impact set. This situation is similar to the case of billiards with a dispersing boundary component, yet here the structure of the singularity set changes with the energy. Utilizing the results of Proposition \ref{thm:global} and Theorem \ref{lem:h2tanofq2} we prove in section \ref{sec:proofsmain2} the following main theorems:   \begin{mainthm} \label{thm:singulconvex}  \textit{A  GWS in general position, for which the billiard is convex   and the potential \(V\) increases along the wall normal, has, at  \(\epsilon_{r}=0\),  only external tangent points. If the billiard is strictly convex, and there exists a finite \(K\) such that the potential \(V\) increases along the wall normal for all \(|q_{2}|\geqslant K\),  then,  for sufficiently large energy , the same statement holds. } \end{mainthm}
In formula, if \(\epsilon _{w}Q''(q_{2})\leqslant0\) (see Eq.  (\ref{eq:wavywal})) and the force is directed towards the wall  \(((\hat n\cdot\nabla V)_{q^w(q_2)}>0 \) for all \(q_{2}\in\mathbb{R}\)   where \(\hat n \) denotes the normal pointing into the billiard), the singular impact set of the GWS is, as in smooth convex billiards, restricted to the billiard boundary. Conversely,
\begin{mainthm}  \label{thm:singulconcave}\textit{A GWS in general position, defined on a semi-dispersing billiard with a potential \(V\) which decreases along the wall normal, has, at  \(\epsilon_{r}=0\), on any energy surface, tangent segments inside the billiard at all  wall points which are in the Hill region. If the billiard is dispersing and there exists a finite \(K\) such that the potential \(V\) decreases along the wall normal for all \(|q_{2}|\geqslant K\),  then,  for sufficiently large energy, the same statement holds. } \end{mainthm}
In the general case the singular impact set is expected to be non-trivial:
\begin{mainthm}  \label{thm:concconv}\textit{Above a critical energy, a GWS in general position which has both concave and convex wall segments has a non-trivial singular impact set. }  \end{mainthm}
For particular cases, the tools introduced in sections \ref{sec:integrable}-\ref{sec:Hill-region} can be utilized to find the relative measure of leaves on which tangent segments occur. For example:
\begin{mainthm}\label{thm:tangslantdc}
\textit{ For sufficiently large energy, the slanted wall Duffing-Center system with \(\alpha\in(0,\frac{\pi }{2})\) has a non-trivial singular impact set; For such energies, the tangent segments occur  on a \(\sin^2(\alpha)(1+O(\frac{1}{\sqrt{H}}))\) portion of the leaves. }
\end{mainthm}

\noindent The proofs of Theorems \ref{thm:integrability},\ref{thm:nearintegrability} are presented is section \ref{sec:proofofmainintegrb} and those of Theorems \ref{thm:nonimpactgen}-\ref{thm:tangslantdc} in section \ref{subsec:impactzones}.

\section{Hierarchy of bifurcations of separable systems}\label{sec:integrable}

In this section we set the grounds for the construction of  the impact Fomenko graphs (IFG) and the impact energy momentum diagrams (IEMBD) for a general wall system (GWS), and construct them specifically for the perpendicular wall systems (PWS) at \(\epsilon_r=0\). All along we use the main example for concreteness. We then use these constructions to prove Theorems \ref{thm:integrability} and \ref{thm:nearintegrability}.

We begin with proving the first statement of Theorems \ref{thm:integrability} as the proof reveals some of the basic properties of the PWS:

\begin{lem}\label{lem:integrabilityonly} At  \(\epsilon_{r}=0\)  the perpendicular wall system is a Liouville integrable HIS.    \end{lem}

\begin{proof}Using the separability of the Hamiltonian, the action-angle coordinates in each degree of freedom may be used to verify globally the conditions  Resp\(F\),Resp\(\theta\)  (see definition \ref{def:Liouvilleintegrable}). In the perpendicular cases the wall preserves the energies $H_{1,2}$ upon impact: the reflection rule translates to $p_{1}\rightarrow-p_{1}$ for $\alpha=\frac{\pi}{2}$, and  $p_{2}\rightarrow-p_{2}$ for $\alpha=\frac{\pi}{2}$, so  the values of the constants of motion \(H_{1,2}\) do not change upon collision and condition RespF is verified.
The separability implies that the property of Liouville integrability needs to be verified only for the d.o.f. which is affected by the impact. Indeed, for a one d.o.f. system with impacts, on each impacting periodic orbit, the reflection rule $p_1\rightarrow -p_1$ corresponds to a jump in the angle parameterizing the invariant circle and a gluing between the angle values before and after impact (see \cite{Issi2019}). The motion remains conjugate to rotation on an invariant (cut and glued) circle, with an accordingly modified rotation number, so the condition for Liouville integrability remains fulfilled.
\end{proof}
In  \cite{neishtadt2008jump} action-angle variables are defined for   1 d.o.f. impact systems and it is shown that these are smooth away from the tangent orbit.
\subsection{The  hierarchy of bifurcation for S3BN  systems }

We briefly review the construction of Energy-Momentum Bifurcation Diagrams (EMBD) and Fomenko graphs for systems of the form (\ref{eq:Hint}). A two d.o.f integrable, autonomous Hamiltonian system has two constants of motion - two independent smooth functions of the phase space which remain constant along trajectories. One of which is the Hamiltonian $H$, and the second invariant will be denoted hereafter by \(H_{2}\), and will be taken to be the partial energy of the second d.o.f. (the choice of the invariants in the EMBD is inessential for the integrable dynamics  \cite{Arnold2007CelestialMechanics,fomenko2004integrable}, yet, for the near-integrable and impact settings specific choices are revealing, see   \cite{Arnold2007CelestialMechanics,litvak2004energy,ShlRK10} for the classical smooth theory and \cite{pnueli2018near} for application in impact systems). An Energy-Momentum Bifurcation Diagram (EMBD) is a plot in $(H,H_{2})$ space, which depicts the regions of allowed motion in phase space, and includes the bifurcation set - the  singular values  of $(H,H_{2})$  that correspond to singular level sets of the system (see \cite{lerman1998integrable,Arnold2007CelestialMechanics,radnovic2008foliations,litvak2004energy}).  Regular energy values are
values at which
the bifurcation set does not bifurcate (the singularity curves of nearby leaves do not intersect and do not fold as a function of the energy, see \cite{radnovic2008foliations} for the formal definition). Bifurcation points of this set define the singular energy values.

\begin{figure}
\begin{centering}
\includegraphics[scale=0.2]{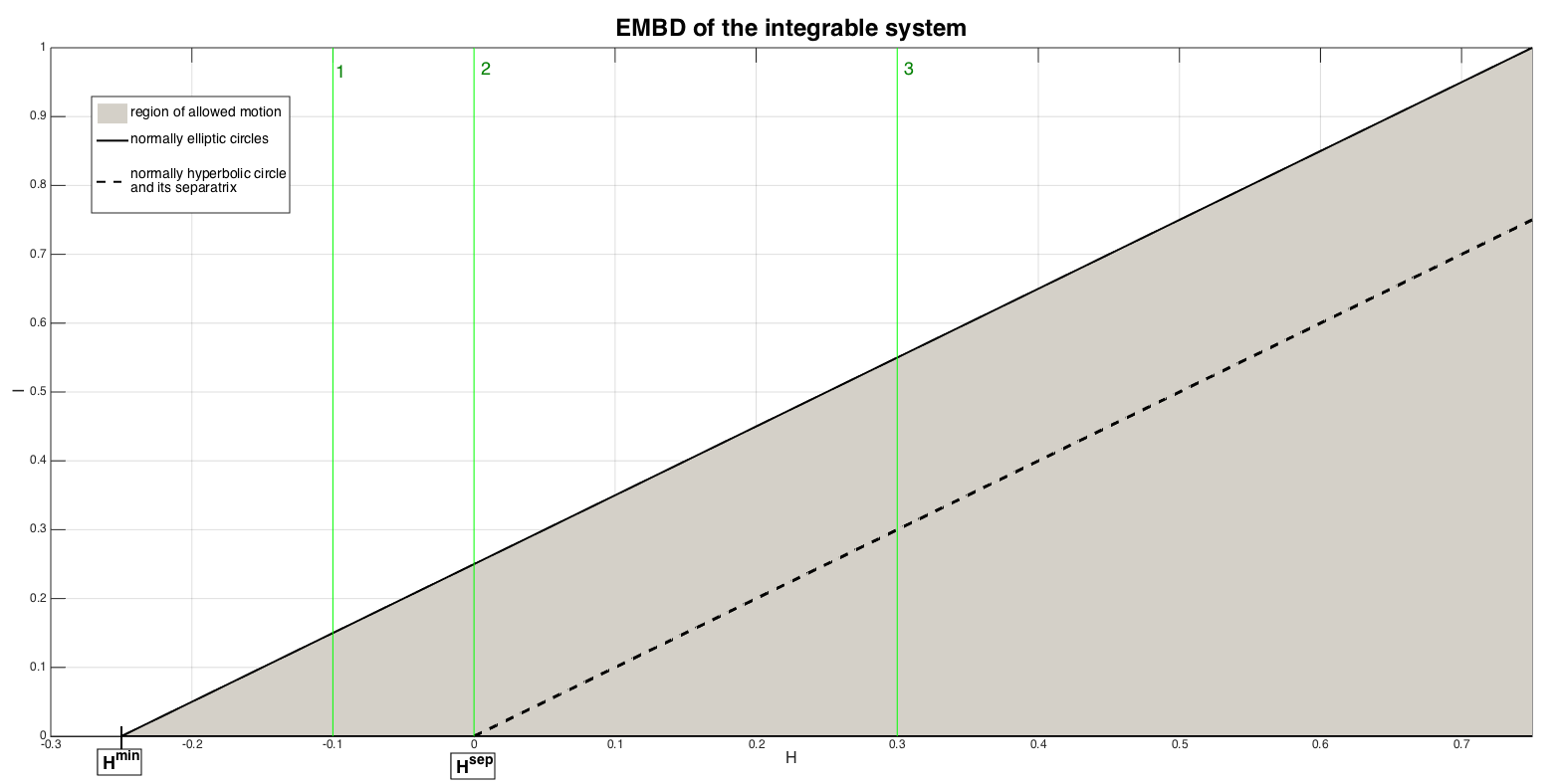}\includegraphics[scale=0.3]{FomenkoFig8table}
\end{centering}
\centering{}\protect\caption{(left) \label{fig:EMBD-integrable}EMBD and (right) Fomenko graphs of the integrable Hamiltonian (\ref{eq:doublewell}).  The two stable types of Fomenko graphs  are separated by the graph at the singular energy value \(H^{sep }\).\ The full circles correspond to elliptic circles with \(I>0\) whereas those with open circles correspond to circles with \(I=0,  \) see \cite{ShlRK05chaos,radnovic2008foliations,ShlRK10}.}
\end{figure}
\noindent{\textbf{The EMBD for the DC Example:}}
\textit{
To  classify the structure of the level sets  of  \(H_{int}^{dc}\) (\ref{eq:doublewell}), let $V_{1,min}=V^{\pm}_{1,min}=V^{dc}_{1}(q_{1s}\pm1), \ V_{2,min}=V^{dc}_2(q_{2c})=0$,   \(H^{min}=\min _{q}V=V_{1,min}+V_{2,min}=-\frac{1}{4}\) and \(H^{sep}=V_1^{dc}(q_{1s})+V^{dc}_2(q_{2c})=0\).
 There are exactly three distinct singular curves in the EMBD (see Figure \ref{fig:EMBD-integrable}) corresponding to the fixed points of the Hamiltonians \(H_1, H_2\) :
\begin{itemize}
\item   The solid line $H_{2}=\omega I=V_{2,min}=0$ corresponds to the elliptic fixed point of \(H_{2}\) at $(q_{2c},0)$. The corresponding level  sets consist of two
normally elliptic circles for $H\in(H^{min},H^{sep})$ (oscillatory motion around each of the centers of \(H_{1}\)) and one normally
elliptic circle for $H>H^{sep}$ (rotational motion around the figure eight of \(H_{1}\)). The line $H_{2}=0$ bounds the allowed region
of motion from below.
\item The solid lines $H^{ell,\pm}(H_{2})=H_{2}+V^{\pm}_{1,min}=H_2+H^{min}$, $H_{2}>0$, correspond to the two elliptic
fixed points  of \(H_{1}\) at $(q_{1s}\pm1,0)$.
Each of these lines correspond to a normally elliptic circle in the full phase
space (due to the symmetric form of the potential these lines coincide, so hereafter we denote them by  \(H^{ell}(H_2)=H^{ell,\pm}(H_2)\)).  These lines bound the region of allowed motion from
above.
\item The dashed line $H^{hyp}(H_{2})=H_{2}=H_2+H^{sep}$, $H_{2}>0$, corresponds to  the level set of the hyperbolic fixed
point   of \(H_{1}\) at   $(q_{1s},0)$. This  level set (\(H_{1}=0\))\ also includes the separatrices in the \((q_{1},p_{1})\) plane. For $H_{2}>0$,
the corresponding level sets are composed of a normally hyperbolic circle
and its separatrices.
\end{itemize}
Since both $H_2$ and $\frac{p_{1}^{2}}{2}\geq0$,
 the region of allowed motion (grey region) is bounded in between the  curves \(H_{2}=0\) and $H^{ell}(H_{2})$.   Each regular point in the EMBD, where regular point is a point in the allowed region of motion which does not belong to any of the singular curves, corresponds to either a single or two disconnected tori.
To distinguish between these cases, namely, for describing the Liouville foliation of energy surfaces, Fomenko graphs are constructed  \cite{fomenko2004integrable}.}

 The isoenergy surfaces correspond to a vertical line in the EMBD for a certain $H$ value. In these graphs (called molecules in \cite{fomenko2004integrable}), each foliation leaf is represented by a point, and hence each smooth family of Liouville tori, a branch, constitutes an edge in the graph. The edges connect vertices which correspond to the singular leaves of the foliation - the intersection of the energy level set with the singular curves in the EMBD. These vertices (called atoms in \cite{fomenko2004integrable}) have different designations according to the type of the singularity. In  \cite{fomenko2004integrable,lerman1998integrable} the  topological classification of isoenergy surfaces of 2 d.o.f Hamiltonian systems is derived. The terminology used here is based on \cite{ShlRK10,ShlRK05chaos,radnovic2008foliations} in which the main ideas behind Fomenko's method are summarized and the hierarchy of bifurcations framework is developed. In particular, the Fomenko graphs supplement the EMBD representation by providing information about the number of tori corresponding to each level set and how these families of tori are connected on a given energy surface.

\noindent{\textbf{The FG for the DC Example:}}
\textit{
For \(H^{min}<H<H^{sep}, \) energy surfaces are composed of two disconnected surfaces, each of them corresponding to a single family of tori connecting the two circles \(q_2=p_2=0,H_{1}(q_1,p_1)=H\)
 (for such \(H\) values this \((q_{1},p_1)\) level set has two circles) with the corresponding circles \(q_{1}=q_{1s}\pm1,p_{1}=0, H_{2}(q_{2},p_2)=H-H^{min}. \) The Fomenko graph for this case corresponds to a pair of edges (so each edge corresponds to a branch of Liuoville leaves) with vertices that correspond to elliptic circles (atoms A) of different topological types (so the marks on these graphs are \(r=0\), and for clarity we denote them by different symbols in the Fomenko graphs, see Fig \ref{fig:EMBD-integrable}). For \(H>H^{sep}\) the energy surface is connected, and the level sets on it have two components for \(H_{2}>H\) and one component for \(H_{2}<H\), so the Fomenko graph is the Y shape molecule, each edge with mark  \(r=\infty\), corresponding to two families of tori connecting two circles (two atoms A) to a separatrix (an atom B) and  another family of tori connecting the separatrix to the single circle \((H_{2}=0,H_{1}(q_1,p_1)=H).  \) In particular, here there are exactly two singular energies, \(H^{min}=V_{1,min}, H^{sep}\),  dividing the energy surfaces to two robust types, those with \(H^{min}<H<H^{sep}\) and those with \(H>H^{sep}\), see Fig \ref{fig:EMBD-integrable}.}

For S3BN Hamiltonians  the level sets are products of the \(H_{1}-\)level sets with the  \(H_{2}-\)level sets, so the leaves of an \((H_{1},H_2)\) level set are the product of all the \(H_{1}-\)leaves with the \(H_{2}-\)leaves.  Thus, the bifurcation set  \(\Sigma\) consists of  a finite number of curves that correspond to  level sets on which one of the d.o.f. has a fixed point, namely\begin{equation}
\begin{array}{ll}
\Sigma _{1,j_{1}}&=\{(H,H_{2}
)|H=H_{2}+V_1(q^{ext}_{1,j_{1}}),H_{2}\geqslant H_{2,min}\}, \\
\Sigma _{2,j_{2}}&=\{(H,H_{2}
)|H_{2}=V_2(q^{ext}_{2,j_{2}}),H\geqslant V_2(q^{ext}_{2,j_{2}})+H_{1,min}\},\\ \\
\Sigma&=\bigcup_{i=1,2}\bigcup_{j_i=1,..,n_i}\Sigma _{i,j_{i}}
\end{array}\label{eq:bifurcationsetnowall}\end{equation}
 where \(H_{i,min}:=\min_{q_{i\in\mathbb{R}}} V_{i}(q_i)\). The singular energies (see \cite{ShlRK10,ShlRK05chaos,radnovic2008foliations}) are the energies at which the curves \(\Sigma _{1,j_{1}}\) and \(\Sigma _{2,j_{2}}\) cross:\begin{equation}\label{eq:critbifsetsm}
\Sigma^{s}=\left\{H|H=H^{s,i_{1},i_{2}}=V_1(q^{ext}_{1,i_{1}})+V_2(q^{ext}_{2,i_{2}}), i_{k}=1,..,n_k,k=1,2   \right\}
\end{equation}
so  \(H^{min}=H_{1,min}+H_{2,min}\in \Sigma^{s}\) is the minimal energy at which motion is allowed.
\subsection{\label{sec:hierarchyleaves}The  hierarchy of bifurcations for general wall systems}

We now add to the EMBD and the IFG information about impacts and tangencies of trajectories that belong to a given level set. Recall that every regular level set of the integrable Hamiltonian is a union of a finite number of tori, the regular Liouville leaves. Singular leaves of (\ref{eq:Hint}) are connected components of singular level sets on which at least one of the Hamiltonians has a fixed point. In a product Hamiltonian (like (\ref{eq:Hint})) each leaf corresponds to a product of the leaves of the one d.o.f. subsystems \(H_{1,2}\), and is spanned by an infinite number of trajectories. Non-resonant regular leaves are covered  by these trajectories densely, resonant regular leaves are covered by infinite number of closed periodic trajectories, whereas singular leaves (e.g. a figure eight separatrix times a circle, an  atom B leaf), are covered by a union of several families of trajectories - periodic ones and bi-asymptotic ones.

Now consider an HIS where the Hamiltonian is integrable (\(\epsilon_r=0\)) and the billiard boundary defines the walls at which impacts occur. Then, some of the leaves of the integrable motion are cut by the boundary, causing trajectories to jump from one cut-leaf to another cut-leaf, where by cut-leaf we mean the union of trajectory segments belonging to a leaf of the integrable Hamiltonian \(H_{int}\)  that reside in the allowed region of motion, namely inside the billiard domain:

\begin{defn} A \emph{  cut-leaf} of   the system (\ref{eq:Hgeneral})\(|_{\epsilon_r=0}\) is the intersection of a leaf of the system  (\ref{eq:Hint}) with the impact allowed region of motion. The \textit{cut-leaf is regular} if the constants of motion are independent on every point on the cut-leaf.  \end{defn}
 In between impacts, a trajectory of  (\ref{eq:Hgeneral})\(|_{\epsilon_r=0}\)   moves on a segment of the smooth motion on the cut-leaf, and there is one-to-one correspondence between the cut-leaf and the leaf. For impact systems, we distinguish between three types of   cut-leaves:

\begin{defn} A \emph{tangent  cut-leaf} is a  cut-leaf which contains at least one tangent segment (the tangent segment may consist of only one point, an exterior tangency point).  A  \emph{transverse impact cut-leaf}   is a  leaf on which some segments impact the wall transversely and all other segments belong to orbits which do not reach the wall at all. A  \emph{non-impact  leaf}  consists only of orbits which do not reach the wall.    \end{defn}

 For a transverse cut-leaf which is a cut leaf of a regular leaf, all trajectory segments impact the wall (transversely). On the other hand, a transverse cut-leaf which is  a cut leaf of a singular leaf (e.g. one which corresponds to separatrix level set in one of the d.o.f.), may contain segments which do not reach the wall (e.g. see Fig \ref{fig:ps-alphapi2}b).
 \begin{defn} A \emph{tangent branch} is an iso-energetic family of  regular Liouville leaves (represented by an edge of the Fomenko graphs) which contains   tangent leaves in its interior.    \end{defn}

\begin{defn} {A \emph{tangent level set} on an isoenergy surface \(H,\) $(H_{1,tan}(H),H_{2,tan}(H)),$ is a level set which contains tangent leaves. A \emph{transverse impact level set} is a level set which does not contain tangent leaves and contains transverse impact leaves. An \emph{non-impact level set} is a level set which contains only non-impacting leaves. }   \end{defn}

Notice that a transverse impact level set may also have some non-impacting leaves - these are leaves that are not affected by the wall (see below for examples).
Along a family of tori belonging to a tangent branch we can find two possible boundary tori:
\begin{defn} \label{def:boundarytanleaf} A tangent leaf is a \textit{tangent boundary   leaf}  if  it divides the tangent branch to impacting and non-impacting tori. A tangent leaf is a   \textit{transverse boundary  leaf} if it divides  a tangent branch to tangent and transverse impacting tori.
   \end{defn}
    Notice that by definition, a tangent leaf  may also contain impacting segments and non-impacting trajectories (see section \ref{sec:Hill-region}) whereas a boundary tangent leaf contains only tangent and non-impacting segments.    \begin{defn}\label{def:iembd}The \emph{Impact EMBD (IEMBD)} is the EMBD of the underlying integrable Hamiltonian where the level sets in the allowed region of motion are divided to transverse impacting zone, non-impacting zone and tangent zone. The \textit{impact zone} of the IEMBD includes the tangent and transverse impact zones. \end{defn}

    \begin{defn}\label{def:ifg} The \emph{Impact Fomenko graphs} (IFG) of a GWS are derived from the Fomenko graphs of the integrable system  (\ref{eq:Hint}) as follows; for non-impacting leaves the IFG is identical to the FG. The tangent and transverse impacting cut-leaves are marked with the symbol \(\tau\) and \(im\) respectively, with edges denoting families of regular cut leaves and atoms denoting the singular cut-leaves. The edges  and atoms of the FG which are not in the billiard domain are  eliminated.\end{defn}

\begin{defn} \label{def:bifurcationset} The \textit{bifurcation set \(\Sigma\) of the IEMBD }of a GWS consists of the bifurcation set of the EMBD and the curve corresponding to the level sets of the tangent and transverse boundary  leaves. An energy level \(h\) of this system is \textit{singular} if the bifurcation set of the IEMBD has singularities at \(h\) and is \textit{regular} otherwise.
   \end{defn}

In particular, at regular energies, the tangent and transverse \textit{boundary leaves} are regular, namely  are cut leaves of regular tori and not of any of the singular level sets.
At singular energies the structure of the IFG changes (see below).

With the above definitions at hand we are ready to fully describe the IEMBD and the IFG and the corresponding structure of the impact and the non-impact sets of   PWS at \(\epsilon_r=0\). For concreteness, we  first describe the impact division for the  horizontal wall case ($q^{w}_{\alpha=0}$), where the billiard is the lower half plane (the vertical wall case is realized by replacing the indices \(1\leftrightarrow2\)).

Define the tangent energy, the minimal singular tangent energy, the tangent curve:
\begin{equation}
H_{2,tan}:=V_{2}(0), \quad H^{s1}=H_{1,min}+H_{2,tan},\quad \Sigma_{2,tan}=\{(H,H_{2,tan})|H\geq H^{s1}\}
\end{equation}
and the $q^{w}_{\alpha=0}$ bifurcation set : \begin{equation}\label{eq:sigmaalfa0}
\Sigma_{2}=\Sigma_{2,tan}\cup\bigcup_{j_1=1,..,n_1 }\Sigma _{1,j_{1}}\cup\ \bigcup_{\{j_2=1,..,n_2|q^{ext}_{2,j_{2}}\leq 0 \} }\Sigma _{2,j_{2}}.
\end{equation}
where \(\Sigma _{i,j_{i}} \) are the bifurcation curves of the smooth system (Eq. (\ref{eq:bifurcationsetnowall})). The curves in  \(\Sigma_{2}\) intersect at the singular tangent energies:\begin{equation}
\Sigma_{2}^{s,tan}=\left\{H=H^{tan,i_{1}}=V_1(q^{ext}_{1,i_{1}})+V_2(0),i_{1}=1,..,n_{1}   \right\}
\end{equation}
and at the singular interior critical energies:
\begin{equation}
\Sigma_{2}^{s}=\left\{H=H^{s,i_{1},i_{2}}=V_1(q^{ext}_{1,i_{1}})+V_2(q^{ext}_{2,i_{2}}) \text{ with } q^{ext}_{2,i_{2}}\leqslant0,i_{1,2}=1,..,n_{1,2} \right\}.
\end{equation}

\begin{prop}\label{lem:hetanglevel1}      The impact structure of the PWS  with $q^{w}_{\alpha=0}$ at  \(\epsilon_{r}=0\)  changes at the bifurcation set \(\Sigma_{2}\) (Eq. \ref{eq:sigmaalfa0}). In particular, for \(H\geq H^{s1}\)  the line \(\Sigma_{2,tan}\)  separates the IEMBD between the impacting (\(H_{2}>H_{2,tan}\)) the tangent  (\(H_{2}=H_{2,tan}\)) and the non-impacting (\(H_{2}<H_{2,tan})\) level sets. The  IFG  changes as the energy is varied only at    the singular impact energies,  namely at    \(\Sigma^{s}_2\cup\Sigma^{s,tan}_{2}\).
\end{prop}

\begin{proof}
The key observation is that for PWS  the tangency property can be studied in one of the \(H_{i}\) systems, here the \(H_{2}\) system. Indeed, the wall equation for  $q^{w}_{\alpha=0}$  is $ q^{wall}(q_1,q_2)=-q_{2}=0$ (\(D=\{q|q_{2}\leqslant 0\}\), see Eq. (\ref{eq:straightwall})), hence,  a tangency    occurs when  \(q_{2}=p_{2}=0\) and    \(H_{2}(q_{2},p_{2})=V_2(0)=H_{2,tan}\). When    \(H_{2}=H_{2}(q_{2},p_{2})=\frac{p_2^2}{2}+V_2(0)>H_{2,tan}\) a transverse impact  occurs on the \(H_{2}\) level set at  \(q_{2}=0\).   Since for 1 d.o.f. \textit{mechanical} systems the points \((\pm p_{2},q_2)\) belong to the same unique leaf which includes \(q_{2}\),  for all \(H_{2}\geqslant H_{2,tan}\) there is a unique leaf of \(H_{2}\) that reach the wall.  By the S3BN assumption all other \(H_2\)-leaves of the \(H_{2}\geqslant H_{2,tan}\) level sets are bounded away from the wall, and some of these may reside outside of the billiard domain (\(H_{2}\)-leaves with \(q_{2}>0\)).

Extending these results to the iso-energetic level sets of the two d.o.f. product system, where the leaves of the product system at the level set \((H_{1}=H-H_2,H_2)\) are the product of the smooth \(H_{1}\)-leaves with the leaves and the cut-leaves of the \(H_2\) system completes the proof.
First, by the above analysis transverse impacts occur only if \(H_{2}> H_{2,tan}\). The level set   \((H_{1,tan}(H)=H-H_{2,tan},H_{2,tan}\)) is in the allowed region of motion iff  \(H\geq H^{s1}\in \Sigma^{s,tan}_{2}\), and thus,  \(\Sigma_{tan}\) indeed divides the allowed region of motion to impacting and non-impacting level sets.

  Next we prove that the IFG change exactly at     \(\Sigma^{s,tan}_{2}\cup\Sigma^{s}_2\), namely, that any change that appear as a result of the wall are included in \(\Sigma^{s,tan}_{2}\) and that all other  singular values at which the IFG changes are included in  \(\Sigma^{s}_2\).

Notice that the tangent/transverse impact  leaves of the full system are the product of the unique tangent/transverse impact \(H_{2}\)-leaf  with all the \(H_{1}\)-leaves, where  \(H_{1}=H-H_2\). There is a finite number of such leaves by the S3BN assumption.

The  \(\Sigma^{s,tan}_{2}\) singularities:

  \(\Rightarrow\)For energies \ \(H^{tan,i_{1}}\in\Sigma^{s,tan}_{2}\) the \(H_{1}\)-level set  \(H_{1,tan}(H^{tan,i_{1}})=V_1(q^{ext}_{1,i_{1}})\) includes a singular \(H_{1}\)-leaf.  At  \(H_{1,tan}(H^{tan,i_{1}}-\Delta H)\)  the   \(H_{1}\)-level set below the singular value  \(V_1(q^{ext}_{1,i_{1}})\) impact the wall whereas for  \(H_{1,tan}(H^{tan,i_{1}}+\Delta H) \)    \(H_{1}\)-leaves  with \(H_1>V_1(q^{ext}_{1,i_{1}})\)  impact the wall as well. Thus the tangent branches change at such values and the IFG changes (even though the FG does not).

\(\Leftarrow\)If the IFG changes at a value \(H^{s}\)  as a result of the wall,  there is a change in its division to non-impact, tangent and transverse impact edges. The claim is that such changes occur only when \ \(H^{s}\in\Sigma^{c,tan}_{2}\), namely, only when the   \(H_{1,tan}(H^{s})\) level set includes a singular leaf. Assume it is not, namely that     \(H_{1,tan}(H^{s})\) is a regular level set. Then, if \(V'(0)\neq0\), the tangent level set is regular, so that tangent leaves at \(H^{s}\) are a finite number of tangent tori that divide a finite number of edges of the FG to impacting and non-impacting parts. A small change in \(H\) can only change the position of the dividing tangent torus along the tangent branches and thus there can be no topological change in the IFG.  If \(V'(0)=0\), namely the tangent \(H_{2}\)-leaf  is singular yet the     \(H_{1,tan}(H^{s})\) is a regular level set, all the tangent leaves are singular of the same type, and they divide the IFG to impacting and non-impacting edges.  Small changes in \(H\) near \(H^{s}\) preserve this same singular product structure -  the impacting and non-impacting edges do not change across  \(H^{s}\), so, again, no change in the IFG can occur.

The  \(\Sigma^{s}_D\) singularities:

\(\Rightarrow\) For the singular energy \(H^{s}=H^{s,i_{1},i_{2}}=V_1(q^{ext}_{1,i_{1}})+V_2(q^{ext}_{2,i_{2}})\in\Sigma^{s}\) the FG of the product system (\ref{eq:Hint}) changes  at the neighborhood of the unique singular leaf which includes the fixed point  \((q^{ext}_{1,i_{1}},p_1=0,q^{ext}_{2,i_{2}},p_2=0)\) (see \cite{fomenko2004integrable,dragovic2009bifurcations}). If  \((q^{ext}_{1,i_{1}},q^{ext}_{2,i_{2}})\in D\), namely \(H^{s,i_{1},i_{2}}\in\Sigma_{D}^{s}\), at least part of this \((q^{ext}_{1,i_{1}},q^{ext}_{2,i_{2}})-\)leaf and the leaves in its neighborhood of the EMBD (i.e. on the FG of the iso-energy surfaces near   \(H^{s}\)) are  in the billiard domain, so the  topological changes that are reflected in the changes of the FG at \(H^{s,i_{1},i_{2}}\) also induce specific changes in the IFG. These topological changes may be listed for the different types of the fixed points of S3BN systems (center-center, saddle-center and saddle-saddle of the mechanical kinds,  so these are the simplest  cases of the general classification in \cite{fomenko2004integrable,dragovic2009bifurcations}) and the different possible positions of the tangency on the singular leaf (see e.g. Fig. \ref{fig:ps-alphapi2}c,d).
Note that if \(V_{2}\) has an extremal point at the origin (\(V_{2}'(0)=0\)) then the singular energies \(V_1(q^{ext}_{1,i_{1}})+V_2(0)\) are in   \(\Sigma^{s,tan}_{2}\cap \Sigma^{s}_{D}\).

   \(\Leftarrow \) There are no other singular energies. We already showed that all the singular energies that are produced by the wall are included in \(\Sigma^{s,tan}_{2}\). All other singular energies that stem from changes in the FG structure near singular level sets are included in \(\Sigma^{s}\). Hence, to complete the proof, we need to show that energy values in \(H^{s,i_{1},i_{2}}=V_1(q^{ext}_{1,i_{1}})+V_2(q^{ext}_{2,i_{2}})\in(\Sigma^{s}\backslash\Sigma^{s}_D)\setminus\Sigma^{s,tan}_2\) do not  induce changes in the IFG. For such values,  \((q^{ext}_{1,i_{1}},q^{ext}_{2,i_{2}})\notin D\), so, there are two cases to consider. The first is  when the entire \((q^{ext}_{1,i_{1}},q^{ext}_{2,i_{2}})-\)leaf  and its neighborhood are all outside of \(D\) (this is always the case for the center-center singularity). Then, all the changes in the FG near this level set are in branches that are eliminated from the FG and are not include in the IFG, so indeed no changes in the IFG are induced. The second case is when part of the \((q^{ext}_{1,i_{1}},q^{ext}_{2,i_{2}})-\)leaf   is in \(D\), namely, its \(H_{2}\) component is an \(H_{2}\)-cut leaf.  Such a situation is possible for the S3BN and   $q^{w}_{\alpha=0}$  case only if  \((q^{ext}_{2,i_{2}}>0,p_{2}=0)\) is a saddle point of the \(H_{2}\) system. Since  \(H^{s,i_{1},i_{2}}\in\Sigma^{s}\backslash\Sigma^{s}_D\)   if there are additional \(H_{2}\)-fixed points on this \(H_{2}\)-level set they  are not in \(D\). Since  \(H^{s,i_{1},i_{2}}\in\Sigma^{s}\backslash\Sigma^{s,tan}_2\)   the impacts with the wall are transverse.   Namely, the \(H_{2}\)-cut-leaf is a transverse impact regular cut leaf, and by the integrability of the PWS, it is  a circle. Hence, its nearby cut-leaves are also regular transverse cut leaves, namely circles.   So, while the \(H_{1}\)-level sets near  \(V_{1}(q^{ext}_{1,i_{1}})\)  change their topology and the IFG reflects these changes, the impact level sets of \(H_{2}\) near   \(V_{2}(q^{ext}_{2,i_{2}})\)  do not change their topology. Hence the IFG does not change at  \(H^{s,i_{1},i_{2}}\).

\end{proof}

The above proof provides the construction of all the robust IFG of the PWS  with $q^{w}_{\alpha=0}$ (and correspondingly for  $q^{w}_{\alpha=\pi/2}$) at  \(\epsilon_{r}=0\). For  \(H<H^{s1}\)   each connected component of  the \(H\)-iso-energy surface either belongs to the non-impact set and then its IFG is identical to that of the FG, or, it is not in the allowed region of motion and then its graph is eliminated (projects to the upper half plane). For  \(H\geq H^{s1}\), at each singular energy, the structure of the IFG changes;
At singular energies of \(\Sigma^{s,tan}_{2}\) new impact branches of the \(H_{1}\) system appear. At the singular energy values \(\Sigma_{2}^{s}\)  the structure of the level sets changes as in the smooth theory, with, when applicable, the cut leaves replacing the leaves of the full system. We demonstrate this construction for the main example.

 \begin{figure}
\begin{centering}
\includegraphics[scale=0.35]{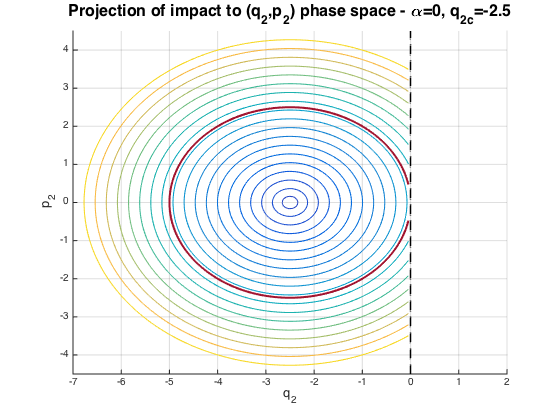}\includegraphics[scale=0.35]{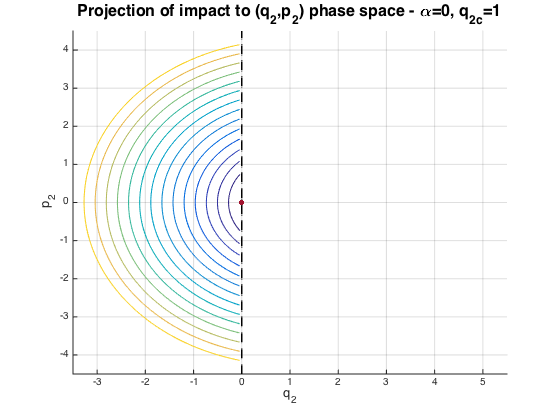}
\par\end{centering}
\centering{}\protect\caption{\label{fig:ps-alpha0}The  \((q_{2},p_{2})\)  plane of the impact system  (SW) with $H_{int}^{dc}$ (\ref{eq:doublewell}) with the wall $q^{w}_{\alpha=0}$. The wall $q_{2}=0$ (dashed black) cuts the leaves of the \(H_{2}\) levels sets when \(H_2>H_{2,tan}\), where the \(H_{2,tan} \) defines the tangent leaf (bold magenta). (a) $q_{2c}=-2.5, H=10$ (center point is inside
the billiard) (b) $q_{2c}=1, H=10$ (center point is outside).}
\end{figure}

\subsection{The perpendicular wall system of the Duffing-Center Hamiltonian.}
{The proof of Proposition \ref{lem:hetanglevel1} demonstrates that the first step in the construction of the IEMBD and IFG is the identification  of the tangent level sets. Figure \ref{fig:ps-alpha0} depicts the two possible robust impact geometries of the linear oscillator  (\(H_2\) of the Hamiltonian (\ref{eq:doublewell})) with the horizontal wall $q^{w}_{\alpha=0}$) - the interior and the exterior cases.  In the first case the center (which is the only singular level set of \(H_{2}\)) is inside the billiard, so the tangent level set is interior and thus the corresponding    \((H_{1,tan}(H),H_{2,tan})\) leaf divides the tangent branch to impacting and non-impacting leaves. In the second case,  the center and thus the tangent level sets are exterior to the billiard, so all the level sets in the allowed region of motion are impacting. The singular case is when the center is on the wall, where the only tangent \(H_{2}\)-level set is the center and all other level sets are impacting.}
 Figure \ref{fig:ps-alphapi2} depicts some of  the impact geometries of the double well potential  (\(H_1\) of the Hamiltonian (\ref{eq:doublewell})).    The top two images depict two of the five possible robust, non-singular cases. The bottom figures depict two of the four singular cases, cases where
the wall (at $q_1=0$) is either tangent to the left separatrix  or passes through one of the three fixed points.

In Fig  \ref{fig:ps-alphapi2}a, the level set \(H_{1,tan}\)
 is outside the separatrix. It divides the nearby level sets to impacting and non-impacting level sets and induces similar division of the product system.  In Fig  \ref{fig:ps-alphapi2}b, the level set \(H_{1,tan}\)
 is inside the separatrix and consists of two components. Here, only the left component contains the origin, and thus only the left branch of the  level sets changes from non-impacting to impacting, namely the left branch is the tangent branch. This multiplicity also carries to the product system.   The other three
non-singular scenarios may be similarly analyzed.
The two bottom images present two singular cases.
In the left bottom image, $q_{1s}=\sqrt{2}$,
i.e. - the left lobe of the separatrix is exactly tangent to the wall
equation. All level sets outside the separatrix impact,
and all level sets inside do not. In the right bottom image $q_{1s}=0$,
and once again all level sets outside the separatrix impact, and all
level sets inside the right branch, do not (the left branch is out of the billiard domain). Tangency
is expected to occur on the right separatrix level set at $q_{1}=0$, yet, this point cannot be reached
in finite time.  So, the right separatrix solutions of the product system are homoclinic to a tangent periodic orbit. See Appendix \ref{appdx:classification} for a detailed description  of the resulting trajectories.

\begin{figure}
\begin{centering}
\includegraphics[scale=0.4]{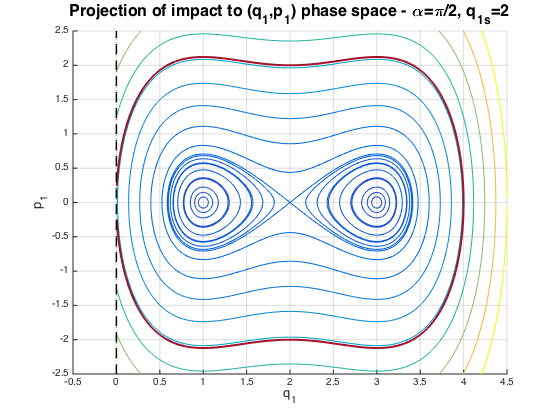}
\includegraphics[scale=0.4]{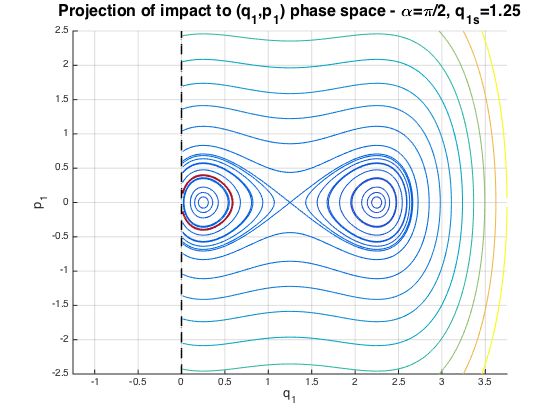}
\par\end{centering}
\begin{centering}
\includegraphics[scale=0.4]{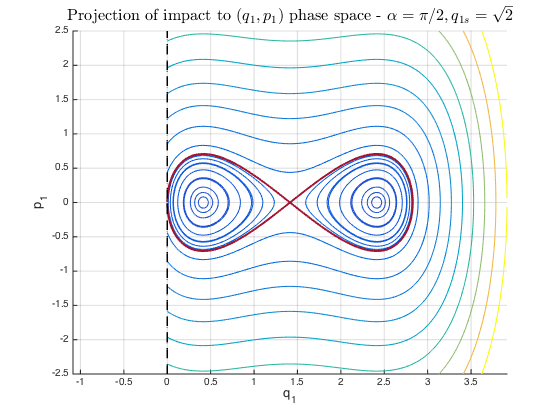}
\includegraphics[scale=0.4]{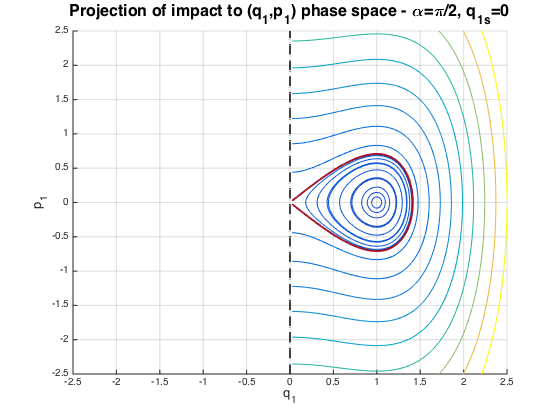}
\par\end{centering}
\centering{}\protect\caption{\label{fig:ps-alphapi2}The  \((q_{1},p_{1})\)  plane of the impact system  (SW) with $H_{int}^{dc}$ (\ref{eq:doublewell}) with the wall $q^{w}_{\alpha=\frac{\pi }{2}}$ and energy $H=10$. The wall $q_{1}=0$ (dashed black) cuts the leaves belonging to the tangent branch of the \(H_{1}\) levels sets when \(H_1>H_{1,tan}\), where  \(H_{1,tan} \) defines the tangent leaf (bold magenta).  Two of the regular cases (a,b) and two of the singular cases (c,d) are shown. (a) $q_{1s}=2$ (wall to the left of separatrix) (b) $q_{1s}=1.25$ (wall intersects the left loop of the separatrix). (c)  $q_{1s}=\sqrt{2}$ (wall tangent to the left lobe of the separatrix) (d) \(q_{1s}=0 \) (wall tangent to the hyperbolic fixed point of \(H_{1}\)).}
\end{figure}

Figures \ref{fig:ps-alpha0} and \ref{fig:ps-alphapi2} demonstrate that for each of the \(H_{i}\) systems there is a unique tangent level set which  divides the \(H_{i}\)-plane to two distinct types of level sets - those having impacting \(H_{i}\)-leaves and those which have only non-impacting \(H_{i}\)-leaves (possibly outside the billiard).

 \begin{figure}
\begin{centering}
\includegraphics[scale=0.5]{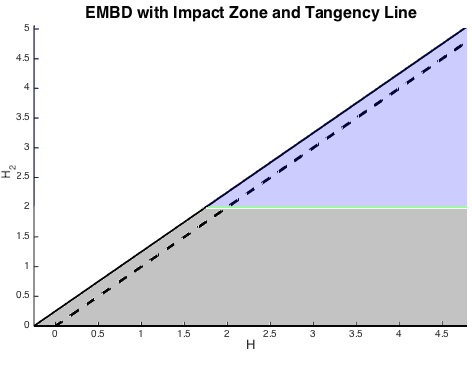} \qquad \includegraphics[scale=0.33]{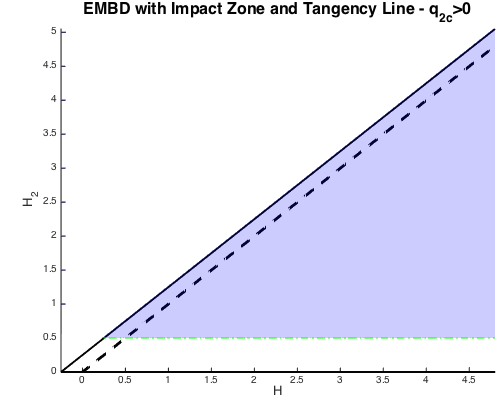}
\par\end{centering}
\protect\caption{\label{fig:EMBD-alpha0} Two of the three  robust IEMBDs of  the impact system (SW) with $H_{int}^{dc}$ (\ref{eq:doublewell}) with the horizontal wall $q^{w}_{\alpha=0}$. (left) The \(H_{2}\) center point is inside
the billiard ($q_{2c}=-2$ ) (right) The center point is outside the billiard ($q_{2c}=1)$. Impact (blue), non-impact (grey) and tangency (green) zones are depicted.
The tangency zone reduces here to the line $(H, H_{2}=V_2(0))$. For
$q_{2c}>0$ the tangent level sets is exterior to the billiard, and thus the impact zone ia an open set which does not include the tangency line (marked by a dashed green line). }
\end{figure}

Figure \ref{fig:EMBD-alpha0} presents the IEMBD for the PWS for the Duffing-Center Hamiltonian (\ref{eq:doublewell}) with the wall $q^{w}_{\alpha=0}$  for  two different cases corresponding to positive and negative \(q_{2c}\) (see Figure \ref{fig:ps-alpha0}). The three different types of level sets (impacting, tangent and non-impacting) are marked by different colors (blue, green and grey, respectively). The corresponding Impact Fomenko Graphs (IFG) are shown in Figure \ref{fig:Fomenko0in}, where the color code for the three different types of leaves is the same as in the IEMBD (with black replacing grey for non-impacting). The  tangent leaves are denoted by a diamond with a subscript $\tau$ and the impacting leaves by a dashed blue line and subscript  $_{im}$.  Impacting singular  level sets are denoted by a  $+$ sign within the atom symbol.
Tangent singular level sets (which appear at the singular energies) are denoted by a diamond sign within the atom symbol.

  To fully classify all possible changes in the IEMBD and the IFG,  we study how the \textit{singular energy values} change with parameters. For example, for the Impact-Duffing-Center system - we study how the singular energies depend on \(q_{2c} , q_{1s}\), the signed distances of the saddle-center point of the potential $V$ from the horizontal/vertical  walls.

\begin{figure}
\begin{centering}
\includegraphics[scale=0.4]{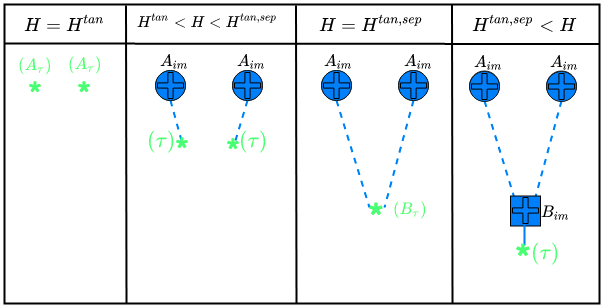}
\includegraphics[scale=0.4]{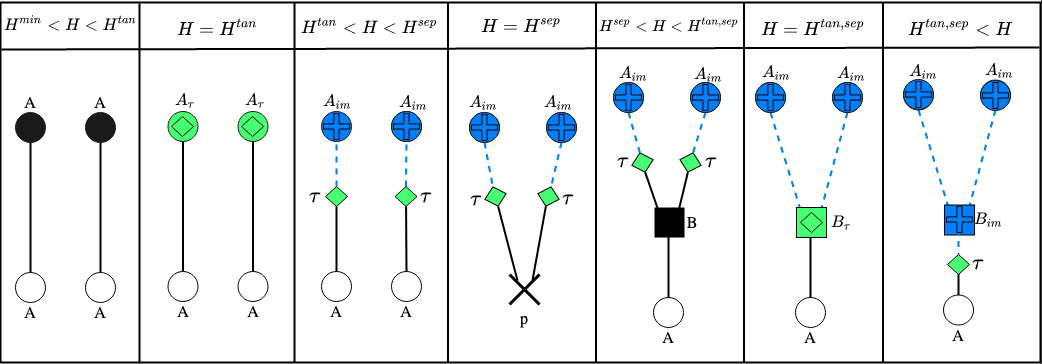}
\includegraphics[scale=0.4]{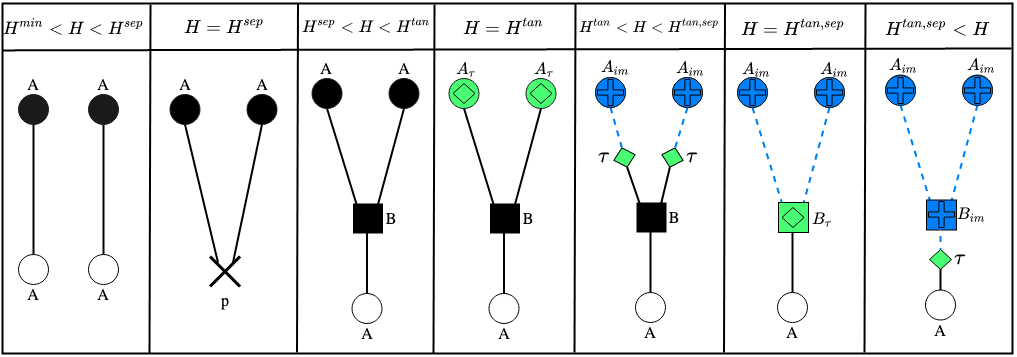}
\par\end{centering}
\centering{}\protect\caption{\label{fig:Fomenko0in} The Impact Fomenko graphs  of the Impact-Duffing-Center system with the wall $q^{w}_{\alpha=0}$ for the three robust regimes (A-C of Fig \ref{fig:Hsingalpha0})   (top) The center point is outside
the billiard (case A of Fig \ref{fig:Hsingalpha0}) (middle) The center point is inside the billiard and the tangent circles appear at lower energy than the separatrix circle ($ q_{2c}<0$ and $ H^{tan}< H^{sep}$ , case B of Fig \ref{fig:Hsingalpha0}.) (bottom) $ q_{2c}<0$ and $ H^{tan}> H^{sep}$ (case C of Fig \ref{fig:Hsingalpha0}). Tangency is marked by a diamond and singular impacting atoms are distinguished by singular non-impacting atoms by the $+$ sign inside the atom. Tangent atoms which are outside the allowed region of motion are denoted by an asterisk and their designation is in parentheses. }
\end{figure}

\begin{figure}
\begin{centering}
\includegraphics[scale=0.45]{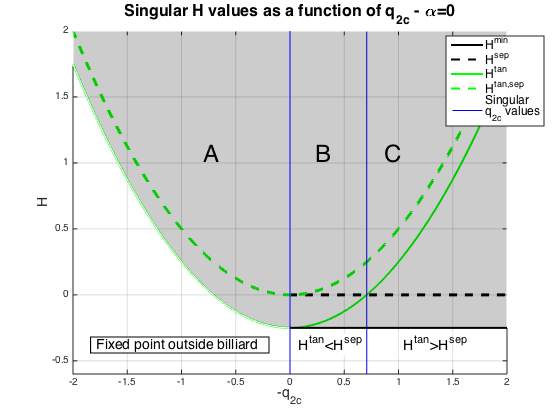}
\par \end{centering}
\centering{}\protect\caption{\label{fig:Hsingalpha0}The dependence of the singular values of $H$ on the center location (the parameter $q_{2c}$) for the horizontal wall $q^{w}_{\alpha=0}$. At the blue vertical lines two singular energy values coincide and the IEMBD and the IFG change at these values as shown in Fig. \ref{fig:EMBD-alpha0} and \ref{fig:Fomenko0in}. The shaded area depicts the allowed region of motion  in \(H\) of the HIS. }
\end{figure}

For the Hamiltonian (\ref{eq:doublewell}) the singular energy values dictated by the smooth dynamics, \(\Sigma^{s}\), are $H^{min}=V_1(q_{1s}\pm1)+V_2(q_{2c})=-\frac{1}{4}$ and $H^{sep}=V_1(q_{1s})+V_2(q_{2c})=0 $, and these values are in \(\Sigma^{s}_D\)  when these fixed points are inside the billiard domain. The singular energies dictated by the tangency  to the the horizontal wall, \(\Sigma_{2}^{s,tan}\), are  $  H^{tan,\pm}=V_1(q_{1s}\pm1)+V_2(0)=-\frac{1}{4}+\frac{\omega^{2}q_{2c}^{2}}{2}$  and $H^{tan,sep}=V_1(q_{1s})+V_2(0)=\frac{\omega^{2}q_{2c}^{2}}{2}$. Figure \ref{fig:Hsingalpha0} depicts the dependence of these functions on \(-q_{2c}\). There are exactly two bifurcation points where two of the singularity curves meet, defining three distinct regions of robust behaviors: A:   $ q_{2c}>0$ , B: $q_{2c}<0$ and  \ $ H^{tan,\pm}> H^{sep}$ and C:  $q_{2c}<0$ and  \ $ H^{tan,\pm}< H^{sep}$ .

Two of the robust IEMBD  (cases A, C) are shown in Fig. \ref{fig:EMBD-alpha0} (case B corresponds to a downward shift of the tangency line in C) and all the three robust IFG  are shown in Fig. \ref{fig:Fomenko0in}, completing the full classification of all possible behaviors of the Hamiltonian (\ref{eq:doublewell}) with impacts with a horizontal wall.
For the DC with a horizontal wall we see that the topology of the level set foliation on the energy surfaces in the three different scenarios is unchanged. The difference between the three scenarios has to do with the order of the singular energies at which the tangent branches change.

\begin{figure}
\begin{centering}
\includegraphics[scale=0.45]{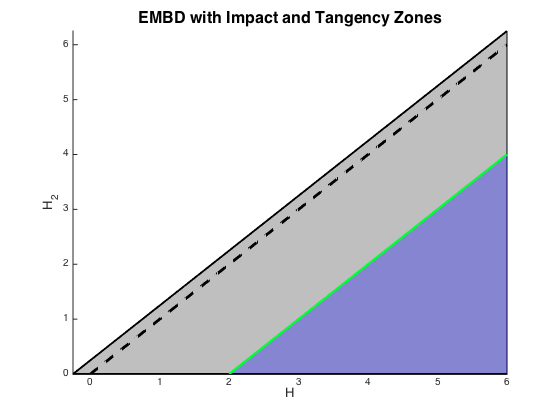}\includegraphics[scale=0.45]{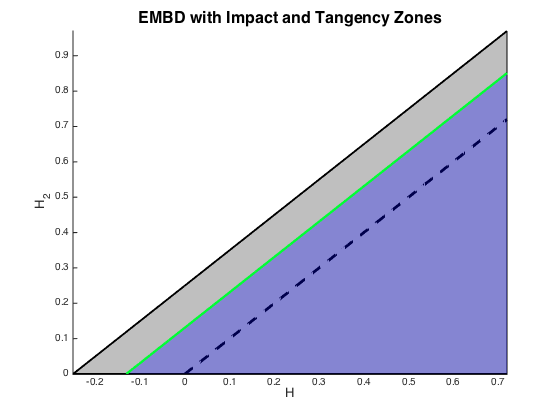}
\par\end{centering}
\centering{}\protect\caption{\label{fig:EMBD-alphapi2}Two of the regular IEMBDs of the Impact-Duffing-Center system with the vertical wall $q^{w}_{\alpha=\frac{\pi}{2}}$.  Impact (blue), non-impact (grey) and tangency (green) zones are depicted. The tangency zone reduces to the line
 $(H,H_{2}=H_{2,tan}=H-V_{1}(0)$).  (a) The  \(H_{1}\) separatrix is  inside
the billiard (region E of Fig. \ref{fig:Hsingalphapi2}, here, $q_{1s}=2$) (b) The  \(H_{1}\) separatrix is  cut by the wall (regions B-D of Fig. \ref{fig:Hsingalphapi2}, here, $q_{1s}=1.3$). See top two figures in Fig \ref{fig:ps-alphapi2} for phase space representation.}
\end{figure}

A similar analysis of the vertical wall ($q^{w}_{\alpha=\frac{\pi }{2}}$), where $H^{tan}=V_1(0)=-\frac{q_{1s}^{2}}{2} -\frac{q_{1s}^{4}}{4} $), is summarized by   Figures \ref{fig:EMBD-alphapi2}-\ref{fig:Fomenkopi2sing}. In particular, Fig. \ref{fig:Hsingalphapi2} demonstrates that for a vertical wall there are five different robust regimes of the IEMBD and IFG. Some examples of these diagrams are depicted in Figs \ref{fig:EMBD-alphapi2},\ref{fig:Fomenkopi2out}, \ref{fig:Fomenkopi2sing}.  Here,  different scenarios correspond to different topology of the energy surfaces, even though the IEMBDs are identical. Indeed, the multi-branch ambiguity in the IEMBD of regions B-E  of Fig \ref{fig:Hsingalphapi2} is lifted by the IFG and the singular energy diagram which specifies which of the two possible branches undergoes impacts, see e.g.  Fig. \ref{fig:Fomenkopi2sing}. These diagrams demonstrate that the level set topology described by the IFG differs from the corresponding FG and that the information encoded in the IFG is not encoded in the IEMBD alone.

\begin{figure}
\begin{centering}
\includegraphics[scale=0.45]{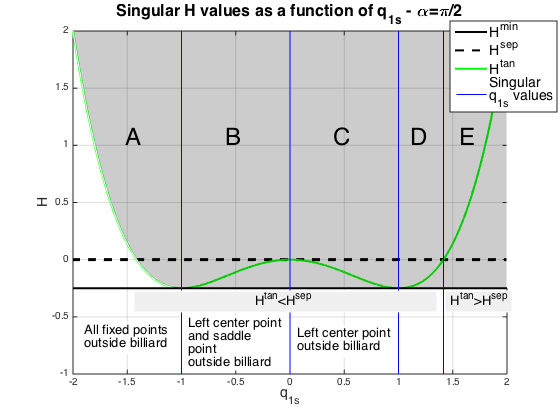}
\par \end{centering}
\centering{}\protect\caption{\label{fig:Hsingalphapi2}The dependence of the singular values of $H$ on the parameter $q_{1s}$  - the signed distance of the saddle-center point of the potential $V$ from the vertical wall  $q^{w}_{\alpha=\frac{\pi }{2}}$.  Four bifurcation values (blue vertical lines) are identified, leading to five robust regimes A-E. The shaded area depicts the allowed region of motion  in \(H\) of the HIS.}
\end{figure}

\begin{figure}
\begin{centering}
\includegraphics[scale=0.41]{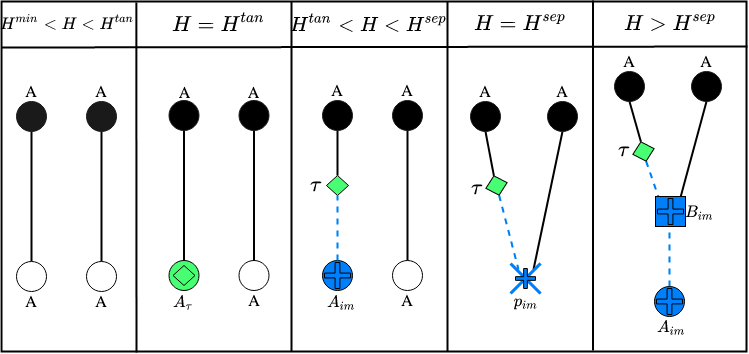}\qquad
\includegraphics[scale=0.4]{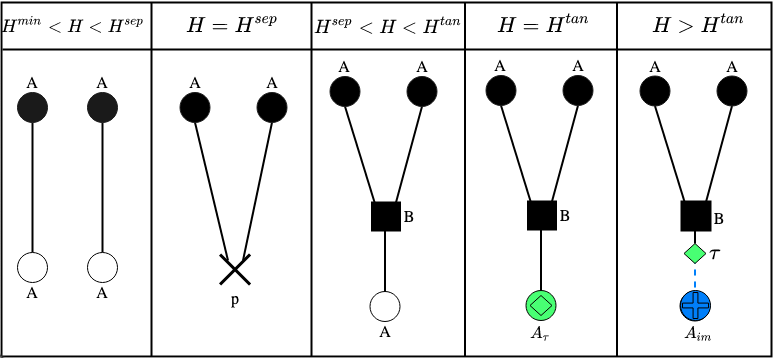}
\par \end{centering}
\centering{}\protect\caption{\label{fig:Fomenkopi2out}Two of the regular Impact Fomenko graphs of the Impact-Duffing-Center system with the vertical wall $q^{w}_{\alpha=\frac{\pi}{2}}$ (a) IFG corresponding to region D in Fig . \ref{fig:Hsingalphapi2} (b) IFG corresponding to region E in Fig . \ref{fig:Hsingalphapi2}.}
\end{figure}

\begin{figure}
\begin{centering}
\includegraphics[scale=0.4]{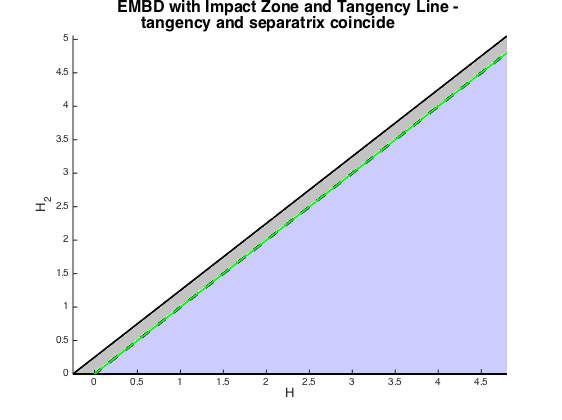}\qquad
\includegraphics[scale=0.26]{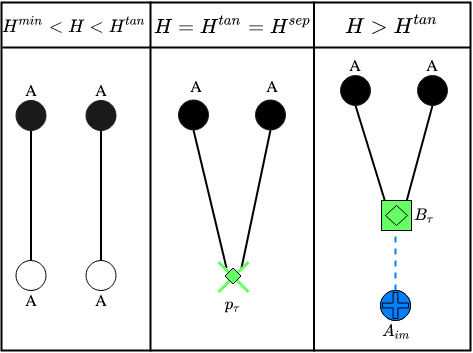} \includegraphics[scale=0.26]{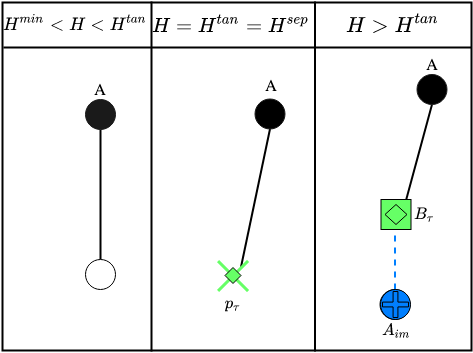}
\par \end{centering}
\centering{}\protect\caption{\label{fig:Fomenkopi2sing}Singular IEMBD and the correspondning two singular IFG. The IFGs correspond to  Figure \ref{fig:ps-alphapi2}c,d respectively, and to the lines separating regimes D and E and B and C respectively.}
\end{figure}
\subsection{\label{sec:proofofmainintegrb}Proofs of the main Theorems \ref{thm:integrability} and  \ref{thm:nearintegrability}}

\noindent\textbf{\\ Proof of theorem \ref{thm:integrability}:
}

\begin{proof}In lemma \ref{lem:integrabilityonly} we proved that the PWS are Liouville integrable HIS. Proposition \ref{lem:hetanglevel1} shows that for any energy surface, the IFG of these systems labels which branches of the level sets impact the wall while the bifurcation curves in the IEMBD provide the extent of each branch, so the impact-division  follows from the construction. The proof was presented for the horizontal wall, $q^{w}_{\alpha=0}$ case,  and the   same results hold for the vertical wall, $q^{w}_{\alpha=\pi/2}$, with interchanging the indices \(1\leftrightarrow2\) and the words ``lower half plane" with ``right half plane".
 Notice that the measure of each of the impact sets may be found since the dynamics is Liouville integrable.

Finally, Figures \ref{fig:Fomenkopi2out} and \ref{fig:Fomenkopi2sing}  provide a construction which demonstrates that the smooth EMBD and the FG of a Hamiltonian of the form (\ref{eq:Hint}) may be insufficient for describing even the topology of the energy surfaces for the PWS.  Indeed,  the FG and the EMBD shown in  Figure \ref{fig:EMBD-integrable} do not depend on the values of \((q_{1s},q_{2c})\), and show that for all parameter values for small energies the energy surface has two connected components and each level set has two leafs whereas for large energies the energy surface has a single component and the level sets have one leaf for small \(H_{2}\) and two leaves for large \(H_{2}\). On the other hand, for the PWS with   $q^{w}_{\alpha=\frac{\pi }{2}}$, we see that when \(q_{1s}<-1\), namely  in region A of  Figure \ref{fig:Fomenkopi2out}, the energy surface of the PWS has, for all energies, a single connected component, and on each level set it has a single leaf for all energy values.  Figure \ref{fig:Fomenkopi2sing} demonstrates  that the IEMBD by itself may be  insufficient for describing the topology of the impact level sets - the figure presents two different IFG for the same IEMBD - showing that the energy surfaces are topologically different and have topologically different foliation.

\end{proof}

\noindent\textbf{Proof of theorem \ref{thm:nearintegrability}:
}\begin{proof}
For sufficiently small \(\epsilon\), away from the tangent tori, the non-impact set is foliated by KAM tori that are \(\epsilon_{r}\) close to the unperturbed tori, so the open region corresponding to non-impact tori remains invariant for sufficiently small \(\epsilon\). Similarly, by choosing a proper local cross-section in the interior of the billiard domain, it is proven in    \cite{pnueli2018near} that  for sufficiently small \(\epsilon\) a Poincare return map near a regular  unperturbed torus that is in the transverse impact unperturbed region (away from tangencies and separatrices and satisfying generic non-resonance and twist condition) is a smooth near integrable twist map, so KAM theory can be applied to this map (the smoothness requirements in the GWS definition guarantee that the return map is \(C^{r}\) smooth with r\textgreater3). Since for 2 d.o.f. systems KAM tori divide the energy surfaces, for sufficiently small \(\epsilon\), there exist  KAM tori bounding the non-impact regions and KAM tori bounding the transverse regions from a small region around the tangent tori, where mixed dynamics occurs.      \end{proof}

\subsection{Multiple vertical and horizontal walls}\label{sec:multiwall}

When the billiard boundary is composed of any combination of horizontal and vertical segments,  the S3BN Hamiltonian flow impacts with the boundary preserve the partial energies, so condition RespF of Definition \ref{def:liouvilimpint} is fulfilled. Namely, the \textit{impact rule} preserves the separability symmetry.
When the billiard boundary is composed of a finite number of infinite vertical and horizontal lines,  namely the  \textit{billiard boundary} also respects the separability symmetry, we propose that the Resp\(\theta\) condition is also fulfilled and the motion in the allowed region of motion is Liouville integrable. This can be proved by using the same construction as in  \cite{Issi2019}, showing that on each regular level set the motion is conjugated, via the angle variables, to directional motion on a finite number of rectangles that are glued in such a way that it is possible to tile with them the plane.

On the other hand, there are cases in which the impact rule respects the separable symmetry yet the billiard boundary does not. Then, the Resp\(\theta\) condition might be violated. For example, when the billiard table of the HIS is exterior to a corner (or, more generally,  is constructed from finite or semi-infinite horizontal and vertical segments having corner points which are larger than \(\pi/2\) as in the billiard tables in \cite{Athreya2012}), the billiard boundary definition involves both coordinates. Then, the dynamics for some level sets is  conjugate to directional motion on a compact, oriented surface of genera larger than 1 \cite{Issi2019}, return maps on this surface produce interval exchange maps, and the motion, in general, is not Liouville integrable. Billiards with such behavior are called quasi-integrable  \cite{Athreya2012,Dragovic2014a,Dragovic2014,Dragovic2015,Dragovic2015a,Moskvin2018}.

The construction of the IEMBD may be easily extended to both the Liouville IHIS (Figure  \ref{fig:IEMBDcorner}a) and the IHIS which are not LIHIS~(Figure  \ref{fig:IEMBDcorner}b).  The bifurcation set includes as many tangency rays as the number of horizontal or vertical segments. The starting points of the rays and the identification of the EMBD regions in which impact is made  with each of the walls are determined by the level set structure of  \(H_{i}\).

 Figure \ref{fig:IEMBDcorner}a shows the IEMBD of the Hamiltonian (\ref{eq:doublewell}) when the motion is confined to the upper quadrant of the \((q_{1},q_{2})\) plane  and both the \(H_{2}\) center and the separatrix loop are in the upper quadrant  \((q_{2c}>0,q_{1s}>\sqrt{2})\). Level sets on which impacts occur with only one of the walls or with both walls are marked on the IEMBD, and, since the motion on regular leaves is always rotational, the corresponding IFG may be  defined in a similar manner to PWS (here it may be beneficial to distinguish impacts with different walls if different types of perturbations are expected).

Figure \ref{fig:IEMBDcorner}b shows the IEMBD of the Hamiltonian (\ref{eq:doublewell}) (with the same   \(q_{2c}>0,q_{1s}>\sqrt{2}\)) when  the billiard boundary is a combination of two semi-infinite walls lying on the negative side of the $q_1$- and $q_2$-axes.  Here, impacts cannot occur with a single wall - to hit the corner both subsystems need to have sufficient energy (the partial energies must belong to the wedge between the two tangency rays). For energies in this wedge, the motion on each level set is conjugated to the directional motion on an L-shaped billiard with changing dimensions and direction, or, equivalently, to the directional motion on a flat surface of genus 2  \cite{Issi2019}.   The resulting IFG and their classification for the various cases are yet to be developed (the  graphs of topological billiards appear to be relevant  \cite{fomenko2019singularities,fomenko2019topological,Moskvin2018}).

\begin{figure}[h]
\begin{centering}
\includegraphics[scale=0.41]{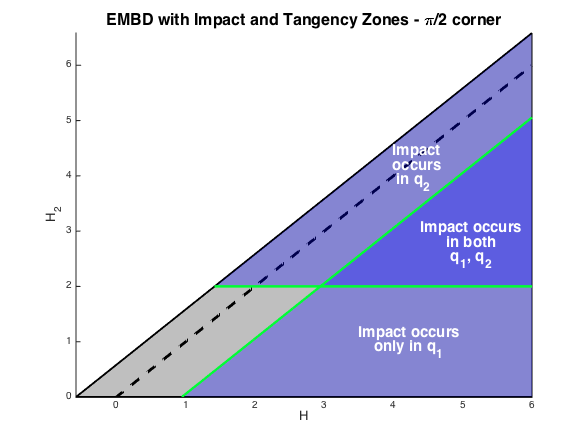}
\includegraphics[scale=0.42]{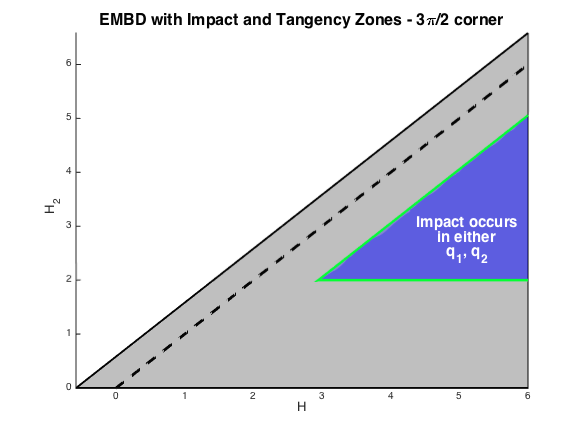}
\par\end{centering}
\protect\caption{\label{fig:IEMBDcorner}IEMBD for motion in (a) The upper quadrant (b) The complement to  the lower quadrant of the \((q_{1},q_2)  \) plane, with elastic reflection from the walls. Impacts with  one wall (light blue) with both walls (dark blue), tangent (green) and no-impacts (grey) zones are depicted. }
\end{figure}

\section{Impacts with a general wall }\label{sec:Hill-region}

When a particle impacts a smooth wall which is not aligned with one of the symmetry axes, energy transfers between the two d.o.f. and thus even when \(\epsilon_{r}=0\) the level sets and the corresponding cut leaves are not invariant under impacts.  The IEMBD and the IFG are  used to distinguish between\textbf{ leaves} that do not impact the wall (these remain invariant when \(\epsilon_{r}=0\)), those which, at first impact, must impact the wall transversely, and those which may touch the wall tangentially. The impact and tangential leaves are, in general, not invariant,  so the classification applies to initial conditions belonging to the corresponding cut leaves (and not to the full trajectories, which jump to other cut-leaves after  impact). By analyzing the structure of the Hill region for the S3BN  systems we show that for a wall in a general position (see definition \ref{def:wallgeneralposition}) the tangent zone becomes non-trivial - it does not degenerate to a line as in the horizontal and vertical wall cases.

\begin{prop}
\label{thm:global} Consider a GWS of the form (\ref{eq:wavywal}) at \(\epsilon _{r}=0\) which is in general position (satisfies Eq. (\ref{eq:generalpositionwall})).   Then, there exist finite  $H$ values, $H^{tmin}_{min}<H^{tr}_{min}$  (given by Eqs. (\ref{eq:htminimin}) and (\ref{eq:htrmin}) respectively), and $H_2$ intervals $H_{2}^{tan}(H)\subseteq H_2^{m}(H),$ (defined by Eq. (\ref{eq:tanghh2}) and (\ref{eq:h2mdef})), such that: \begin{enumerate}
\item
An exterior GWS has no allowed motion for \(H<H^{tmin}_{min}\) and, for  all  \(H>H^{tmin}_{min}\), all the  leaves with segments of non-zero length in the allowed region of motion are  impact cut-leaves.
\item For an interior GWS, for all  \(H<H^{tmin}_{min}\)  the allowed region of motion is non-empty
  iff at at least one of the local minimizers of    \(V\) which is in the billiard domain \(V\) is smaller than \(H^{tmin}_{min}\). Then, all the allowed leaves are in the non-impact zone.   \item For \(H>H^{tmin}_{min}\) the  allowed leaves of the level set \((H_{1}=H-H_{2},H_{2})\)
belong to the tangency zone iff \(H_{2}\in H_{2}^{tan}(H)\) and to the  impact zone iff  \( H_2\in H_2^{m}(H)\).
\item For all   \(H>H^{tmin}_{min}\) the set  \(H_{2}^{tan}(H)\) is a finite collection of segments and it has a positive length.

\item  For  \(H>H^{tr}_{min}\) the  set of \(H_{2}\) values belonging to the iso-energy transverse impact  zone   \( (H_{2}^{tr}(H)=H_2^{m}(H)\setminus H_{2}^{tan}(H)\)) has positive measure, in fact, there exist  \(H_{2}^{lower-tr}(H)>0\) such that this set   includes the interval \([0,H_{2}^{lower-tr}(H)]\).
\end{enumerate}
 \end{prop}

\noindent\textit{
Proof Outline:}
In sections  \ref{sec:hillfoliation} we develop the tools needed for proving this proposition: we establish that the smooth motion on iso-energy level sets projects to disjoint rectangles in the configuration space and that each rectangle corresponds to a leaf of the level set (Proposition \ref{thm:hillregionprl}). We then study how the wall intersects the Hill region and the rectangular projection of the leaves as a function of \(H \),   and in lemmas \ref{lem:claim1}-\ref{lem:claim5} we use these observations to prove  claims 1-5.
Notably, in claims 3-5, the tangent level sets  include cases of external tangencies, where the tangent segments are not in the billiard domain (as in convex billiards).
In section \ref{sec:interiortangentsegments} we study when tangent segments are in the billiard domain and utilize the above proposition and these results to prove Theorems \ref{thm:nonimpactgen}-\ref{thm:tangslantdc}.

\subsection{The Hill region foliation for separable systems} \label{sec:hillfoliation}

 We study  first how  the smooth motion projects to the configuration space, the space at which impacts are defined.

\begin{defn}
\cite{Arnold2007CelestialMechanics} The Hill region of a smooth Hamiltonian system is the
allowed region of motion in the configuration space. \end{defn}

 For a general mechanical 2 d.o.f. system the level sets of \(V\) determine the Hill region geometry ($\mathcal{D}_2^{Hill}(H)=\{(q_{1},q_{2})\in\mathbb{R}^{2}:V(q_{1},q_{2})\leq H\})$. If the motion is not ergodic on the energy surface, typical orbits may project to subsets of the Hill region. In particular, when the motion is integrable, the energy surface foliation induces specific projections to the configuration space. As proved below, in the separable setup (S3BN systems),  the projection of the foliation is very simple - the  Hill region is foliated by rectangles in  the configuration space, where each rectangle corresponds to a leaf of an iso-energy  level set. We thus define:

\begin{defn}\label{def:PRL} The \textit{Projected Rectangle of a Leaf} of a S3BN Hamiltonian is the projection of the leaf to  the configuration space. \end{defn}

\begin{mainthm}
\label{thm:hillregionprl}\textit{ For the S3BN systems the Hill region \(\mathcal{D}_2^{Hill}(H)\) is foliated by a collection of  rectangles, the Projected Rectangles of Leaves (PRLs), denoted by \(R^{k}(H_1,H_2)\): \begin{equation}
\mathcal{D}_2^{Hill}(H)=\bigcup_{H_2\in[V_{2,min},H-V_{1,min}]} \biguplus _{k_{i}=1,...,K(H_{i})}R^{(k_{1},k_2)}(H_{1},H_{2})|_{H_{1}=H-H_{2}}\label{eq:dhillprl}
\end{equation}  where \(K(H_{i})\) denotes  the number of  the \(H_i\)-Liouville leaves.
The rectangles belonging  to the same level set are disjoint: \(R^{k}(H_{1},H_2)\cap R^{m}(H_{1},H_2)=\emptyset\) for all \(k\neq m\).}   \end{mainthm}

\begin{proof}
Since the smooth motion is separable, \(\mathcal{D}_2^{Hill}(H)=\bigcup_{H_1+H_2=H}\mathcal{D}^{Hill}_1(H_{1})\times\mathcal{D}_1^{Hill}(H_{2})\), where   \(\mathcal{D}_1^{Hill}(H_{i})\) denotes   the Hill region of the one d.o.f. mechanical systems \(H_{i}\). Under the S3BN assumption, for all \(H_{i}\geq H_{i,min}\), the region \(\mathcal{D}_1^{Hill}(H_{i})\)  is composed of a finite collection of closed disconnected intervals,
 \([q^{k_i}_{i,min}(H_{i}),q^{k_i}_{i,max}(H_{i})],k_i=1,..,K_i(H_i)\),
  which are the maximal intervals on which \(V_i(q_{i})\leq H_i\) for all \(q_{i}\in [q^{k_i}_{i,min}(H_{i}),q^{k_i}_{i,max}(H_{i})]\)  (if \(V_{i}(q)\) has a local maximum it is an interior point in this interval, so \(q^{k_i}_{i,min}(H_{i})=q^{k_i}_{i,max}(H_{i})\) iff \(q^{k_i}_{i,min}(H_{i})=q^{ext}_{i,k_{i}}:=q^{min}_{i,k_{i}}\) is a local minimizer of \(V_i\)).
More generally, on the singular level sets, where \(H_{i}=V_i(q^{ext}_{i,k_{i}}), \) the \(k_i\) index changes as at this critical energies the leaf structure changes.
So, with each branch of PRLs, \(R^{(k_{1},k_2)}\) there is a rectangle of  \((H_{1},H_2)\) values on which this family is supported. Above the local maxima of \(V_{i}\), i.e. for all \(H_{i}>H_{i,max}=\max_{k_{i}} V_i(q_{i,k_{i}}^{ext})\), by the S3BN assumption, there is a single finite length interval (\(K_i(H_i)=1\)),   the \(H_{i}\) range on which the corresponding PRL family is defined becomes infinite and the smallest PRL belonging to this family is determined by the \(H_{i}\)-level set of \(H_{i,max}\).

For lower partial energies, when \(K_i(H_i)>1\),  by the S3BN assumption, the segments are finite and disjoint.
  By the mechanical form of \(H_i\), each of the  intervals \([q^{k_i}_{i,min}(H_{i}),q^{k_i}_{i,max}(H_{i})]\) is a projection of a single leaf of the level set of \(H_i\), namely there is one-to-one correspondence between all the \(H_{i}\)-leaves and the intervals \([q^{k_i}_{i,min}(H_{i}),q^{k_i}_{i,max}(H_{i})]\) .

Since the Liouville leaves of  the level set \((H_{1},H_2)\) correspond to a product of the individual \(H_{i}\) leaves, for all finite \(H_{i}\geq H_{i,min},i=1,2\), the level set  \((H_{1},H_2)\)  has  a finite number of leaves,  \(K_{1}(H_1)K_2(H_2)\geqslant1\). The projection of  the \(k=(k_{1},k_2)\)  Liouville leaf of this level set to the configuration space is the product of the two corresponding intervals, namely the rectangle :   $$R^{k}(H_{1},H_{2})=[q^{k_{1}}_{1,min}(H_{1}),q^{k_{1}}_{1,max}(H_{1})]\times[q^{k_{2}}_{2,min}(H_2),q^{k_{2}}_{2,max}(H_2)]$$where \(k_i=1,..,K_i(H_i)\).  All distinct  rectangles belonging to the same level set are disjoint since the intervals  \([q^{k_i}_{i,min}(H_{i}),q^{k_i}_{i,max}(H_{i})]\) are all disjoint. Setting, on a given energy surface, \(H_{1}=H-H_2\), and letting \(H_{2}\) vary between its minimal to its maximal allowed value,  formula (\ref{eq:dhillprl}) is established.
\end{proof}

 Figure \ref{fig:Hill-rectangles} shows trajectories of the Duffing-Center integrable system in the configuration space. The trajectories are  non-resonant and thus densely fill rectangles inside the Hill region (the boundary of the Hill region is marked in black) - these are the  PRLs inside the Hill region. Notice that there are exactly two flow directions passing through each interior point in a given  PRL  (as \(q_i\) and \(H_i\) uniquely define \(|p_i|\)) and that the PRLs structure reflects the symmetries of the potentials (\((q_{1}-q_{1s})\rightarrow-(q_{1}-q_{1s})\) and \((q_{2}-q_{2c})\rightarrow-(q_{2}-q_{2c})\)) and of the mechanical form of the Hamiltonian.
While the PRL of leaves belonging to the same level set are disjoint, iso-energy PRL do overlap since their boundaries, set by $q^{k}_{i,min/max}(H_i)$, are piecewise continuous in $H_i$, and on an energy surface \(H_{1,2}\) vary on  the intervals of the allowed region of motion.

\subsection{Walls intersecting the Hill region}

Next, we study how the collection of all the PRLs corresponding to a given energy surface intersect the wall.
 Let \(q_w(q_{2})=(\epsilon _{w}Q(q_{2}),q_2)\) denote the wall (\ref{eq:wavywal}) parameterization by \(q_2\), and denote the intersection of the wall with the Hill region by \(S_{w}(H)\):\begin{equation}\label{eq:swofH}
S_{w}(H)=\{q | q=q_w(q_{2}) \text{ and } q_{w}(q_{2})\in{\mathcal{D}_2^{Hill}(H)}\}.
\end{equation}
Since the Hill regions are nested for increasing \(H\), so do \(S_{w}(H)\) and the projection of \(S_{w}(H)\) to the \(q_{2}\) axis, \(L_{w}(H)\).
  By the S3BN assumption, for any fixed \(H\), \(S_{w}(H)\) consists of  a finite collection of  closed  segments with disjoint interiors, \(S_{w}^{j}(H)\) : \(S_{w}(H)=\bigcup_{j=1,..,k_{w}(H)}S_{w}^{j}(H)\).
Let  \(\partial S_{w}(H)=\{q_{w}^{j}( H)\}_{j=1,..,2k_w(H)}=\{(q_{w,1}^{j}( H),q_{w,2}^j(H))\}_{j=1,..,2k_w(H)}=\{(\epsilon _{w}Q(q_{w,2}^{j}( H)),q_{w,2}^{j}( H))\}_{j=1,..,2k_w(H)}={\mathcal{\partial D}_2^{Hill}(H)\cap S_{w}(H)} \) denote all the boundary points of these segments (for ease of notation,  tangent points with the wall are counted twice).

Next we study the nature of impacts at  \(q_{w}(q_{2})\in S_{w}(H)\).
 \begin{lem}
\label{lem:prlwall}A  point on the wall, \(q_{w}(q_{2})\), belongs to the Hill region \(\mathcal{ D}_2^{Hill}(H)\) iff  \(H>H^{tmin}(q_{2})\), where
\begin{equation}\label{eq:htmingen}
H^{tmin}(q_{2})=V_{1}(\epsilon _{w}Q(q_2))+V_2(q_2).
\end{equation} For    \(H>H^{tmin}(q_{2})\), the wall point \(q_{w}(q_{2})\) is an interior point of  \(S_{w}(H)\). It is an interior point of a PRL of the iso-energy level set \((H-H_2,H_2)\) iff \(H_2\in\mathcal{H}_{2}(q_2,H)=( V_2(q_2),H-V_{1}(\epsilon _{w}Q(q_2)))\) and belongs to such a PRL boundary  iff \(H_{2}\in\partial\mathcal{\bar H}_{2}(q_2)\), where \(\mathcal{\bar H}_{2}(q_2)\) denotes the closure of \(\mathcal{ H}_{2}(q_2)\).     \end{lem}
\begin{proof}By the separable mechanical form of \(H,H_1\) and \(H_2\),  each partial energy of a trajectory  reaching  \(q_{w}(q_{2})\)  must be larger or equal to the  potential energy at  \(q_{w}(q_{2})\). Since, by definition, at   \(S_{w}(H)\) boundary points the kinetic energy vanishes, \(H>H^{tmin}(q_{2})\)  necessarily means that  \(q_{w}(q_{2})\)  is  an interior point of  \(S_{w}(H)\).   Since   \(H-V_{1}(\epsilon _{w}Q(q_2)= V_2(q_2)+H-H^{tmin}(q_{2})\) it is clear that both kinetic energies are strictly positive for all \(H_2\in\mathcal{H}_{2}(q_2)\), so   \(q_{w}(q_{2})\) is also an interior point of a PRL.  When \(H_2 =V_2(q_2)\) the vertical momentum vanishes,  \(p_2=0\), so the wall point belongs to a horizontal boundary of the PRL and when  \(H_2 =H-V_{1}(\epsilon _{w}Q(q_2))\),  \(p_1=0\), and the wall point belongs to a  vertical boundary of a PRL.
\end{proof}

The minimal energy at which impacts occur is\begin{equation}\label{eq:htminimin}
H^{tmin}_{min}=\min_{q_2}H^{tmin}(q_{2})=H^{tmin}(q_{2}^{w-min}),
\end{equation}
and the wall is tangent to the Hill region of  \(H^{tmin}_{min}\) at  \(q_{w}(q_{2}^{w-min})\), where  \(q_{2}^{w-min}\) is the minimizer of (\ref{eq:htmingen}) (indeed \((V_{1}'(q_1),V_2'(q_2))\cdot(\epsilon _{w}Q'(q_2),1)=V_{1}'(q_1)\epsilon _{w}Q'(q_2)+V_2'(q_2)=\frac{d}{dq_{2}}H^{tmin}(q_{2})\) vanishes at \(q_{w}(q_{2}^{w-min})\)).
 For walls in general position, since the minima of \(V\) is not on the wall, \( H^{tmin}_{min}>H^{min}\).

Let\begin{equation}\label{eq:tanghh2def}
H_{2}^{t}(q_{2,}H)=V_{2}(q_2)+\frac{H-H^{tmin}(q_{2})}{1+(\epsilon _{w}Q'(q_{2}))^{2}}.
\end{equation} \begin{lem}
\label{lem:prlwalltrans}  For   \(H>H^{tmin}(q_{2})\), the impact at \(q_{w}(q_{2})\) on the iso-energy level set \((H-H_2,H_2)\) is transverse if \(H_{2}\neq H_{2}^{t}(q_{2,}H)\) and has a tangency when   \(H_{2}= H_{2}^{t}(q_{2,}H)\).  The level set \((H-H_{2}^{t}(q_{2,}H),H_{2}^{t}(q_{2,}H))\) has also a  transverse impact segment  at \(q_{w}(q_{2})\) iff  \(\epsilon _{w}Q'(q_{2})\neq0\).     \end{lem}\begin{proof}
  By lemma \ref{lem:prlwall}, impacts occur at \(q_w(q_2)\) on  a leaf of  \((H-H_2,H_2)\) iff \(H_2\in\mathcal{\bar H}_{2}(q_2,H)\), so we need to  consider only this range of \(H_{2}\)  values.
The PRL with \(H_{2}=V_2(q_2)\) is the PRL with a horizontal boundary passing through \(q_w(q_2)\) (since \(p_{2}=0\) at  \(q_w(q_2)\)). The PRL with  \(H_{2}=H-V_1(\epsilon _{w}Q(q_2))\) is  the PRL with a vertical boundary passing through \(q_w(q_2)\) (since \(p_{1}=0\) there), see lemma \ref{lem:prlwall} (notice that these two PRLs are not necessarily of the same family). All the intermediate values of \(H_{2}\) are realized at \(q_w(q_2)\) by PRLs that contain \(q_w(q_2)\) as an interior point with non-zero momenta, in the two directions:\begin{equation}
\frac{p_{1}}{p_2}\vert_{q_w(q_2)}=\pm\frac{\sqrt{H-H_2-V_1(\epsilon _{w}Q(q_2))}}{\sqrt{H_{2}-V_2(q_2)}}
\end{equation}  So,   these directions change continuously and monotonically with \(H_{2}\) between \(0\) and  \(\pm\infty\):\begin{displaymath}
\frac{d}{dH_{2}}(\frac{p_{1}}{p_2}\vert_{q_w(q_2)})=\pm\frac{\sqrt{H_{2}-V_2(q_2)}}{\sqrt{H-H_2-V_1(\epsilon _{w}Q(q_2))}}\frac{H^{tmin}(q_{2})-H}{(H_{2}-V_2(q_2))^2}
\end{displaymath}   The tangent vector to the wall at \(q_w(q_2)\) is \(\frac{(\epsilon _{w}Q'(q_{2}),1)}{\sqrt{1+(\epsilon _{w}Q'(q_{2}))^{2}}}\) and the normal into the billiard domain is \(\frac{(1,-\epsilon _{w}Q'(q_{2}))}{\sqrt{1+(\epsilon _{w}Q'(q_{2}))^{2}}}\) (see Eq.  (\ref{eq:wavywal})).  So, \(p_{1}/p_2\) has exactly one intermediate value \(H^{t}_2(q_{2},H)\)  at which \(p_{1}=\epsilon _{w}Q'(q_{2})p_{2}\), where tangency occurs, and at all other \(H_{2}\) values the impact is non-tangent.  It follows from eqs. (\ref{eq:Hint},\ref{eq:wavywal}) that at  \(H^{t}_2(q_{2},H)\)  :
\begin{equation}
H= H^{t}_2(q_{2},H)  (1+(\epsilon _{w}Q'(q_{2}))^{2})-(\epsilon _{w}Q'(q_{2}))^{2}V_2(q_{2})+V_{1}(\epsilon _{w}Q(q_2))
\end{equation}Hence
\begin{equation}\label{eq:tanghh2}
H_{2}^{t}(q_{2,}H)=\frac{H+(\epsilon _{w}Q'(q_{2}))^{2}V_{2}(q_2)-V_{1}(\epsilon _{w}Q(q_2))}{1+(\epsilon _{w}Q'(q_{2}))^{2}} =V_{2}(q_2)+\frac{H-H^{tmin}(q_{2})}{1+(\epsilon _{w}Q'(q_{2}))^{2}}.
\end{equation}
For the level set  \((H-H_{2}^{t}(q_{2,}H),H_{2}^{t}(q_{2,}H))\), at \(q_w(q_2)\), the  momenta directions are \(\pm(\epsilon _{w}Q'(q_{2}),\pm1) \), where, as noted above,    \((\epsilon _{w}Q'(q_{2}),1) \)  is tangent to the boundary.   The directions  \(\pm(-\epsilon _{w}Q'(q_{2}),1) \)  are distinct from the tangent one iff \(\epsilon _{w}Q'(q_{2})\neq0\) and in this case, since \(\pm(-\epsilon_{w}Q'(q_{2}),1)\cdot\hat  n=\mp2\epsilon _{w}Q'(q_{2})\), they correspond to transverse impacts at  \(q_w(q_2)\) with one direction pointing into the billiard domain.    \end{proof}

In particular, \(H_{2}^{t}(q_{2},H^{tmin}(q_{2}))=V_2(q_2)\) (see (\ref{eq:htmingen})),   \(H_2^t(q_2,H)\) grows linearly with \(H\), and   \(H_2^t(q_2,H)\)  depends smoothly on \(q_2\) for all \(H>H^{tmin}(q_{2})\), namely for all \(q_2\)  such that \(q_w(q_2)\in S_w^{j}(H)\backslash \partial S_{w}^{j}(H)\) for some \(j\). When  \(\epsilon _{w}Q'(q_{2})=0 \), \(H_{2}^{t}(q_{2,}H)=H-V_{1}(\epsilon_{w}Q(q_2))\), so   \(q_w(q_2)\) belongs to the vertical boundary of the corresponding PRL (see lemma \ref{lem:prlwall}).

 \begin{lem}
\label{lem:prlintrfixedpoint}Each PRL of a level set \((H_{1},H_2)\) includes a local minimizer of \(V\) which is an interior point of the PRL (for singular PRL that reduce to a segment, an interior point means an interior point of the segment).   \end{lem}
\begin{proof}
 This follows from the fact that the \(H_i\)-leaves, namely the segments \([q^{k_i}_{i,min}(H_{i}),q^{k_i}_{i,max}(H_{i})]\), correspond to the maximal extent of the level sets of \(V_i\) (recall the proof of Theorem \ref{thm:hillregionprl}),  which, by the S3BN assumption, must include at least one local minimizer of \(V_{i}\) in their interior.  \end{proof}

\subsection{Proof of Proposition \ref{thm:global} }\label{subsec:impactzones}

First, we note that by the S3BN assumption \(H^{tmin}(q_2)\) (of Eq. (\ref{eq:htmingen})) is bounded from below and grows to infinity for large \(|q_{2}|\) thus its minimum, \(H^{tmin}_{min}\) (Eq. (\ref{eq:htminimin})) is finite. Let
\begin{equation}\label{eq:htrmin}
H^{tr}_{min}=\max _{i}V_{1}(\epsilon _{w}Q(q_{2,i}^{min})), \text{where } V_2 (q_{2,i}^{min})=\min_{q_2} V_2(q_2)=0
\end{equation}denote the maximal value of \(V_{1}\) at the wall points corresponding to global minimizers of \(V_2\). By the S3BN assumption there is a finite number of such minimizers of \(V_{2}\) and they occur in some bounded set (since \(V_2\) grows to infinity at large \(|q_2|\)). Thus,  since \(Q\) and \(V_1\) are  smooth,
\(H^{tr}_{min}\) is bounded.\begin{lem}
(claim 1 of Proposition \ref{thm:global})\label{lem:claim1} An exterior GWS has no allowed motion for \(H<H^{tmin}_{min}\) and, for  all  \(H>H^{tmin}_{min}\), all the  leaves with segments of non-zero length in the allowed region of motion are  impact cut-leaves.
\begin{proof}
  Since each PRL  is a rectangle which includes in its interior a  point which is a local minimizer of \(V\)(by lemma \ref{lem:prlintrfixedpoint}), and since for exterior GWS all local minimizers are not in the billiard domain, we conclude that each PRL  which includes  a point inside the billiard domain must also include a minimizer of \(V\) which is not in the billiard domain, so  it must intersect the wall, namely the leaf is an impact leaf as claimed.
 Since, by lemma \ref{lem:prlwall} and Eq. (\ref{eq:htminimin}), for  \(H<H^{tmin}_{min}\) all the iso-energy leaves do not intersect the wall, by the above argument they all must lie outside of the billiard domain so no motion is allowed for   \(H<H^{tmin}_{min}\).

 \end{proof}  \end{lem}

 \begin{lem}
\label{lem:wallfixedpoints}(claim 2 of Proposition \ref{thm:global})\label{lem:exteriorPRL} For an interior GWS, for all  \(H<H^{tmin}_{min}\)  the allowed region of motion is non-empty
  iff at at least one of the local minimizers of    \(V\) which is in the billiard domain \(V\) is smaller than \(H^{tmin}_{min}\). Then, all the allowed leaves are in the non-impact zone.

   \end{lem}

\begin{proof} Since, by lemma \ref{lem:prlwall} and Eq. (\ref{eq:htminimin}), for  \(H<H^{tmin}_{min}\) all the iso-energy leaves do not intersect the wall, the leaves in the allowed region of motion belong to the non-impact zone.
We need to check when some of these non-impacting leaves are in the allowed region of motion.

 \(\Rightarrow\)If \((q_{1,i}^{ext},q_{2,i}^{ext})\)   is a minimizer of \(V \) in the billiard domain such that  \(H^{min,i}=V(q_{1,i}^{ext},q_{2,i}^{ext})<H^{tmin}_{min}\), then, for small \(\delta\), the PRL of \(H^{min,i}+\delta\) are small rectangles that include the minimizer and are thus, for sufficiently small \(\delta\) insider the billiard domain.

\(\Leftarrow\) If all the local minima with minimizers in the billiard domain are larger than \(H^{tmin}_{min}\), then all the PRLs that include these minimizers as an internal point belong to iso-energy level sets with energy larger than  \(H^{tmin}_{min}\). Since, by lemma \ref{lem:prlintrfixedpoint}, all PRL must include a minimizer, we conclude that all the level sets with      \(H<H^{tmin}_{min}\)  include minimzers that are outside the billiard domain, and since they do not impact the wall they have no part which is in the allowed region of motion. \end{proof}

 \begin{lem}\label{lem:h2mh2tan}(Claim 3 of Proposition \ref{thm:global}:)
For \(H>H^{tmin}_{min}\) the  allowed leaves of the level set \((H-H_{2},H_{2})\)
belong to the  impact zone iff  \( H_2\in H_2^{m}(H)\) and to the tangency zone iff \(H_{2}\in H_{2}^{tan}(H)\) where
\begin{equation}
\label{eq:h2mdef}
\begin{array}{ll}
H_{2}^m(H) &= \displaystyle \bigcup_{j=1,...,k_{w}(H)}\ [\min _{q_w(q_2)\in S_{w}^{j}(H)}V_2(q_2),\max_{q_w(q_2)\in S_{w}^{j}(H)} (H-V_1(\epsilon _{w}Q(q_2))]
\end{array}
\end{equation} and\begin{equation}\label{eq:h2tanminmaxform}
H_{2}^{tan}(H)=\bigcup_{j=1,...,k_{w}(H)}[\min_{q_w(q_2)\in S_w^{j}(H)}H_{2}^{t}(q_{2,}H),\max_{q_w(q_2)\in S_w^{j}(H)}H_{2}^{t}(q_{2,}H)].
\end{equation}
 \end{lem}
\begin{proof}
For   \(H>H^{tmin}_{min}\)   the set of collection of wall intervals, \( S_{w}(H)\), has non-empty interior. We check when a transverse/tangent impact occurs at an interior point  \(q_w(q_2)\) of  \( S_{w}(H)\).  Taking the union on all such level sets leads to the definition of \(H_{2}^{m}(H)\) and  \(H_{2}^{tan}(H)\).
By lemma \ref{lem:prlwall}, for   \(H>H^{tmin}_{min}\)    the union of \(H_{2 }\) intervals of level sets with PRLs which contain a wall point is precisely (see Eq. \ref{eq:swofH}):
\begin{equation}
\label{eq:h2mdef}
\begin{array}{ll}
H_{2}^m(H) &=  \displaystyle\bigcup_{j=1,...,k_{w}(H)} \ \bigcup_{q_w(q_2)\in S_{w}^{j}(H)}\mathcal{\bar H}_{2}(q_2,H) \\ \\
 &= \displaystyle\bigcup_{j=1,...,k_{w}(H)} \ \bigcup_{q_w(q_2)\in S_{w}^{j}(H)}[V_2(q_2),H-V_1(\epsilon _{w}Q(q_2))]  \\ \\
 &= \displaystyle \bigcup_{j=1,...,k_{w}(H)}\ [\min _{q_w(q_2)\in S_{w}^{j}(H)}V_2(q_2),\max_{q_w(q_2)\in S_{w}^{j}(H)} (H-V_1(\epsilon _{w}Q(q_2))]
\end{array}
\end{equation}
where the third line follows from the continuity of \(V_{i}\). Namely, \(H_{2}^m(H)\) is composed of at most \(k_{w}(H)\)  intervals, where  \(k_{w}(H)\)  is the number of disjoint intervals of \( S_{w}(H)\). By construction, for all  \(H_{2}\in H_{2}^m(H)\), there exists at least one \(q_2\in\mathbb{R}\) such that \(q_w(q_2)\in S_{w}(H)\) and \((H_{2}-V_2(q_2)),(H-H_{2}-V_1(\epsilon _{w}Q(q_2))\geqslant0 \), namely the level set \((H-H_{2},H_2)\) has impacting leaves (some of which are tangent leaves and possibly tangent boundary leaves that consist of a single point on the wall).

By lemma \ref{lem:prlwalltrans} a tangency occurs at  \(q_w(q_2)\in S_{w}(H)\)   on a level set  \((H-H_{2},H_2)\) iff  \(H_{2}=H_{2}^{t}(q_{2,}H)\). By  Eq. (\ref{eq:tanghh2}) \(H_{2}^{t}(q_{2,}H)\) depends smoothly on \(q_{2}\) iff   \(q_w(q_2)\in S^{j}_{w}(H)\)    (as \((p_{1},p_2)\) are real iff   \((H_{2}-V_2(q_2)),(H-H_{2}-V_1(\epsilon _{w}Q(q_2))\geqslant0 \)), so Eq. (\ref{eq:h2tanminmaxform}) follows.
\end{proof}

\begin{lem}\label{lem:claim4}(Claim 4 of Proposition \ref{thm:global}:)
For all   \(H>H^{tmin}_{min}\) the set  \(H_{2}^{tan}(H)\) is a finite collection of segments and it has a positive length.\end{lem}
\begin{proof} Recall that for all   \(H>H^{tmin}_{min}\)  the set \(S_{w}(H)\) is composed of closed segments with non-empty interior (see Eq. (\ref{eq:swofH})), and their projection to the \(q_{2}\) axis is denoted by \(L_{w}(H)\). By  Eq. (\ref{eq:tanghh2}) and  Eq. (\ref{eq:h2tanminmaxform}), to each wall point \(q_{w}(q_2)\in S_{w}(H) \) the iso-energy level set   \((H-H_{2}^{t}(q_{2,}H),H_{2}^{t}(q_{2,}H))\) has a tangent point at \(q_{w}(q_2)\), where \(H_{2}^{t}(q_{2,}H)\) depends smoothly on its arguments as long as \(H>H^{tmin}(q_{2})\), which is always the case for all  \({q_2\in L_{w}(H) }\). So, to prove the claim we need to prove that for all \(H>H^{tmin}_{min}\) the function \(H_{2}^{t}(q_{2,}H)\) is not a constant on \(L_{w}(H)\) which implies that it maps at least some of the  intervals of \(L_{w}(H)\) to intervals of positive lengths. Since  the intervals  \(L_{w}(H)\) are nested and \(\frac{\partial}{\partial H}H_{2}^{t}(q_{2,}H)=\frac{1}{1+(\epsilon _{w}Q'(q_{2}))^{2}}>0\) and \(Q'(q_{2})\) is bounded on the finite intervals \(L_{w}(H)\), it is sufficient to prove the claim for  \(H=H^{tmin}_{min}+\delta\) at   arbitrary small  \(\delta\). Notice that   \(q_{w}(q_{2}^{w-min})\in S_{w}(H)\) is an interior of   \(L_{w}(H)\)  for all    \(H>H^{tmin}_{min}\),  so the claim follows once we prove that
 \(\frac{\partial}{\partial q_{2}}H_{2}^{t}(q_{2,}H^{tmin}_{min}+\delta)|_{q_{2}^{w-min}}\)
 or
 \(\frac{\partial^2}{\partial q_{2}^2}H_{2}^{t}(q_{2,}H^{tmin}_{min}+\delta)|_{q_{2}^{w-min}}\)  do not vanish at sufficiently small \(\delta\).   Indeed, \begin{equation}
 \begin{array}{ll}
\frac{\partial}{\partial q_{2}}H_{2}^{t}(q_{2,}H)|_{q_{2}^{w-min}}&=
\left[V_{2}'(q_2)-
\frac{H^{tmin'}(q_{2})}{1+(\epsilon_{w}Q'(q_{2}))^{2}}-2\epsilon _{w}Q'(q_{2})Q''(q_{2})\frac{H-H^{tmin}(q_{2})}{(1+(\epsilon _{w}Q'(q_{2}))^{2})^2}\right]_{q_{2}^{w-min}}\\
&=V_{2}'({q_{2}^{w-min}})-\frac{2\delta\epsilon _{w}Q'({q_{2}^{w-min}})Q''({q_{2}^{w-min}})}{(1+(\epsilon _{w}Q'({q_{2}^{w-min}}))^{2})^2}\\
&=-\epsilon _{w}Q'({q_{2}^{w-min}})(V_{1}'(\epsilon _{w}Q({q_{2}^{w-min}}))+\frac{2\delta Q''({q_{2}^{w-min}})}{(1+(\epsilon _{w}Q'({q_{2}^{w-min}}))^{2})^2})
\end{array}\end{equation}
so  if (i):
\(Q'({q_{2}^{w-min}}) V_{1}'(\epsilon _{w}Q({q_{2}^{w-min}}))\neq0\) or (ii): \( Q'({q_{2}^{w-min}})Q''({q_{2}^{w-min}})\neq0\),   for sufficiently  small \(\delta>0\) the first derivative does not vanish and the lemma follows. Otherwise,
\begin{equation}
 \begin{array}{ll}
\frac{\partial^{2}}{\partial q_{2}^2}H_{2}^{t}(q_{2,}H)|_{q_{2}^{w-min}}&=\left[V_{2}''(q_2)-\frac{H^{tmin''}(q_{2})}{1+(\epsilon _{w}Q'(q_{2}))^{2}}+4\epsilon_{w}Q'(q_{2}) Q''(q_{2})\frac{H^{tmin'}(q_{2})}{(1+(\epsilon _{w}Q'(q_{2}))^{2})^2}\right]_{q_{2}^{min}}\\
& +\left[4(\epsilon _{w}Q'(q_{2})Q^{''}(q_{2}))^{2}\frac{H-H^{tmin}(q_{2})}{(1+(\epsilon _{w}Q'(q_{2}))^{2})^3}-2\epsilon _{w}(Q'(q_{2})Q''(q_{2}))'\frac{H-H^{tmin}(q_{2})}{(1+(\epsilon _{w}Q'(q_{2}))^{2})^2}\right]_{q_{2}^{w-min}}\\
&=\left[V_{2}''(q_{2})-\frac{H^{tmin''}(q_{2})}{(1+(\epsilon_{w}Q'(q_{2}))^{2})^2}-2\epsilon _{w}(Q''(q_{2})^{2}+Q'(q_{2})Q'''(q_{2}))\frac{\delta}{(1+(\epsilon _{w}Q'(q_{2}))^{2})^2}\right]_{q_{2}^{w-min}}\\
&=\begin{cases}\left[-\epsilon _{w}Q''(q_{2})V_{1}'(\epsilon _{w}Q(q_{2}))-2\epsilon _{w}\delta Q''(q_{2})^{2}\right]_{q_{2}^{w-min}} & \text{if }\qquad Q'({q_{2}^{w-min}})=0  \\
\\
\left[\frac{\epsilon _{w} Q'(q_{2})\left(\epsilon _{w}Q'(q_{2})(V_{2}''(q_{2})-V_{1}''(\epsilon _{w}Q(q_{2})))-2\delta Q'''(q_{2})\right)}{(1+(\epsilon _{w}Q'(q_{2}^{min}))^{2})^2}\right]_{q_{2}^{w-min}} &  \\
\qquad\qquad\qquad \text{if }|V_{1}'(\epsilon _{w}Q({q_{2}^{w-min}}))|+|Q''(q_{2}^{w-min})|=0 &\\
\end{cases}\\
\end{array}
\end{equation}(where we used in the second equality that \(Q'(q_{2}^{w-min})Q''(q_{2}^{w-min})=0\)). Thus, if
 \(\frac{\partial}{\partial q_{2}}H_{2}^{t}(q_{2},H)|_{q_{2}^{w-min}}=0\), the condition for
 \(\frac{\partial^2}{\partial q_{2}^2}H_{2}^{t}(q_{2},H)|_{q_{2}^{w-min}}\neq0\) for small \(\delta\) is that (iii)  \(Q''(q_{2}^{w-min})\neq0\) or (iv)   \(Q'({q_{2}^{w-min}})(V_{2}''({q_{2}^{w-min}})-V_{1}''(\epsilon _{w}Q(q_{2}^{min})))\) or  (v) \( Q'({q_{2}^{w-min}})Q'''(q_{2}^{min}))\neq0\) are satisfied. Eq. (\ref{eq:generalpositionwall}) implies that at least one of the conditions (i-v) are satisfied so for a GWS in general position  \(H_{2}^{t}(q_{2,}H)\) is non constant on \(S_{w}(H)\).      \end{proof}
\begin{lem}\label{lem:claim5}(Claim 5 of Proposition \ref{thm:global}:)
 For  \(H>H^{tr}_{min}\) the  set of \(H_{2}\) values belonging to the iso-energy transverse impact  zone   \( (H_{2}^{tr}(H)=H_2^{m}(H)\setminus H_{2}^{tan}(H)\)) has positive measure, in fact, there exist  \(H_{2}^{lower-tr}(H)>0\) such that this set   includes the interval \([0,H_{2}^{lower-tr}(H)]\).
\end{lem}
 \begin{proof} First we show that for a fixed \(H_{2}\)  value and sufficiently large \(H\), the  wall   (\ref{eq:wavywal}) intersects transversely the horizontal boundaries of all the PRLs \(R^{(k_{1},k_2)}(H-H_{2},H_{2})\)  so the level set \((H-H_{2},H_{2})\) is impacting and \(H_{2}\in H_2^{m}(H)\). We then find, for a fixed \(H_{2}\), a lower bound to the energy, \(H^{tr}(H_2)\), above which all impacts with the PRLs  \(R^{(k_{1},k_2)}(H-H_{2},H_{2})\)  are transverse, and establish that \(H^{tr}(H_2)\) is strictly monotone in \(H_{2}\) and is finite for \(H_2=\min_{q_2} V_2 (q_2)=0\). Then,  setting \(H^{tr}_{min}=H^{tr}(0)\) and \(H_{2}^{lower-tr}(H)\) to be the inverse function of  \(H^{tr}(H_2)\) we complete the proof by proving that   \(H^{lower-tr}(H_2)>0\).

For a fixed \(H_{2}\) value, since, by the S3BN assumption \(H_{2}\) has bounded level sets,  the vertical projection of the PRLs,  \([q^{k_{2}}_{2,min}(H_2),q^{k_{2}}_{2,max}(H_2)]\), is finite for all \(k_{2}=1,..,K_2(H_2) \). Let \( \mathcal{L}(H_{2}):=\cup_{k_{2}=1,..K_{2}(H_{2})}[q^{k_2}_{2,min}(H_2),q^{k_2}_{2,max}(H_2)]\) denote this projection of the \(H_{2}\) level set to the \(q_2\) axis. Since \(Q(q_{2})\) is \(C^{r+1}\), the wall intersection with the PRLs,  \(\{q_{w}(q_2)=(\epsilon _{w}Q(q_{2}),q_{2})\}_{q_{2}\in\mathcal{L}(H_{2})}\),  lies  within the bounded vertical strip   \(q_{1}\in[\min_{q_{2}\in\mathcal{L}(H_{2})}\epsilon _{w}Q(q_{2}),\max_{q_{2}\in\mathcal{L}(H_{2})}\epsilon _{w}Q(q_{2})]\). By the S3BN assumption on the growth of \(V_{1}\), the PRLs widths increase to infinity with \(H_{1}=H-H_{2}\), namely \(q^{k_{1}}_{1,min}(H-H_2)\rightarrow-\infty,\ q^{k_{1}}_{1,max}(H-H_2)\rightarrow\infty\) for all \(k_{1}=1,..,K_1(H-H_2) \) (the growth rate with \(H\) depends on the form of \(V_1\)). Thus, for sufficiently large \(H\), the vertical strip that bounds the wall intersects  each of the PRLs and the wall  divides each of the PRL vertically to two regions, one of which is inside the billiard.

The impact with the wall  along the wall intersections with the PRLs \(R^{k}(H-H_{2},H_{2})\)   is:
\begin{equation}
\begin{array}{ll}
|(p_{1},p_2)\cdot \hat n\vert_{q_{w}(q_{2},H),q_{2}\in\mathcal{L}(H_{2})} &=  \frac{|\pm\sqrt{2(H-H_{2}- V_1(\epsilon _{w}Q(q_2)))}\pm\epsilon_{w}Q'(q_{2})\sqrt{2(H_{2 }- V_2(q_2))}|}{\sqrt{1+(\epsilon_{w}Q'(q_{2}))^{2}}}
\end{array}
\end{equation}
where, as we show next, the first term is defined for the energies that we consider. Indeed, defining \(C_{1}(H_2)=H_{2}+\max_{q_{2}\in\mathcal{L}(H_{2})} V_1(\epsilon _{w}Q(q_2))\), \( C_{2}(H_2)=H_{2}-\min_{q_{2}\in\mathcal{L}(H_{2})} V_2(q_2)=H_{2 }\) (since the global minimizers of \(V_{2}\) are in  \(\mathcal{L}(H_{2})\) for all \(H_2\geq0\)),  \(  C_{3}(H_{2})=\max_{_{q_{2}\in\mathcal{L}(H_{2})}} |\epsilon _{w}Q'(q_{2})|\) and   \(H^{tr}(H_{2})=C_{1}(H_2)+C_{2}(H_2)\left({C_{3}(H_2)}\right)^2\), the impacts with the wall are defined and transverse for all \(q_{2}\in\mathcal{L}(H_{2})\)  provided \({H> H^{tr}(H_{2})}\) :
\begin{equation}
\begin{array}{ll}
|(p_{1},p_2)\cdot \hat n\vert_{q_{w}(q_{2},H),q_{2}\in\mathcal{L}(H_{2})} & \geqslant \frac{|\sqrt{2(H-C_{1}(H_2))}-C_{3}(H_{2})\sqrt{2C_{2}(H_2)}|}{\sqrt{1+(C_{3}(H_{2}))^{2}}}|_{H> H_{tr}(H_{2})}>0. \\

\end{array}
\end{equation}
     Since the intervals \( \mathcal{L}(H_{2})\) are nested, for all \(H_{2}\in[0,H]\) the  functions     \(C_{1,2,3}(H_2) \) (and thus \(H_{tr}(H_{2})\)) are  bounded, continuous, piecewise smooth and monotone increasing functions of \(H_2\). Moreover, by its definition, \(C_{1}(H_2)\)  is  strictly monotone for all \(H_{2}\geq0\) and thus so is   \(H^{tr}(H_{2})\). At \(H_{2}=0,\)   \(H_{tr}(0)=\max _{i}V_{1}(\epsilon _{w}Q(q_{2,i}^{min}))=H^{tr}_{min}\) of Eq. (\ref{eq:htrmin}), where \(q_{2,i}^{min}\) are the global minimizers of \(V_2\) (the horizontal periodic orbits become transverse once \(\epsilon _{w}Q(q_{2,i}^{min})\) is in the range of \(V_{1}\)). By the strict monotonicity of  \(H^{tr}(H_{2})\) we obtain that  \(H^{tr}(H_{2})>H_{tr}(0)\) for all \(H_{2}>0 \). Hence, for   \(H>H^{tr}_{min}\), the function  \(H_{2}^{lower-tr}(H)\), the inverse function of \(H^{tr}(H_{2}),\) satisfying \(H^{tr}(H_{2}^{lower-tr}(H))=H\), is well defined and is strictly monotone, so  \(H_{2}^{lower-tr}(H)>0\) for all \(H>H^{tr}(0)\).  Thus, for all    \(H>H^{tr}_{min}\), the \(H_{2}\) interval  \([0,H_{2}^{lower-tr}(H)]\) consists of  level sets at which all impacts are transverse and it is included in     \( H_2^{m}(H)\setminus H_{2}^{tan}(H)\) so this set is of positive measure as claimed.     \end{proof}

 The transverse impact set      \( H_2^{m}(H)\setminus H_{2}^{tan}(H)\) may include additional intervals to   \([0,H_{2}^{lower-tr}(H)]\).
For example, it may include neighborhoods of horizontal normal modes which are far from the global minimizers of \(V_{2}\) and neighborhoods of vertical normal modes which are shadowed by the wall (i.e.  leaves close to those of \((V_1(q_1^{min}),H-V_1(q_1^{min}))\) centered at   \((q_1^{min},q_2^{ext,j})\)  where \(V_{1}'(q_1^{min})=0,V_{1}''(q_1^{min})>0\),  the line \(q_{1}=q_1^{min}\) intersects the wall transversely at a finite positive number of isolated points and \(q_2^{ext,j}\) are the extremizers of \(V_2\)).

\begin{figure}
\begin{centering}
\includegraphics[scale=0.45]{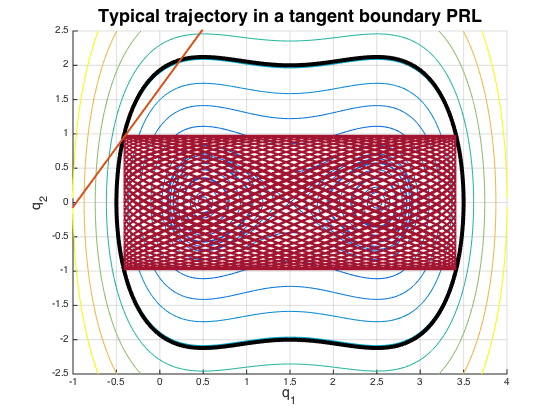}\includegraphics[scale=0.45]{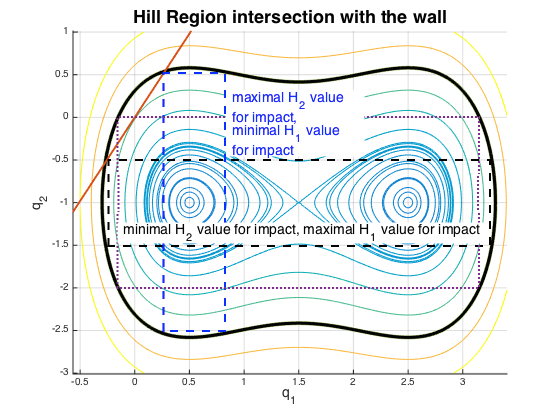}
\protect\caption{\label{fig:Hill-impact} Configuration space figures of the Duffing-Center system with potential level lines and the addition of a (general) slanted wall. The boundary of the Hill region is depicted in bold black. (a) A typical trajectory in a tangent boundary PRL - the only intersection between the wall (in orange) and the rectangle filled by the trajectory is on the boundary of the Hill region, i.e. at the PRL boundary. (b) Intersection between the Hill region and a slanted wall in the configuration space. The wall (orange line) intersects \(\mathcal{D}_2^{Hill}(H)\) boundary along the segment \(S_{w}(H)\). The projected rectangles of the two leaves \((H-H_{2},H_{2})\) with corners at this segment end points (dashed rectangles) are drawn schematically and have a single tangent trajectory and no transverse impacts, so these are tangent boundary leaves. The PRL  which is intersected transversely by the wall at neighboring PRL boundaries   (dotted rectangles) is a tangent leaf: it includes segments with transverse impacts  and tangent segments.    }\end{centering}
\end{figure}

\begin{figure}
\begin{centering}
\includegraphics[scale=0.45]{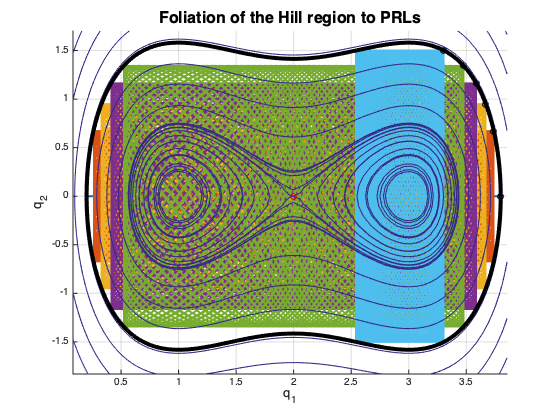}
\protect\caption{\label{fig:Hill-rectangles} Non-resonant iso-energy trajectories of the Duffing-Center integrable system (\ref{eq:doublewell}) in the configuration space. $H=1$, and the boundary of the Hill region is marked in black. The trajectories densely fill rectangles in the configuration space corresponding to their initial $H_2$ values - these are different PRLs of the system.}\end{centering}
\end{figure}

\begin{figure}
\begin{centering}
\includegraphics[scale=0.35]{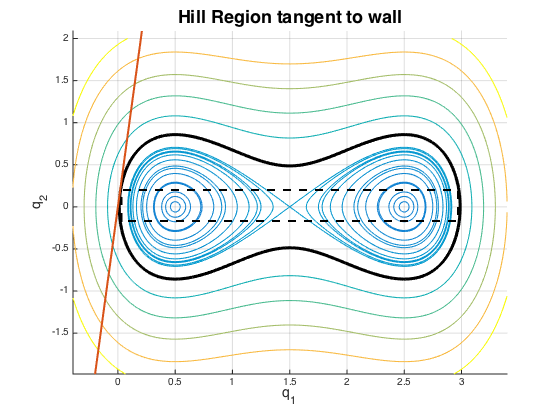} \includegraphics[scale=0.35]{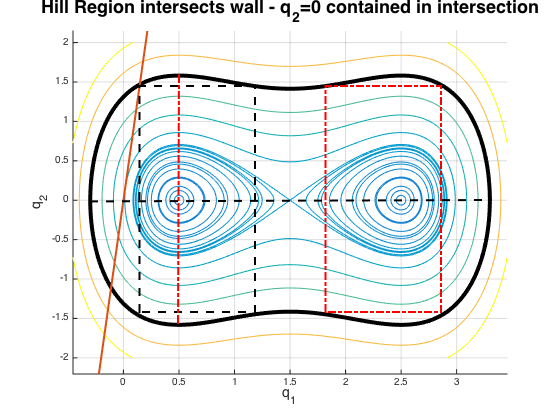} \includegraphics[scale=0.24]{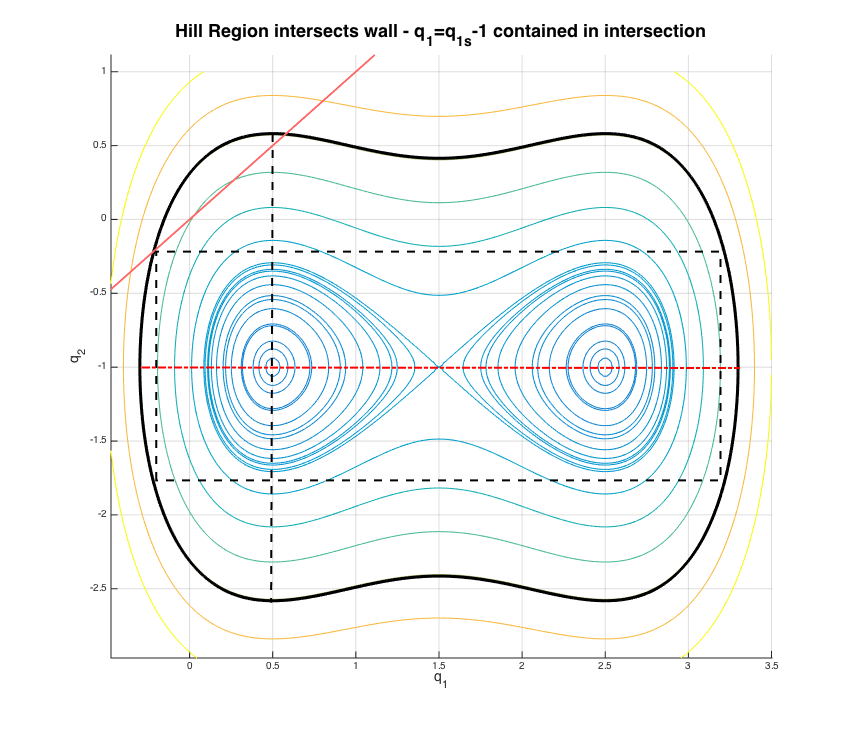}
\protect\caption{\label{fig:Hill-edges} Intersections of the Hill region with the wall - see subsections \ref{subsec:impactzones}, \ref{subsec:slantedwall} for notations.   (a) At \(H=H^{tmin}_{min}\) the Hill region is tangent to the wall at the corner point of  the PRL of the boundary tangent leaf. (b) The horizontal periodic orbit impacts the wall $(Q(q_{2c}),q_{2c})\in S_{w}(H)$, i.e.
 $H>H^0$ (see (\ref{eq:h0})). (c) The vertical periodic orbit impacts the wall $(q_{1s}-1,Q^{-1}(q_{1s}-1))\in S_{w}(H)$, i.e.  $H>H^{1,-}$ .}\end{centering}
\end{figure}

\label{sec:thmproofs}

 \subsection{The Duffing-Center potential and a slanted wall}\label{subsec:slantedwall}

The proof of Proposition \ref{thm:global} provides a constructive methodology for identifying the  non-impact, tangency and transverse impact zones in the IEMBD. Similar, more detailed calculations regarding the impacts with each family of PRLs can be made to construct the IFG.  We demonstrate the application of this construction for finding the IEMBD to the case of the DC  Hamiltonian (\ref{eq:doublewell}) with a slanted wall (\(q_{w}(q_2)=\cot \alpha\cdot q_2,q_2)\) with \(\alpha\in(0,\frac{\pi}{2})\)).  For concreteness, we consider the Hamiltonian (\ref{eq:doublewell}) when the extrema points of the  Duffing-Center potential are within the billiard domain:\begin{equation}\label{eq:extremalinslantedbil}
q_{1s}\geq1,\quad q_{1s}-q_{2c}\cot \alpha >1\quad\ \omega=1
\end{equation}
and choose parameters in a certain range so that the critical energies obey certain ordering; Let:
\begin{equation}
H^{tmin,dc}(q_{2})=V_{1}(\cot\alpha \cdot q_{2})+V_{2}(q_2) =
\frac{1}{4}\cdot(\cot\alpha \cdot q_2-q_{1s})^{4}-\frac{1}{2}\cdot(\cot\alpha \cdot q_2-q_{1s})^{2}+\frac{\omega^{2}}{2}\cdot(q_{2}-q_{2c})^{2}
\label{eq:h0}
\end{equation}and\begin{equation}\begin{array}{lll}
H^{tmin,dc}_{min}& =\min_{q_{2}} H^{tmin,dc}(q_{2})
\\
H^{0}
&=H^{tmin,dc}( q_{2c})&=V_{1}(\cot\alpha \cdot q_{2c})=H^{tr,dc}_{min}\\
H^{1,\pm}&=H^{tmin,dc}((q_{1s}\pm1)\cdot\tan\alpha)&=-\frac{1}{4}+\frac{\omega^{2}}{2}\cdot((q_{1s}\pm1)\cdot\tan\alpha-q_{2c})^{2}\\
H^{1,0}&=H^{tmin,dc}(q_{1s}\cdot\tan\alpha)&=\frac{\omega^{2}}{2}\cdot(q_{1s}\cdot\tan\alpha-q_{2c})^{2}.
\end{array}
\label{eq:h0critval}
\end{equation}Since
\(H^{tmin,dc}(q_{2})\) is an asymmetric quartic polynomial in \(q_{2}\) with positive \(q_{2}^{4}\) coefficient,  \(H^{tmin,dc}_{min}\) is finite, and
\(H^{tmin,dc}(q_{2})\) has at most one additional local maximum and minimum. For concreteness we choose \(\alpha\in(\arctan\frac{1}{\omega},\frac{\pi}{2})\)  so that
\(H^{tmin,dc}(q_{2})\) has only one extremum (thus a global minimum) and the other parameters are chosen such that  the critical energies of the system, given by Eqs. (\ref{eq:h0}) satisfy the ordering:\begin{equation}\label{eq:horder1}
H^{tmin,dc}_{min}<H^{0}<H^{1,-}<H^{1,0}<H^{1,+}.
\end{equation}Other cases can be similarly analyzed.  These open conditions are satisfied for  \((\alpha,q_{1s},q_{2c},\omega)=(\frac{\pi }{2}-0.1,2,0,1)\), 
as in Fig \ref{fig:EMBD-impact},  and thus are also satisfied for a neighborhood of these parameters.

We establish that the structure of the impact zones changes non-trivially with \(H\) and that eventually, for \(H>H^{1,+}\)  the slanted wall intersects transversely every PRL of every level set \((H-H_{2},H_2)\) in the allowed region of motion. Yet, surprisingly, for  sufficiently large \(H\), a large portion of these level sets also admit tangencies (see item 6 in the Proposition and Figure \ref{fig:EMBD-impact}d):\begin{prop}
\label{prop:globalslanted} The impact zones  of  the  \(\alpha-\)slanted Duffing-Center  system, with  \(\alpha\in(\arctan\frac{1}{\omega},\frac{\pi}{2})\), for the cases in which the potential extremal points are within the billiard  domain (Eq (\ref{eq:extremalinslantedbil})),  and  the critical energies satisfy the ordering (\ref{eq:horder1}), have the following properties:\begin{enumerate}
\item
For all \(H\in [-\frac{1}{4},H^{tmin,dc}_{min})\) all the level sets are in the non-impact zone and the HIS energy surface coincides with the energy surface of  (\ref{eq:doublewell}).  \item  For  \(H\in [H^{tmin,dc}_{min},H^{0})\), the tangency and the impact zones coincide and correspond to a single \(H_{2}\) interval of a positive length:      \(H_{2}^{tan}(H)=H_2^{m}(H)=[V_2(q_2^a(H)),H-V_1(\cot\alpha \cdot q_2^b(H))]\), where \(q_2^{a,b}(H)\) are the two consequent zeroes of \(H^{tmin,dc}(q_2^{a,b}(H))=H\). The non-impact zone is divided to lower and upper segments of positive lengths: \( H_{2}^{non-impact}(H)=[0,V_2(q_2^a(H))\cup(H-V_1(\cot\alpha \cdot q_2^b(H)),H+\frac{1}{4}].\)
\item
For  \(H\in [H^{0},H^{1,-})\), the non-impact zone consists of a single segment: \( H_{2}^{non-impact}(H)=(H-V_1(\cot\alpha \cdot q_2^b(H)),H+\frac{1}{4}] \),  and the tangent zone lower boundary separates from the impact zone lower boundary:  \(H_{2}^{tan}(H)=[H_{2}^{tan,1}(H),H-V_1(\cot\alpha \cdot q_2^b(H))]\subset H_2^{m}(H)=[0,H-V_1(\cot\alpha \cdot q_2^b(H))]\), where \(0<H_{2}^{tan,1}(H)<H-V_1(\cot\alpha \cdot q_2^b(H)).\)
\item
For \(H>H^{1,-}\) the tangent set  consists of an interval which has positive length and is strictly within the  impact zone, which coincides with the allowed region of motion: \  \(H_{2}^{tan}(H)=[H_{2}^{tan,1}(H),H_{2}^{tan,2}(H)]\) \(\subset H_2^{m}(H)=[0,H+\frac{1}{4}]\).

\item
For  \(H\in [H^{1,-},H^{1,0}]\) there are no non-impact level sets,  yet all the right branch of the IFG has no-impact leaves. For  \(H\in (H^{1,0},H^{1,+})\)  impacts of the right IFG branch emerge and for   \(H>H^{1+}\) all leaves on the energy surface are impact leaves.

\item For any fixed   \(\alpha\in(0,\frac{\pi}{2})\), for all parameter values, for sufficiently large \(H\), the relative measure of the tangency
zone  is proportional to $\sin^{2}\alpha$:  \(|H_{2}^{tan}(H)|/|H_{2}^{m}(H)|=\sin^{2}\alpha\cdot(1+O(\frac{1}{\sqrt{H}}))\). \end{enumerate}
\end{prop}
\begin{proof}
Recall that the boundaries of \(H_{2}^m(H)\) are found by minimizing \(V_2(q_2)\) and maximizing \(H-V_1(\epsilon _{w}Q(q_2))\) over the segments  \(S_w(H)\)  (Eq. (\ref{eq:h2mdef})). Similarly, the boundaries of the  tangency zone are found by minimizing and maximizing \(H_{2}^{t}(q_{2,}H)\) over   \(S_w(H)\)    (Eq. (\ref{eq:tanghh2},\ref{eq:h2tanminmaxform})). To prove the proposition we establish some properties of    \(S_w(H)\) and then find the lower and upper bounds of \(H_{2}^{m,tan}(H)\) .

The segments  \(S_w(H)\), or more precisely, their projections to the \(q_{2}\) axis, \(L_{w}(H)\), are found by solving the quartic inequality: \begin{equation}
\begin{array}{ll}
H^{tmin,dc}(q_{2})&\leqslant H.
\end{array}\label{eq:swdoublewell}
\end{equation}
Computation of the second derivative of \(H^{tmin,dc}(q_{2})\)  shows that for \(\omega^2\tan^{2}\alpha>1\) it is always positive, thus \(H^{tmin,dc}(q_{2})\) has a global minimum and no other extrema. Using the inequality (\ref{eq:extremalinslantedbil}) a simple computation shows that \(H^{tmin,dc'}(q_{2c})<0\) and as \(H^{tmin,dc'}\) is monotone increasing, the global minimizer of \(H^{tmin,dc}(q_{2})\) is above the minimizers of \(V\), namely, \(q_{2}^{w-min}>q_{2c}\).  Thus (\ref{eq:swdoublewell}) is satisfied on a single interval, \(L_{w}(H)=[q_2^a(H),q_2^b(H)]\) which grows in size with \(H \)  as   \(\frac{dq_2^{a}}{dH}<0,\frac{dq_2^{b}}{dH}>0\) and\  \begin{equation}
\min_{q_{2}\in L_{w}(H)}H^{tmin,dc}(q_{2})=H^{tmin,dc}(q_{2}^{w-min})=H^{tmin,dc}_{min}, \quad q_{2}^{w-min}>q_{2c}
\end{equation}
By the definition of  \(H^{tmin,dc}(q_{2})\), for \(H<H^{tmin,dc}_{min}\), the wall does not intersect the Hill region. Since, by (\ref{eq:extremalinslantedbil}), the two global minima of \(V\) are within the billiard domain, and thus \(H^{tmin,dc}_{min}\) is larger than the global minima, the corresponding leaves are in the billiard by claim 2 of Proposition \ref{thm:global} and the first item of the proposition follows.

To prove items 2-4, we calculate/give bounds on  the lower and upper boundaries of \(H_{2}^m(H),H_{2}^{tan}(H)\) for the different regimes of \(H\).

\noindent\textit{Lower boundary of \(H_{2}^m(H)\) }: The lower boundary of \(H_{2}^m(H)\)  is \(\min_{q_2\in L_{w}(H)} V_2(q_2)\). By lemma \ref{lem:prlwall}, \(H^{0}\) is the energy above which \(S_{w}(H)\) includes the wall point \((\cot{\alpha}\cdot{}q_{2c},{}q_{2c})\). Since we showed that  \(q_{2}^{w-min}>q_{2c}\), at this energy $q_2^a(H^0)=q_{2c}$. This lower boundary  of \(L_{w}(H)\) is above \(q_{2c}\) for \(H<H^{0}\) and below it for  \(H>H^{0}\). Hence, the potential \(V_{2}\) is monotonically decreasing along  \( L_{w}(H)\)  for \(H<H^{0}\), so \( \min_{q_{2}\in L_{w}(H)}V_2(q_2)= V_{2}(q_2^a)\) and   \(\frac{d}{dH}V_{2}(q_2^a(H))<0\). At \(H^{0}\) the function  \(V_{2}(q_2^a(H))\) reaches its minimum (0) and for  \(H>H^{0}\),
  \(V_2\)  increases along \(q_2^a(H)\) while  \( \min_{q_{2}\in L_{w}(H)}V_2(q_2)= V_{2}(q_{2c})=0\).  Thus,  for \(H<H^{0}\) the lower boundary of  \(H_{2}^m(H)\) is \(V_{2}(q_2^a(H))\) and for all  \(H>H^{0}\) it vanishes: \begin{equation}\label{eq:lowbm}
\text{lower boundary of } H_{2}^m(H)=\begin{cases} V_{2}(q_2^a(H)) & H\in[H^{tmin,dc}_{min},H^{0}) \\
0  & H>H^{0}
\end{cases}
\end{equation}
\\
\noindent\textit{Lower boundary of \(H_{2}^{tan}(H)\) }: To minimize \(H_{2}^{t}(q_{2,}H)\) on \(L_{w}(H)\), notice that from Eq. (\ref{eq:tanghh2})       \(H_{2}^{t}(q_{2,}H)_{q_{2}\in L_{w}^{j}(H)\backslash\partial L_{w}^{j}(H)}> V_2(q_2)\geq0\) and that \(H_{2}^{t}(q_2,H)_{q_{2}\in\partial L_{w}^{j}(H)}=V_2(q_{2})_{q_{2}\in\partial L_{w}^{j}(H)}\), hence, for  \(H<H^{0}\), where \(V_2(q_2)\) is monotonically decreasing on \(L_{w}(H)\), we obtain that \(\min _{q_{2}\in L_{w}(H)} H_{2}^{t}(q_{2},H)=V_2(q_2^{a}(H))\) and thus the lower boundary of  \(H_{2}^{tan}(H)\) and of \(H_{2}^m(H)\) coincide for \(H\in[H^{tmin,dc}_{min},H^{0})\). On the other hand, for \(H>H^{0}\),
since \(V_2(q_2^{a,b}(H))>0\)  and      \(H_{2}^{t}(q_{2,}H)_{q_{2}\in L_{w}^{j}(H)\backslash\partial L_{w}^{j}(H)}> V_2(q_2)\geq0\)   we get that the minimal value of \(H_{2}^{t}(q_{2,}H)\) is achieved in an interior point of \( L_{w}(H)\):\begin{equation}
0<H_{2}^{tan,1}(H):=\min _{q_{2}\in L_w(H)} H_{2}^{t}(q_{2},H)<\min \{V_{2}(q_{2}^{a}),V_{2}(q_{2}^{b})\}
\end{equation} \begin{equation}
\text{lower boundary of } H_{2}^{tan}(H)=\begin{cases}V_{2}(q_2^a(H)) & H\in[H^{tmin,dc}_{min},H^{0}) \\
H_{2}^{tan,1}(H)  & H>H^{0} \\

\end{cases}
\label{eq:lowbtan}\end{equation}

\noindent\textit{Upper boundary of \(H_{2}^m(H)\) }:
The upper boundary of \(H_{2}^m(H)\)  is \(\max_{q_2\in L_{w}(H)}(H-V_1(\cot\alpha\cdot q_{2}))=H-\min_{q_2\in L_{w}(H)} V_1(\cot\alpha\cdot q_2)\). Recall that \(H^{1,\pm},H^{1,0}\) of Eq. (\ref{eq:h0}) are the energies above which \(S_{w}(H)\) includes the wall points \((q_{1s}\pm1,(q_{1s}\pm1)\cdot\tan\alpha)\) and  \((q_{1s},q_{1s}\cdot\tan\alpha)\), respectively.
For \(H<H^{1,-}\), the potential \(V_{1}\) is monotonically decreasing along \(L_{w}(H)\), so  the upper boundary of \(H_{2}^m(H)\) is \(H-V_1(\cot\alpha\cdot q_{2}^{b})\). For \(H>H^{1,-}\), the the global minimzer of \(V_{1}\) is included in the interval so \(\max_{q_2\in L_{w}(H)}(H-V_1(\cot\alpha\cdot q_{2}))=H+\frac{1}{4} \), namely the upper boundary coincides with the boundary of the allowed region of motion. \begin{equation}
\text{upper boundary of } H_{2}^m(H)=\begin{cases} H-V_1(\cot\alpha\cdot q_{2}^{b}(H)) & H\in[H^{tmin,dc}_{min},H^{1,-}) \\
H+\frac{1}{4}  & H>H^{1,-} \\

\end{cases}
\label{eq:upboundm}\end{equation}

\noindent\textit{Upper boundary of \(H_{2}^{tan}(H)\) }: Rewriting Eq. (\ref{eq:tanghh2}) as: \begin{displaymath}
H_{2}^{t}(q_{2,}H) =H-V_{1}(\epsilon _{w}Q(q_2))+\frac{(\epsilon _{w}Q'(q_{2}))^{2}(H^{tmin}(q_{2})-H)}{1+(\epsilon _{w}Q'(q_{2}))^{2}}=H-V_{1}(\cot\alpha\cdot q_{2})+\cos^{2}\alpha \cdot(H^{tmin}(q_{2})-H)
\end{displaymath} we see that\begin{equation}
H_{2}^{t}(q_{2,}H)_{q_{2}\in L_{w}^{j}(H)\backslash\partial L_{w}^{j}(H)}< H-V_{1}(\cot\alpha\cdot q_{2})\leqslant  H+\frac{1}{4}\label{eq:interiorh2t}
\end{equation}      and  \(H_{2}^{t}(q_2^{a,b}(H),H)_{}=H-V_{1}(\cot\alpha\cdot q_2^{a,b}(H)))=V_2(q_2^{a,b}(H))\). For \(H<H^{1,-}\), the potential \(V_{1}\) is monotonically decreasing along \(L_{w}(H)\), and thus the upper boundary of \(H_{2}^m(H)\) and \(H_{2}^{tan}(H)\) coincides and is given by     \(H-V_1(\cot\alpha\cdot q_{2}^{b}(H))\). For \(H^{1,-}<H< H^{1,+}\), or \(H>H^{1,+}\), since \(H_{2}^{t}(q_2^{a,b}(H),H)<H+\frac{1}{4}\), by Eq. (\ref{eq:interiorh2t}) the maximal value of \(H_{2}^{t}(q_{2,}H)_{q_{2}\in L_{w}^{j}(H)}\) is smaller than \(H+\frac{1}{4}, \) which is the upper boundary of  \(H_{2}^{tan}(H)\). At \(H=H^{1,\pm}\), since by definition, \(V_1(\cot\alpha\cdot q_{2}^{b}(H^{1,\pm}))=V_1(q_{1s}\pm1)=-\frac{1}{4}\) the tangency and impact zone upper boundary coincide. Denoting this maximal value by \(H_{2}^{tan,2}(H)\), where \begin{equation}
\max\{V_{2}(q_2^{a}(H),V_{2}(q_2^{b}(H)\}<H_{2}^{tan,2}(H)<H+\frac{1}{4}
\end{equation} we establish:
\begin{equation}
\text{upper boundary of } H_{2}^{tan}(H)=\begin{cases} H-V_1(\cot\alpha\cdot q_{2}^{b}(H)) & H\in[H^{tmin,dc}_{min},H^{1,-}) \\
H_{2}^{tan,2}(H)  & H>[H^{1,-},H^{1,+}) \\
H+\frac{1}{4}  & H=H^{1,+}\\
H_{2}^{tan,2}(H)  & H>H^{1,+}
\end{cases}
\label{eq:upbountan}\end{equation}

Items 2-4 follow from equations (\ref{eq:lowbm}-\ref{eq:upbountan}).

To prove item 5, notice that for \(H<H^{1,0}\) the segment \(S_{w}(H)\) does not include any wall points  with \(q_{1}>q_{1s}\), so the right edge of the IFG has no impacts with the wall for such values. For  \(H^{1,0}<H<H^{1,+}\), since the segment \(L_{w}(H)\) does not reach the wall point \((q_{1s}+1,(q_{1s}+1)\cdot\tan\alpha)\),   the PRLs of leaves of the IFG right branch centered at \((q_{1s}+1,q_{2c})\) with sufficiently small width do not impact. On the other hand, for \(H\geqslant H^{1,+}\)  the wall intersects this central singular PRL and all the  PRLs around it, namely all the PRLs of the right branch of the IFG. The left and outer branches of the IFG are also clearly intersected since the slanted wall slope is positive.

Finally, to establish item 6,
first notice that independent of the parameters, for sufficiently large \(H\) the end points of \(L_{w}(H)\) are determined by the quartic term in (\ref{eq:h0}), so,   \(q_2^{a,b}(H)\approx\tan\alpha  \cdot( q_{1s}\mp\left( 4H \right)^{1/4}\)). Also, notice that Eq. (\ref{eq:tanghh2}) for a slanted wall becomes:
\begin{equation}\label{eq:tanghh2slant}
H_{2}^{t}(q_{2,}H)=\sin^2(\alpha)H+\cos^2(\alpha)V_{2}(q_2) - \sin^2(\alpha)V_{1}(\cot\alpha \cdot q_{2}):=\sin^2(\alpha)H+G(q_{2}),
\end{equation}so for the DC potential  \(G\) is  quartic in \(q_{2}\) with negative \((q_{2})^{4}\) terms. Hence, there exists some  \(H^G\), such that for all  \(H>H^{G}\), the global maximum  of the function \(G(q_2)\), \(G_{max}=\max_{q_2\in L_{w}(H)} G(q_{2})\)  is realized within this interval. For all \(H>H^{G}\), for all \(q_2\in L_{w}(H)\),  \(H_{2}^{t}(q_{2},H)\leqslant\sin^2(\alpha)\cdot H+G_{max}\). On the other hand, for sufficiently large \(H,\) the minimal value of \(H_{2}^{t}(q_{2,}H)\) is realized at the interval boundaries, where     \(q_2^{a,b}(H)=\mp\tan\alpha  \cdot\left( 4H \right)^{1/4}(1+O(H^{-1/4}))\), so  \(H_{2}^{t}(q_2^{a,b}(H),H)=V_2(q_2^{a,b}(H))=\omega^{2}\tan^{2}\alpha  \cdot H^{1/2}(1+O(H^{-1/4}))\).
We conclude that \(H_{2}^{tan}(H)\approx[\omega ^{2}\tan^2(\alpha)\sqrt{H},\sin^2(\alpha)\cdot H+G_{max}]\) and thus \(|H_{2}^{tan}(H)|=\sin^2(\alpha)\cdot H(1+O(\frac{1}{\sqrt{H}})).\)
 On the other hand,   for sufficiently large \(H\), \(H_{2}^{m}(H)\) coincides with the allowed region of motion, namely \(H_{2}^{m}(H)=[0, \frac{1}{4}+H]\), and \(|H_{2}^{m}(H)|=H(1+O(\frac{1}{H}))\).

\end{proof}

Figure \ref{fig:EMBD-impact} shows the IEMBD for an \(\alpha\) value which is close to \(\frac{\pi}{2}\)  for various energy ranges. For intermediate  \(H\) values the impact zone (blue) and the tangency zone (green) are both of comparable sizes whereas for large \(H \), since here \(\sin^2(\alpha)=0.99\), only a very small fraction of the impact zone corresponds to transverse impacts.
At such large energies, the transverse impact zones correspond to segments close to  the "normal modes" - initial conditions that start near  stable periodic orbits of (\ref{eq:doublewell}) with almost all the kinetic energy concentrated at one of the partial systems. For large \(H\) these horizontal (near the \(H_{2}=0\) boundary) or vertical (near the \(H_{2}=H+\frac{1}{4}\) boundary) segments  cross the wall and thus the impacts close to them becomes transverse (and far from the PWS limits the periodic motion is destroyed). Thus, for a fixed \(\alpha\in(0,\frac{\pi}{2})\), for sufficiently large $H$ (in particular, \(H\gg\max\{H^{1,+},H^G\})\), the interval
of allowed $H_2$ values is divided to three sub intervals:
\begin{equation}\label{eq:h2devisionkargeh}
\begin{cases}
H_2\in[0,\mathcal{O}(\tan^2(\alpha)\sqrt{H})] & \mbox{transverse "horizontal" impacts}\\
H_2\in[\mathcal{O}(\sin^2(\alpha)\sqrt{H}),\sin^2(\alpha)H+\mathcal{O}(1)] & \mbox{tangencies and impacts }\\
H_{2}\in[\sin^2(\alpha)\cdot H+\mathcal{O}(1),H+\mathcal{O}(1)] & \mbox{transverse "vertical" impacts}
\end{cases}
\end{equation}
These zones are, in general, not invariant.  If \(\alpha\approx\pi/4\) it may happen that orbits will remain for a long time in the non-tangency zone hoping between near horizontal to near vertical motions. Such asymptotic results  may be similarly performed for other Hamiltonians of the form (\ref{eq:Hint}) and to a more general form of the wall.

\begin{figure}
\begin{centering}
\includegraphics[scale=0.45]{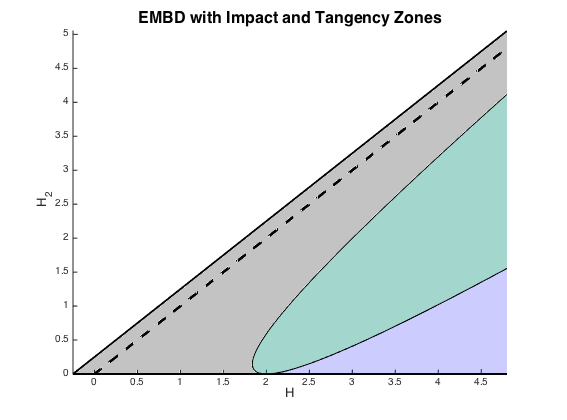}\includegraphics[scale=0.45]{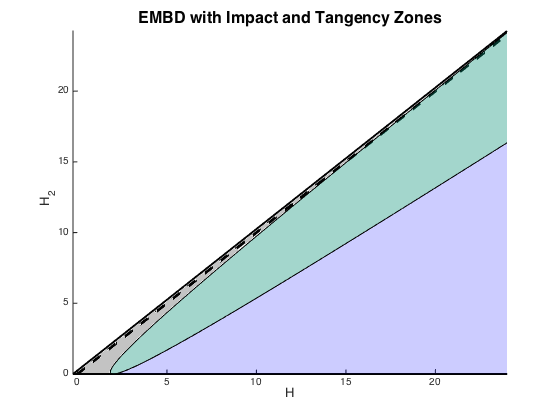}
\par\end{centering}

\begin{centering}
\includegraphics[scale=0.45]{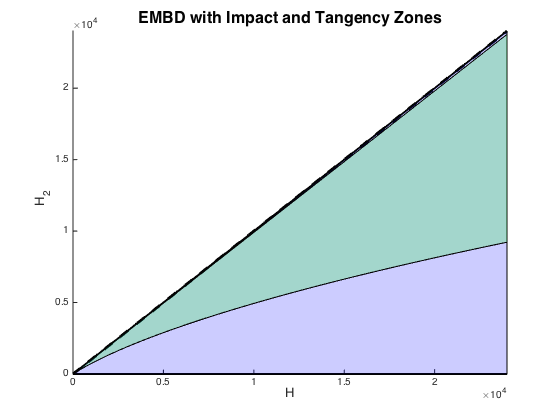}\includegraphics[scale=0.45]{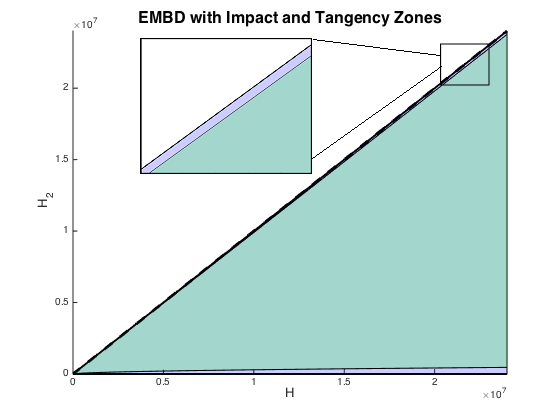}
\par\end{centering}

\protect\caption{\label{fig:EMBD-impact}IEMBD for the system (\ref{eq:doublewell}) with a slanted wall  $(\alpha=\frac{\pi}{2}-0.1)$.
The transverse impact zone (blue), the tangency zone (green), and the no-impact zone (grey) are shown for both low (a-b) and large (c-d) energy ranges. }
\end{figure}

Finally, the two near perpendicular cases, for which Theorem \ref{thm:nearintegrability} (see also   \cite{pnueli2018near}) applies, correspond to near integrable dynamics. Then, in the transverse impact zones, away from separatrices, KAM theory  applies, so a  large portion of the transverse impact zone remains invariant and admits mostly quasiperiodic motion. The IEMBD for the near integrable cases thus serves as a classifying tool for the global, long-term behavior of trajectories. Smooth near integrable behavior emerges in both the non-impact and the transverse-impact zones, and complicated singular behavior appears in a limited zone which includes the tangent zone and its small neighborhood  \cite{pnueli2018near}.
For a fixed energy \(H\), in the integrable limits  (\( \alpha\rightarrow{}0\) or   $\alpha\rightarrow\frac{\pi}{2}$), the tangency zone in the IEMBD approaches the limiting tangency line, namely, its width shrinks to zero.  Namely,  the integrable limits and the large \(H \) limit do not commute - item 6 of Proposition   \ref{prop:globalslanted}  and the estimates (\ref{eq:h2devisionkargeh}) apply only for a fixed \(\alpha\) value. Indeed, in the integrable limits, some of the critical energies above which the asymptotic analysis is applicable  become infinite; for $\alpha\rightarrow{}0$, $H^{1,\pm}\rightarrow V_{2}(0)+H^{min}=-\frac{1}{4}+\frac{\omega^{2}}{2}\cdot (q_{2c})^{2}$ and $H^{0}\rightarrow\infty$ whereas for $\alpha\rightarrow\frac{\pi}{2}$, $H^{1,-}\rightarrow \infty$ and $H^{0}\rightarrow V_{1}(0)=-\frac{1}{2}\cdot(q_{1s})^{2}+
\frac{1}{4}\cdot(q_{1s})^{4}$, see Fig. \ref{fig:alphaasymp}. See \cite{pnueli2016thesis} for some detailed calculations of various asymptotic limits.

\begin{figure}
\begin{centering}
\includegraphics[scale=0.4]{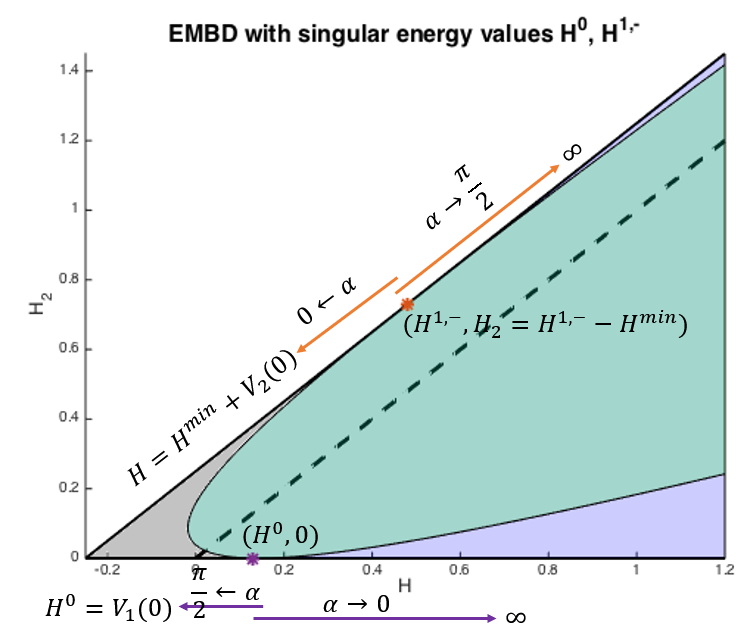}
\par\end{centering}

\protect\caption{\label{fig:alphaasymp}$H^0$ and $H^{1,-}$ behavior for angles close to perpendicular.}
\end{figure}

\subsection{Interior tangent segments}\label{sec:interiortangentsegments}

In lemma \ref{lem:prlwalltrans} we established that at a given wall point   \(q_w(q_2)\), tangency occurs for  \(H\geq H^{tmin}(q_2)\)  on the corresponding leaf of the level set    \((H-H_{2}^{t}(q_{2,}H),H_{2}^{t}(q_{2,}H))\) whereas for all other   \(H_2 \in \mathcal{  H}_{2}(q_2,H)\),   the impacts at    \(q_w(q_2)\) on the  leaf of  \((H-H_{2},H_2) \) are transverse. This lead to the definition of the the tangent \(H_{2}\) intervals \(H_{2}^{tan}(H)\) of lemma  \ref{lem:h2mh2tan} that divide the level sets to non-impacting, impacting and those with a tangency. By definition, the tangency at    \(q_w(q_2)\) may be external (as in convex billiards), and, if this is the case for all \(q_{2}\in L_{w}(H)\),  the tangent leaves include only external tangent points. Then,  the singular impact set of the iso-energy surface consists of  a finite collection of finite-length phase space segments and the impact-division is simpler than the general case in which the images and preimages of the singularities are in the allowed region of motion.

    Next we provide conditions under which the  leaf     \((H-H_{2}^{t}(q_{2,}H),H_{2}^{t}(q_{2,}H))\) includes a tangent segment to    \(q_w(q_2)\) which is in the billiard domain   (Theorem \ref{lem:h2tanofq2}).
Denote by
\( \mathcal{R}(q_2),\hat n(q_{2})\)  the radius of curvature and the inward normal of the wall at \(q_2\):\begin{equation}
 \mathcal{R}(q_2)= \frac{(1+(\epsilon_{w}Q'(q_{2}))^{2})^{3/2}}{\epsilon _{w}Q''(q_{2})},\qquad\hat n(q_{2})= \frac{(1,-\epsilon _{w}Q'(q_{2}))}{\sqrt{1+(\epsilon _{w}Q'(q_{2}))^{2}}}
\end{equation}and let \begin{equation}
H^{tsegm}(q_2)=\max\{0,-\frac{1}{2} \mathcal{R}(q_2)(\hat n\cdot \triangledown V)_{q_w(q_2)}\}=\begin{cases}0 & \epsilon _{w}Q''(q_{2})\cdot (\hat n\cdot \triangledown V)_{q_w(q_2)}>0 \\
0   & \epsilon _{w}Q''(q_{2})=0 \ \& \ (\hat n\cdot \triangledown V)_{q_w(q_2)}>0  \\
\infty\ & \epsilon _{w}Q''(q_{2})=0  \ \& \ (\hat n\cdot \triangledown V)_{q_w(q_2)}<0 \\
>0 & \epsilon _{w}Q''(q_{2})\cdot (\hat n\cdot \triangledown V)_{q_w(q_2)}<0 \\
\end{cases}\label{eq:H2tansegment}
\end{equation} \begin{defn}\label{def:nondegwallpoint} A wall point \({q_w(q_2)}\) is wall-non-degenerate if \(| (\hat n\cdot \triangledown V)_{q_w(q_2)}|+|\epsilon _{w}Q''(q_{2})|\neq0\).\end{defn}We show that the function \(H^{tsegm}(q_2)\) and the sign of \(\mathcal{R}(q_2)\) determine whether a tangency at    \(q_w(q_2)\)  for a given energy \(H\) is external or not: \begin{mainthm} \label{lem:h2tanofq2} Provided \(q_w(q_2)\) is non-degenerate, there exists a segment which is tangent  at  \(q_w(q_2)\)    and is inside the billiard domain on the iso-energy surface \(H\) iff \begin{equation}
\begin{cases}H\in (H^{tmin}(q_{2}),H^{tmin}(q_{2})+H^{tsegm}(q_2)) & \text{when }Q''(q_{2})\geqslant0, \\
H>H^{tmin}(q_{2})+H^{tsegm}(q_2)& \text{when  }Q''(q_{2})<0.
\end{cases}
\label{eq:hsegminbilliard}
\end{equation} At \(H=H^{tmin}(q_{2})\) the tangent segment is in the billiard domain if \( \left(\hat n\cdot \triangledown V\right)_{q_w(q_2)} <0\) and is outside the billiard domain if \( \left(\hat n\cdot \triangledown V\right)_{q_w(q_2)} >0\).  \end{mainthm}
\begin{proof}
By lemma \ref{lem:prlwall} for    \(H< H^{tmin}(q_{2})\) the wall point  \(q_w(q_2)\)   is not in the Hill region, so we need to consider      \(H\geqslant H^{tmin}(q_{2})\). Then, by lemma   \ref{lem:prlwalltrans}, a tangent point occurs at \(q_w(q_2)\)    on the level set   \((H-H_{2}^{t}(q_{2,}H),H_{2}^{t}(q_{2,}H))\), so,  by Eq. (\ref{eq:tanghh2}), at tangency \begin{equation}\label{eq:p2attangency}
p_{2}^{2}=2(H_{2}^{t}(q_{2,}H)-V_{2}(q_{2}))=2\frac{H-H^{tmin}(q_{2})}{1+(\epsilon _{w}Q'(q_{2}))^{2}}.
\end{equation}

 Starting at a tangent initial condition on the wall at  \(q_w(q_2)\), the flow moves the solution into the billiard domain iff   \(q_1(\Delta t)-\epsilon _{w}Q(q_{2}(\Delta t))>0\). The zero and first order terms in \(\Delta t\) vanish at tangency on the wall (since at tangency \(q_{1}=\epsilon _{w}Q(q_{2}),(p_1-\epsilon _{w}Q'(q_{2})p_{2})=0\)), and positivity of the \(\frac{1}{2}\Delta t^{2}\) terms becomes:\begin{equation}\label{eq:q12inbilliard}
-V'_{1}(\epsilon _{w}Q(q_{2}))+\epsilon _{w}Q'(q_{2})V'_{2}(q_{2})-\epsilon _{w}Q''(q_{2})p_{2}^2>0
\end{equation} see also \cite{RK2014smooth} where non-degenerate tangency is similarly defined for \(n \) d.o.f. HIS.
Thus, if \(Q''(q_{2})=0\),  the  tangent segment is in the billiard for all \(H\geqslant H^{tmin}(q_{2})\)  iff  \(\hat n(q_{2})\cdot \triangledown V|_{q_w(q_2)}<0\) (by the wall-non-degeneracy condition, otherwise,   \(\hat n(q_{2})\cdot \triangledown V|_{q_w(q_2)}>0\) and the segment is strictly outside the billiard). Then, \(H^{tsegm}(q_2)=\infty\), so this condition is in agreement with the first  line of  Eq. (\ref{eq:hsegminbilliard}) and with the condition at  \(H= H^{tmin}(q_{2})\).
Similarly, if  \(p_2=0\), namely if \(H=H^{tmin}(q_{2})\),  Eq. (\ref{eq:q12inbilliard}) becomes  \(-\hat n(q_{2})\cdot \triangledown V|_{q_w(q_2)}>0\), as stated.

Otherwise,  if \(Q''>0\) and  \(p_{2}\neq0\), rewriting Eq. (\ref{eq:q12inbilliard}) as:

\begin{equation}
0<p_{2}^2<\frac{ -V'_{1}(\epsilon _{w}Q(q_{2}))+\epsilon _{w}Q'(q_{2})V'_{2}(q_{2})}{\epsilon _{w}Q''(q_{2})}, \quad Q''>0,
\end{equation}and using  Eq. (\ref{eq:p2attangency}) these conditions become, for \(Q''>0\) \begin{equation}
H^{tmin}(q_{2})<H<H^{tmin}(q_{2})+\max\{0,\frac{1}{2}\frac{ -V'_{1}(\epsilon _{w}Q(q_{2}))+\epsilon _{w}Q'(q_{2})V'_{2}(q_{2})}{\epsilon _{w}Q''(q_{2})}(1+(\epsilon _{w}Q'(q_{2}))^{2})\}
\end{equation}so, by the definition of   \(H^{tsegm}(q_2)\)  Eq. (\ref{eq:hsegminbilliard}) is proved for \(H>H^{tmin}(q_{2})\) and \(Q''(q_{2})>0\).

If \(Q''(q_{2})<0\), a tangent segment is in the billiard domain iff \begin{equation}
p_{2}^2=2\frac{H-H^{tmin}(q_{2})}{1+(\epsilon _{w}Q'(q_{2}))^{2}}>\frac{-(-V'_{1}(\epsilon _{w}Q(q_{2}))+\epsilon _{w}Q'(q_{2})V'_{2}(q_{2}))}{-\epsilon _{w}Q''(q_{2})}.
\end{equation}
If the right hand side is negative  the above inequality is  satisfied for all real \(p_{2}\), namely whenever \(H\geq H^{tmin}(q_{2})\). If the right hand side is positive, the energy needs to be sufficiently large: \begin{equation}
H>H^{tmin}(q_{2})+\max\{0,\frac{1}{2}\frac{ -V'_{1}(\epsilon _{w}Q(q_{2}))+\epsilon _{w}Q'(q_{2})V'_{2}(q_{2})}{\epsilon _{w}Q''(q_{2})}
(1+(\epsilon _{w}Q'(q_{2}))^{2})\}\end{equation}
Thus Eq. (\ref{eq:hsegminbilliard}) is proved for \(Q''(q_{2})<0\).

\end{proof}

Geometrically,     Theorem \ref{lem:h2tanofq2}  may be interpreted as follows. At tangency the momentum component which is normal to the wall vanishes:    \( p_{\bot}=0\). From the equations of motion we conclude that  \(\dot p_{\bot}=-\left(\hat n\cdot \triangledown V\right)_{q_w(q_2)}=(-V'_{1}(\epsilon _{w}Q(q_{2}))+\epsilon _{w}Q'(q_{2})V'_{2}(q_{2}))/(1+(\epsilon _{w}Q'(q_{2}))^{2})^{1/2} \). Thus, the geometrical meaning of a positive normal force, \(-\left(\hat n\cdot \triangledown V\right)_{q_w(q_2)}\), is that it turns the tangent velocity vector into the billiard domain. If the  billiard is convex (\(Q''>0\)), the turning must be sufficiently sharp with respect to the parallel speed to remain in the domain whereas for a dispersing billiard the tangent segment remains in the domain even if there is no turning. A negative normal force means that the velocity vector turns away from the billiard domain, so the conclusions are  reversed: convex billiards have no tangent segments in the billiard domain for any energy whereas   locally dispersing   billiard boundaries have a tangent segment in the billiard domain for sufficiently  high energy where the parallel speed overcomes the turning. Finally, when \(p=0\) at the wall, positive normal force brings  the tangent segment into the billiard and negative one pushes it outside.

\noindent{\textbf{Interior tangent segments for the \(\alpha\)-slanted wall DC Example:}}
\textit{ In Theorem \ref{thm:tangslantdc} we establish that for the non-perpendicular cases, for sufficiently large \(H\), the majority of level sets in \(H_{2}^{tan}(H) \) that have a tangent point also have interior
tangent segments.}

\subsection{\label{sec:proofsmain2}Proofs of Theorems \ref{thm:nonimpactgen}-\ref{thm:tangslantdc} }

\noindent\textbf{Proof of Theorem \ref{thm:nonimpactgen}: }\textit{For sufficiently small \(\epsilon_{r}\) and arbitrary \(\epsilon_{w}\), the phase space measure of the iso-energy non-impact set on regular energy surfaces  is \(O(\sqrt{\epsilon_{r}})\) close to the measure of the non-impact zone defined by the IFG and  IEMBD of the GWS  at  \(\epsilon_{r}=0\).  For interior GWS this set has a positive \(O(1)\) measure for a range of energies, whereas for exterior GWS the measure of this set is at most  of \(O(\sqrt{\epsilon_{r}})\).}
\begin{proof}
 At \(\epsilon_{r}=0\), if the non-impact zone of the IFG has positive measure, the iso-energy non-impact set of the GWS  includes a connected part - a positive measure set of families of non-impacting Liouville leaves which belong to the non-impact zone of the IFG. These are the leaves with PRLs that do not intersect or touch the wall, and, with the exception of a few singular leaves, they correspond to invariant tori of the smooth integrable system. In particular, since the energy is regular, the boundaries of the non-impact zone correspond to regular invariant tori. The non-impact set can possibly include also a disconnected part - a measure zero set of resonant orbits or singular orbits belonging to isolated resonant cut leaves which also contain tangent and impacting segments.

Each energy surface with an  \(\epsilon_{r}=0\) non-impact zone of measure \(m_{0}(H)>0\) produces,  for sufficiently small $\epsilon_r$,  by standard KAM results for  smooth systems  (recall  the S3BN assumption and the smoothness assumption on \(V_{r}\)), an invariant set bounded by KAM tori which are bounded away from impacts. By KAM, these boundary tori are, at worst case (i.e. if the non-impact zone boundary is resonant),  \(O(\sqrt{\epsilon_{r}})\) close to the unperturbed tori, so the measure of this invariant non-impact set, which includes also internal non-impact resonances, is, at worse case,  \(m_{\epsilon_{r}}(H)=m_{0}(H)+O(\sqrt{\epsilon_{r}})\).

 The   \(\epsilon_{r}=0\) measure zero set of initial conditions that belong to isolated tangent resonant leaves, when exist, may produce resonant non-impacting islands of smooth periodic and quasiperiodic motion. Their measure, as they belong to the smooth  resonant set is at most of  \(O(\sqrt{\epsilon_{r}})\),  thus, we conclude that the measure of the iso-energy non-impact set is \(m_{0}(H)+O(\sqrt{\epsilon_{r}})\).

 For interior GWS, the PRLs that are centered around the minimizers of \(V\) which are inside the billiard domain with partial energies close to the local minima values belong to the non-impact zone, so  \(m_{0}(H)>0\) for \(H\) values close to the corresponding minima values of \(V\) (which, by assumption [N] of the S3BN condition, are non-degenerate). Exterior GWS do not have  a non-impact zone at \(\epsilon_{r}=0\) by Claim 1 of Proposition  \ref{thm:global}, namely for exterior GWS  \(m_{0}(H)=0\) for all \(H\). Thus, if the  non-impact  set is non-empty it consists only of smooth resonance islands of measure at most  \(O(\sqrt{\epsilon_{r}})\) and it does not contain any non-impacting KAM tori.
\end{proof}

\noindent\textbf{Proof of Theorem  \ref{thm:singulconvex}:}\textit{ A  GWS in general position, for which the billiard is convex   and the potential \(V\) increases along the wall normal has, at  \(\epsilon_{r}=0\),  only external tangent points. If the billiard is strictly convex, and there exists a finite \(K\) such that the potential \(V\) increases along the wall normal for all \(|q_{2}|\geqslant K\),  then,  for sufficiently large energy, the same statement holds. }
\begin{proof}

 If \(\epsilon _{w}Q''(q_{2})\geqslant0\) and \(\left(\hat n\cdot \triangledown V\right)_{q_{w}(q_2)}>0\) for all    \(q_{2}\in\mathbb{R}\), then,  by Eq. (\ref{eq:H2tansegment}),\  \(H^{tsegm}(q_2)\equiv0\) and by definition \ref{def:nondegwallpoint} all the wall points are wall-non-degenerate. Hence, by Theorem \ref{lem:h2tanofq2} (see Eq. (\ref{eq:hsegminbilliard})),  there are no energies for which the tangent segment is in the billiard domain.

If   \(\epsilon _{w}Q''(q_{2})>c>0\) for all \(q_{2}\)  and, for \(|q_{2}|\geqslant K\),  \(\left(\hat n\cdot \triangledown V\right)_{q_{w}(q_2)}>0\), then \  \(H^{tsegm}(q_2)=0\) for  \(|q_{2}|\geqslant K\).  For    \(|q_{2}|\leqslant K\), by Eq. (\ref{eq:H2tansegment}) and the S3BN assumption, for all allowed \(V\) and \(\epsilon _{w}Q\) there exists a finite \(K_{1}>0\)  such that \(0\leqslant H^{tsegm}(q_2)<\frac{K_{1}}{c}\) . Let \(K_{2}=\max_{|q_{2}|\leqslant K} H^{tmin}(q_2)\), so \(K_2\) is finite by the S3BN assumption. Then, as all wall points are wall-non-degenerate, by  Theorem \ref{lem:h2tanofq2}, for all  \(H> K_2+\frac{K_{1}}{c}\), there are no tangent segments in the billiard at \(q_{w}(q_2)\) for all \(q_{2}\) (for    \(|q_{2}|\leqslant K\) as the energy is sufficiently high and for     \(|q_{2}|\geqslant K\)  because  \  \(H^{tsegm}(q_2)=0\)).

. \end{proof}

\noindent\textbf{Proof of Theorem  \ref{thm:singulconcave}:}\textit{ A GWS in general position, defined on a semi-dispersing billiard with a potential \(V\) which decreases along the wall normal, has, at  \(\epsilon_{r}=0\), on any energy surface, tangent segments inside the billiard at all  wall points which are in the Hill region. If the billiard is dispersing and there exists a finite \(K\) such that the potential \(V\) decreases along the wall normal for all \(|q_{2}|\geqslant K\),  then,  for sufficiently large energy, the same statement holds. }
\begin{proof}
  If \(\epsilon _{w}Q''(q_{2})\leqslant0\) and \(\left(\hat n\cdot \triangledown V\right)_{q_{w}(q_2)}<0\) for all    \(q_{2}\in\mathbb{R}\), then,  by Eq. (\ref{eq:H2tansegment}),\  \(H^{tsegm}(q_2)\equiv0\) and by definition \ref{def:nondegwallpoint} all the wall points are wall-non-degenerate. Hence, by Theorem \ref{lem:h2tanofq2} (see Eq. (\ref{eq:hsegminbilliard})),  for all \(H\geqslant H^{tmin}(q_2)\), the iso-energy tangent leaf includes a  \(q_{w}(q_2)\)-tangent segment  in the billiard domain. Recall that by lemma \ref{lem:prlwall} \(q_{w}(q_2)\in S_{w}(H)\) iff \( H^{tmin}(q_2)< H\), so we see that all the wall points which are in the Hill region have a tangent segment in the billiard, and this tangent segment belongs to the tangent zone.

If   \(\epsilon _{w}Q''(q_{2})<-c<0\) for all \(q_{2}\)  and, for \(|q_{2}|\geqslant K\),  \(\left(\hat n\cdot \triangledown V\right)_{q_{w}(q_2)}<0\), then \  \(H^{tsegm}(q_2)=0\) for  \(|q_{2}|\geqslant K\), and, as in the proof of   \ref{thm:singulconvex}, there exist finite positive \(K_{1,2}\) such that   for    \(|q_{2}|\leqslant K\),  \( H^{tmin}(q_{2})+H^{tsegm}(q_2)<K_{2}+\frac{K_{1}}{c}\) . Then, as all wall points are wall-non-degenerate, by  Theorem \ref{lem:h2tanofq2}, for all  \(H> K_2+\frac{K_{1}}{c}\), there are\ tangent segments in the billiard at \(q_{w}(q_2)\) for all \(q_{2}\) (for    \(|q_{2}|\leqslant K\) as the energy is sufficiently high and for     \(|q_{2}|\geqslant K\)  because  \  \(H^{tsegm}(q_2)=0\)).
     \end{proof}
\noindent\textbf{Proof of Theorem \ref{thm:concconv}:} \textit{
 Above a critical energy, a GWS in general position which has both concave and convex wall segments has a non-trivial singular impact set.  }

\begin{proof}    Given that there exists a  \(q_{2}\) such that the wall is concave at  \(q_{w}(q_2)\), namely     \(\epsilon _{w}Q''(q_{2})<-c<0\), we conclude that this point is wall-non-degenerate and that \(H^{tsegm}(q_2)\) is finite (possibly zero). Hence, for all \(H>H^{tmin}(q_{2})+H^{tsegm}(q_2)\) the tangent segment to    \(q_{w}(q_2)\) is in the billiard so the singular impact set includes the forward and backward impact flow of this segment. Moreover, if \(H^{tsegm}(q_2)>0\), for \(H\in(H^{tmin}(q_{2}),H^{tmin}(q_{2})+H^{tsegm}(q_2))\) then the tangent segment to    \(q_{w}(q_2)\) is not in the billiard even though     \(q_{w}(q_2)\) is in the Hill region, namely, the singular impact set changes non-trivially with energy even at interior points of \(S_{w}(H)\).   \end{proof}

\noindent\textbf{Proof of Theorem \ref{thm:tangslantdc}:} \textit{
 For sufficiently large energy, the \(\alpha-\)slanted wall Duffing-Center system with \(\alpha\in(0,\frac{\pi }{2})\) has a non-trivial singular impact set; For such energies, the tangent segments occur  on a \(\sin^2(\alpha)(1+O(\frac{1}{\sqrt{H}}))\) portion of the leaves.  }

\begin{proof}     By item 6 of Proposition \ref{prop:globalslanted}, the tangent zone portion for sufficiently large \( H\), namely for \(H>H^{asym}=H^{asym}(q_{1s},q_{2c},\omega^{2},\alpha )\), is \(\sin^2(\alpha)(1+O(\frac{1}{\sqrt{H}}))\). Recall that the tangent zone is the union of level sets on which a tangency at the wall at some point    \(q_{w}(q_2)\in S_{w}(H)\)  is detected. Since the slanted wall is flat \((Q''(q_{2})=0)\), by Theorem \ref{lem:h2tanofq2}, the tangent segments to the wall at    \(q_{w}(q_2)\) are in the billiard domain for all \(H\geqslant H^{tmin}(q_{2})\) if   \(0<-V'_{1}(\epsilon _{w}Q(q_{2}))+\epsilon _{w}Q'(q_{2})V'_{2}(q_{2})=-
(\cot \alpha \cdot q_2-q_{1s})^{3}+\cot \alpha \cdot (1+\omega^{2})q_2-\cot \alpha \cdot(q_{1s}+\omega^{2}q_{2c})\), and are out of the billiard domain if the strict inequality is reversed. Since this is a cubic polynomial with negative coefficient on the \((q_2)^{3}\) term, there exist finite \(q_2^{tsegm,\pm}=q_2^{tsegm,\pm}(q_{1s},q_{2c},\omega^{2}, \alpha )\) such that  for all  \( q_2< q_2^{tsegm,-}\) the inequality is satisfied and for all  \( q_2> q_2^{tsegm,+}(q_{1s},q_{2c},\omega^{2},\alpha )\) the inequality is reversed.  Since, by Eq. (\ref{eq:tanghh2slant}), \(H_{2}^{t}(q^{tsegm,-}_2,H)=\sin^2(\alpha)H+G(q^{tsegm,-}_2)\), and \(H_{2}^{t}(q^{a}_{2}(H),H)=V_{2}(q^{a}_{2}(H))=\omega^{2}\tan^{2}\alpha  \cdot H^{1/2}(1+O(H^{-1/4}))\), for all \(H>\max\{\ H^{tmin}(q_2^{tsegm,-}), H^{tmin}(q_2^{tsegm+}),H^{asym}\}\) the    interior-tangent-segment zone includes the set \([\omega^{2}\tan^{2}\alpha  \cdot H^{1/2}(1+O(H^{-1/4})),\sin^2(\alpha)H+G(q^{tsegm,-}_2)]\), which has the same lower boundary as  \(H_2^{tan}(H)\) and its upper boundary differs from \(H_2^{tan}(H)\) by the \(O(1)\) term \( (G_{max}-G(q^{tsegm,-}_2))\).   For such large energies the tangencies associated with wall points at large positive \(q_2\) are outside the billiard whereas those associated with large negative \(q_2\) are in the billiard. Yet, by the symmetry of \(V_{2}\) the level sets for large positive and large negative \(q_{2}\) essentially coincide, so the asymptotic results for large \(H\) show that the majority of level sets in the tangency zone indeed include a tangent segments  in the billiard, making the tangent singularity set non-trivial.  \end{proof}

\section{Discussion}\label{sec:discussion}

 We analyze a class of integrable Hamiltonian systems undergoing impacts from a wall  by constructing their Impact Energy-Momentum Bifurcation Diagram (IEMBD) and  Impact Fomenko Graphs (IFG) and by examining the intersections of the wall with their projected rectangles of the Liouville leaves  to the configuration space.  These  bring new analysis tools  to the field of non-smooth impact systems, in which the dynamics, in contrast with billiards, depend non-trivially on the energy. We study this class of HIS in the integrable,  near-integrable and far from integrable cases. \ First, in section \ref{sec:integrable}, we construct the IEMBD and the IFG for  systems with a wall which is perpendicular to one of the axes. We prove that such systems are Liouville integrable and that the IEMBD and IFG are needed to describe their level set topology and impact division (Theorem \ref{thm:integrability}). We then apply these tools together with the  impact-KAM theory of   \cite{pnueli2018near} to describe  the impact division of  nearby systems - systems with a slight deformation of the perpendicular wall and with a small smooth coupling term (Theorem \ref{thm:nearintegrability}). In section \ref{sec:Hill-region} we study HIS with a general wall which is not necessarily near perpendicular. To this aim we construct and study the structure of the projected rectangles of the Liouville leaves  to the configuration space. Proposition \ref{thm:global} and Theorem \ref{thm:hillregionprl} summarize the main conclusions we obtain from these constructions regarding the non-impact, tangent and transverse impact zones of the IEMBD for general systems and Proposition \ref{prop:globalslanted} summarizes these for some cases of the Duffing-center potential  with a slanted wall. Utilizing Theorem \ref{lem:h2tanofq2},    we establish that in some cases, like walls with focusing boundary at sufficiently high energy, the tangent impact set is trivial (Theorem \ref{thm:singulconvex}). In other cases, like for billiards with dispersing wall, we establish that the singular impact set is non-trivial (Theorem \ref{thm:singulconcave}-\ref{thm:concconv}) and establish that at high energy the tangent segments that are in the billiard are those that are tangent at dispersing segments of the wall.

\

\noindent\textbf{Future directions and open problems: }

\

\noindent\textbf{Classification of Liouville integrable HIS:} We introduced a class of  LIHIS and demonstrated that their IEMBD and IFG provide, as in the smooth case, global information with regards to the level set structure and the qualitative behavior of such systems under perturbations. The  classification and development of a Fomenko-Zeichang Theory  \cite{fomenko2004integrable} for LIHIS is a natural exciting direction to take. To this aim, definitions of new impact atoms may be needed  (e.g. a B atom which is cut in the middle, see Fig. \ref{fig:ps-alphapi2}d).  Studying the leaf structure of other examples of LIHIS, such as other separable systems in domains with several infinite walls,  radially-symmetric HIS in  circular billiards (Appendix 8 in \cite{Pnueli2020}), integrable HIS in an ellipse  \cite{Fedorov2001,Radnovic2015,Dragovic2014a} and of all of these systems in higher dimensions,  may contribute to the development of such theory.

\noindent\textbf{Classification of  integrable HIS which are not LIHIS:} In  \cite{Issi2019} the appearance of IHIS which are not LIHIS was demonstrated and in section \ref{sec:multiwall} the construction of IEMBD for such cases was proposed (see Figure \ref{fig:IEMBDcorner}b). Classification of such systems, including the construction of IFG with directional motion on leaves which are compact, oriented surfaces of genera higher than one, is another exciting and non-tivial future direction. Recent works on the Fomenko-graphs of quasi-integrable billiards may be relevant for the classification of such systems \cite{fomenko2019singularities,fomenko2019topological,Moskvin2018}.  \\
\noindent\textbf{Transverse impact sets:}  To establish the existence of an invariant set within the transverse impact set which is bounded away from the tangent set, the wall parameter (\(\epsilon_{w}\) of Eq. (\ref{eq:wavywal}))   may serve as a continuation parameter as its connects between the near integrable regime and the non-integrable regime. Thus, we expect that for sufficiently small \(\epsilon_{r}\)  there exist a finite \(\epsilon_{w}(\epsilon_{r})>0\)    at which the last dividing torus between the transverse and non-transverse regimes is still intact. Such a torus, by definition, must lie within the transverse impact set, thus, the measure of the transverse impact leaves provides an upper bound to  the union of such invariant components of the transverse impact set. It would be interesting to study numerically the existence and properties of such last invariant transverse curve.

\noindent\textbf{Applications:}   HIS supply a modeling framework for various physical systems, for example, for classical approximation of chemical reactions in which impacts model the  strong atomic repulsion forces and the smooth potentials model the smooth attraction forces and external fields \cite{LRK12,RK2014smooth}. Global analysis  of such systems by the IEMBD and the IFG provides an opportunity to obtain qualitative information on the type of solution that occur in  such complex systems.  Systems arising in Chemistry and Physics where the dynamics are governed by smooth near integrable  fields and sharp repulsion from  surfaces of various shapes can be now studied.

\section*{Acknowledgment}
We acknowledge support by ISF  Grant 1208/16. Additionally, some of this material is based upon work supported by the National Science Foundation under Grant No. 1440140, while the author was in residence at the Mathematical Sciences Research Institute in Berkeley, California, during the Fall semester of 2018.

\newpage
\section*{List of abbreviations}

\paragraph*{d.o.f} degrees-of-freedom

\paragraph*{EMBD} Energy-Momentum Bifurcation Diagram (see, e.g. \cite{Arnold2007CelestialMechanics,lerman1998integrable,fomenko2004integrable})

\paragraph*{FG}  Fomenko graphs (see, e.g. \cite{fomenko2004integrable})

\paragraph*{HIS} Hamiltonian impact systems

\paragraph*{LIHIS} Liouville-integrable HIS (definition \ref{def:liouvilimpint})
\paragraph*{NIHIS} Nearly integrable HIS
(definition  \ref{def:nihis})

\paragraph*{IEMBD} Impact-EMBD
(definition \ref{def:iembd})

\paragraph*{IFG} Impact Fomenko graphs
(definition \ref{def:ifg})

\paragraph*{PRL} projected rectangle of leaves
(definition \ref{def:PRL})

\paragraph*{S3BN class} Separable, Smooth, Simple, Bounded Non-degenerate potentials (definition \ref{def:s3b})

\paragraph*{GWS} General Wall System (definition \ref{def:gws})
\paragraph*{PWS}  perpendicular wall system - a GWS with a perpendicular wall (definition \ref{def:pws})\paragraph*{DC} Duffing-center
Hamiltonian (Eq. (\ref{eq:doublewell}))\newpage
\appendix
\section*{Appendix}

\section{Trajectory types  for the DC Hamiltonian impacting from perpendicular walls}\label{appdx:classification}
\subsection{Horizontal wall ($\alpha=0$)}
Consider the impact system with the Duffing-center potential (\ref{eq:doublewell}) when the wall is horizontal. Then,  the impact dynamics may be projected  to the $(q_2,p_2)$ phase space (see Fig \ref{fig:ps-alpha0}). Thus, all  types of trajectories, depending on the value of $I$ and on the parameter $q_{2c}$, may be classified as either non-impacting, tangent, or impacting (upon reflection the motion continues on the same \(H_{2}\) level set after the discontinuous jump $p_{2}\rightarrow-p_{2}$).

\noindent For $q_{2c}<0$:
\begin{itemize}
\item Motion on level sets with $I\leq\frac{V_{2}(0)}{\omega}$ remains
unchanged
\item All level sets with $I>\frac{V_{2}(0)}{\omega}$ achieve transversal impact.
\end{itemize}
For $q_{2c}>0$:
\begin{itemize}
\item The allowed region of motion (as defined by the billiard boundary)
is defined by $I>\frac{V_{2}(0)}{\omega}$. In this region all trajectories
achieve transversal impact.
\item The value $I=\frac{V_{2}(0)}{\omega}$ corresponds to a single tangent point on the wall, which cannot be reached by any trajectory inside the billiard region.
\end{itemize}
For $q_{2c}=0$ (singular case):
\begin{itemize}
\item All regular trajectories corresponding to $I>0$ impact.
\item The value $I=0$ corresponds to a single tangent point on the wall which is a stable fixed point.
\end{itemize}

\subsection{Vertical wall ($\alpha=\frac{\pi}{2}$).}
Consider the DC  system when the wall is vertical.  Here  impacting trajectories, upon reflection,  continue on the same \(H_{1}\) level set with  $p_{1}\rightarrow-p_{1}$.
Due to the richer phase space structure (see Fig. \ref{fig:ps-alphapi2}), there are more sub-cases to consider: \\

\noindent
For $q_{1s}>\sqrt{2}$ (the separatrix
region is inside the allowed region of motion, Figure \ref{fig:ps-alphapi2}a):
\begin{itemize}
\item Motion on level sets such that $H_{1}\leq V_{1}(0)$ remains unchanged
\item All level sets with $H_{1}>V_{1}(0)$ achieve transversal impact.
\end{itemize}

For $1\leq q_{1s}<\sqrt{2} $
(wall intersects the left separatrix loop, Fig \ref{fig:ps-alphapi2}b):
\begin{itemize}
\item Motion on level sets such that $H_{1}\leq V_{1}(0)$ remains unchanged
\item All level sets with $H_{1}>0$, achieve transversal impact
\item For level sets with $V_{1}(0)<H_{1}<0$, trajectories in the right
node of the separatrix remain unchanged, whereas in the left node
of the separatrix these trajectories achieve transversal impact
\end{itemize}

For $0<q_{1s}<1$:
\begin{itemize}
\item Motion on level sets such that $H_{1}\leq V_{1}(0)$ remains unchanged
\item Motion in the right node of the separatrix remains unchanged
\item All level sets with $H_{1}>V_{1}(0)$ which are not in the right node
of the separatrix achieve transversal impact.
In particular, all trajectories in the left node of the separatrix
reflect
\end{itemize}

For $-1\leq q_{1s}<0$ (only the right
elliptic fixed point is inside the allowed region of motion):
\begin{itemize}
\item Motion on level sets such that $H_{1}\leq V_{1}(0)$ remains unchanged
\item All level sets with $H_{1}>V_{1}(0)$ inside the billiard domain achieve transversal impact
\end{itemize}

For $q_{1s}<-1$ (all the fixed points
are outside the region of allowed motion):
\begin{itemize}
\item All trajectories in the allowed region of motion achieve transversal impact with
the wall
\item The value $H_1=V_1(0)$ corresponds to a single tangent point on the wall which cannot be reached by any trajectory inside the billiard domain.
\end{itemize}

For $q_{1s}=\sqrt{2}$ or $q_{1s}=0$
(wall coincides with the leftmost point on the separatrix or the saddle
fixed point respectively - see Figure \ref{fig:ps-alphapi2}c,d):
\begin{itemize}
\item Motion on level sets such that $H_{1}\leq0$ remains unchanged
\item Level sets with $H_{1}>0$ achieve transversal impact.
\end{itemize}


\end{document}